\documentclass[onecolumn, 11pt]{article}

\pdfoutput=1

\usepackage[utf8]{inputenc}

\usepackage{longtable}
\usepackage[left=2cm,right=2cm,top=3cm,bottom=2.5cm]{geometry}
\usepackage{media9} 
\usepackage{fancyhdr}
\usepackage{setspace}
\usepackage{import}
\usepackage[hang, small,labelfont=bf,up,textfont=it,up]{caption}
\usepackage{subcaption}

\usepackage{sectsty}
\sectionfont{\centering}
\usepackage{float}

\usepackage{blkarray}
\usepackage{graphicx}
\usepackage{mwe}

\newcommand{\motif}[2]{\vcenter{\hbox{\includegraphics[scale=0.3]{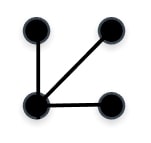}}}}

\usepackage{todonotes}
\usepackage[english]{babel}										
\usepackage[protrusion=true,expansion=true]{microtype}
\usepackage{gensymb}
\usepackage{booktabs}
\usepackage{mathptmx}
\usepackage{siunitx}
\usepackage{subcaption}
\usepackage{nomencl}
\usepackage{multirow}
\usepackage{framed}
\usepackage{multicol}
\usepackage{url}
\usepackage{amsmath, amsfonts,amsthm,amssymb}

\usepackage{pdflscape}

\usepackage{color}
\makeatletter
\def\BState{\State\hskip-\ALG@thistlm}
\makeatother

\usepackage{verbatim}

\usepackage{tablefootnote}

\theoremstyle{plain}
\newtheorem{theorem}{Theorem}[section]

\theoremstyle{definition}
\newtheorem{example}[theorem]{Example}
\theoremstyle{definition}
\newtheorem{remark}[theorem]{Remark}

\usepackage[explicit]{titlesec}
\titleformat{name=\paragraph,numberless}[runin]
  {\normalfont\normalsize\bfseries}{}{15pt}{\uline{#1.}}
  
\usepackage[toc,page]{appendix}

\usepackage{algorithm}
\usepackage{algpseudocode}

\usepackage{footnote}
\makesavenoteenv{tabular}
\makesavenoteenv{table}

\usepackage{threeparttable}

\usepackage{siunitx}
\usepackage{dcolumn}
\newcolumntype{d}[1]{D{.}{.}{#1}}
\usepackage{etoolbox}
\robustify\bfseries
\usepackage{chngcntr}

\usepackage{array}

\usepackage[numbers]{natbib}

\newcommand{\tr}{\operatorname{tr}}

\newcommand{\var}{\operatorname{var}}
\newcommand{\dto}{\stackrel{d}{\longrightarrow}}
\newcommand{\rN}{\ensuremath{\mathrm{N}}}

\usepackage{enumitem}

\usepackage{mathtools}

\usepackage{sectsty}
\sectionfont{\bfseries\Large\raggedright}

\usepackage{tikz}
\newcommand\hlight[1]{\tikz[overlay, remember picture,baseline=-\the\dimexpr\fontdimen22\textfont2\relax]\node[rectangle,fill=blue!50,rounded corners,fill opacity = 0.2,draw,thick,text opacity =1] {$#1$};}

\begin{document}

\singlespace

\title{Subgraphs and Motifs in a Dynamic Airline Network}

\author{
{Marius Agasse-Duval and Steve Lawford$^\dagger$}\\[2pt]
Data, Economics and Interactive Visualization (DEVI) group, ENAC (University of Toulouse),\\7 avenue Edouard Belin, CS 54005, 31055, Toulouse, Cedex 4, France\\
$^\dagger$Corresponding author. Email: steve.lawford@enac.fr
}

\date{}

\maketitle

\vspace{0.5cm}

\begin{abstract}

    \noindent How does the small-scale topological structure of an airline network behave as the network evolves? To address this question, we study the dynamic properties of small undirected subgraphs using 15 years of data on Southwest Airlines' domestic route service. We use exact enumeration formulae to identify statistically over- and under-represented subgraphs, known as motifs and anti-motifs. We discover substantial topology transitions in Southwest's network and provide evidence for time-varying power-law scaling between subgraph counts and the number of edges in the network. We also suggest a node-ranking measure that can identify important nodes relative to specific local topologies. Our results extend the toolkit of subgraph-based methods and provide new insight into transportation networks and the strategic behaviour of firms.
    
\end{abstract}

\section{Introduction}
A network \emph{motif} is a connected subgraph, usually with a small number of nodes, that occurs significantly more often in a real-world network than it does in an ensemble of appropriately-chosen random graphs. Motifs were first introduced by \cite{milo_etal02}, who applied them to biochemical gene regulation networks, ecosystem food webs, neuronal connectivity networks, sequential logic electronic circuits, and a network of hyperlinks from the World Wide Web.\footnote{An early study by some of these authors presented a specific application of motifs to genetics \cite{shen-orr_etal02}.} They found evidence that distinct sets of motifs are associated with different types of network, and suggest that motifs are basic structural elements, or topological interaction patterns, each of which may perform precise specialized functions, and that can be used to define universal network classes (e.g., evolutionary, information processing, etc.). Their paper was rapidly followed by many subsequent studies that looked for motifs in biological data, in particular in gene regulation and neuroanatomical networks e.g., \cite{alon07, dobrin_etal04, prill_etal05, sporns_kotter04, yeger-lotem_etal04} and more recently \cite{chen_etal13, wu_etal13}. However, the presence and interpretation of motifs in economic or transportation networks has received very little attention.

Graph-theoretic research on transportation networks typically focuses on macroscopic features such as network diameter, or on microscopic measures that include various node centralities to identify important nodes. In this paper, we examine ``mesoscopic'' subgraph-based measures that fall between these global and local extremes. We describe the scaling behaviour of small and possibly overlapping subgraphs on up to five nodes, and identify motif dynamics in a transportation network, using 15 years of data on the U.S. domestic airport--airport route network of one of the world's largest passenger carriers, Southwest Airlines.

The network that we study has several notable features. First, it is quite small, with no more than 88 nodes and 522 edges, in 2013Q4 (Section \ref{sec:network_data}). By contrast, many real-world networks are extremely large. For example, Facebook and Twitter reported 2.01 billion and 328 million monthly active users, respectively, in the second quarter of 2017 \cite{facebook17, twitter17}. The academic search engine Google Scholar covered an estimated 160 million indexed documents in 2014 \cite{orduna-malea_etal15}. The Stanford Large Network Dataset Collection \cite{snapnets} lists at least 120 large technological, social, communication and other graphs (or subgraphs) with thousands to millions of nodes and edges. The small size of our network enables us to apply very accurate exact methods for subgraph and motif analysis, and it is not necessary to use the fastest available approximate sampling techniques. Second, airports and routes describe the topology of a human-made technological network, and its evolution over time will closely reflect a carrier's strategic, economic and operational decisions, as well as other regulatory and spatial (geographical) constraints. For this reason, we would expect the interpretation of subgraph-based measures, and the dynamic behaviour of motifs, to be rather different to that which is observed for naturally-evolving biological networks, or for social networks on (say) collaborations between scientists or informal links between company executives.

We make the following specific contributions:

\begin{itemize}
    \item We consider small (three, four, and five-node) undirected subgraphs (Section \ref{sec:graphs_and_motifs}). We use exact enumeration formulae to count subgraphs, and identify motifs (and anti-motifs) with reference to two null random networks, chosen to have some of the same characteristics as the real network, namely the Erd{\H{o}}s-R{\'{e}}nyi random graph $G(n, p)$, and a rewiring model closely inspired by \cite{milo_etal02} (Section \ref{sec:identify_motifs}). There are few available studies of motifs in economic and transportation networks and none, to our knowledge, that use exact methods. We provide new evidence of motifs and anti-motifs in an airline network and show that their significance varies over time.
   
    \item We investigate scaling in subgraph counts (Section \ref{sec:scaling_properties}). We find that the number of subgraphs of a given type generally increases over time, as the size of the network grows. While this is not surprising, we also identify a possible power-law scaling regularity between subgraph counts and the number of edges $m$ in the network, of the form $y = A \, m^{\beta}$. This scaling is stable across a wide range of number of edges but appears to undergo a dramatic transition during the sample period, corresponding to several scaling regimes. We present evidence that the power-law exponent $\beta$ in each regime is related to the number of nodes $b$ in each subgraph, either as $\beta \approx b-1$ (regime 1) or $\beta \approx b/2$ (regime 2). We draw comparisons with the implied scaling properties of an Erd{\H{o}}s-R{\'{e}}nyi random graph and several deterministic graphs, and use a toy regime-switching model to provide insight into changes in scaling behaviour.
    
    \item We mention two further applications of subgraphs to the descriptive analysis of networks (Section \ref{sec:subgraph_centrality} and Appendix \ref{sec:results_spatial}). First, we propose a new subgraph-based method for ranking nodes, using the number of a particular type of subgraph in a network that are incident to a given node, and show that it can give different rankings to standard measures such as degree or betweenness centrality. We compare our results with the ``subgraph-centrality'' of \cite{estrada_rodriguez-velazquez05}, which we find to be highly correlated with degree (and other standard) centrality measures on our data. In Appendix \ref{sec:results_spatial}, we briefly examine the spatial distribution of the triangle subgraph, and show that it suggests a clear geographical shift in the concentration of network activity over time.
    
\end{itemize}
Our work is based on a large literature in graph theory and complex systems, and we discuss related research in Section \ref{sec:related}. In Section \ref{sec:conclusions}, we conclude and present directions for future work. All proofs, some supporting figures and tables, and discussions of spatial subgraphs and bounds on the number of complete subgraphs, are collected in Appendices \ref{sec:proofs}--\ref{sec:bounds_on_complete}.

\clearpage

\section{Subgraphs and motifs}\label{sec:graphs_and_motifs}
We begin with an overview of the relevant tools of graph theory that we will use in this paper. Major monographs on the subject include mathematical aspects \cite{diestel17}, applications to social networks and economics \cite{jackson08}, and algorithms \cite{jungnickel08}. Algorithms for graph search, shortest path length, and maximum flow, are also covered in detail by \cite[Section 6]{cormen_etal09}. The comprehensive survey by \cite{newman03} provides a complex systems perspective.
\subsection{Preliminaries}
A \emph{graph} is an ordered pair $G = (V(G), E(G))$, where $V(G)$ is a set of \emph{nodes} and $E(G)$ is a set of \emph{edges} $E(G) \subseteq V(G) \times V(G)$. When there is no ambiguity, we write $V = V(G)$ and $E = E(G)$. The number of nodes and edges are denoted by $n = |V|$ and $m = |E|$ respectively. We refer to a graph by its $n \times n$ \emph{adjacency matrix} $g$ with representative element $(g)_{ij}$. In this paper, we consider \emph{simple} (no self-links or multiple edges) undirected and unweighted graphs, so that $(g)_{ii} = 0$ (no self-links), $(g)_{ij} = (g)_{ji}$ (undirected) and $(g)_{ij} \in \{0, 1\}$ (unweighted, no multiple edges). We use $(i, j) \in E$ to denote an edge between nodes $i$ and $j$, and say that they are \emph{directly-connected}. A \emph{walk} between nodes $i$ and $j$ is a sequence of edges $\{(i_{r}, i_{r+1})\}_{r=1,\ldots,R}$ such that $i_{1} = i$ and $i_{R+1} = j$. A \emph{path} is a walk containing distinct nodes. A graph is \emph{connected} if there is a path between any pair of nodes $i$ and $j$. A \emph{cycle} (or a \emph{simple cycle}) is a walk (or path) that starts and ends at the same node. We assume that every theoretical network used in this paper is connected and, furthermore, all of our empirical networks are also connected.

The \emph{diameter} (or \emph{average path length}) is the maximum (or mean) shortest path length across all pairs of nodes in a graph. The \emph{degree} $k_i = \sum_{j}(g)_{ij}$ is the number of nodes that are directly-connected to node $i$, and the \emph{degree distribution} $P(k)$ is the probability distribution of $k$ over $G$.\footnote{Unless otherwise stated, all summations are computed over the full range of permitted values of the index of summation.} In a $k$-\emph{regular} graph, every node has degree $k$. The \emph{(1-degree) neighbourhood} of node $i$ in $G$ is denoted $\Gamma_{G}(i) = \{j: (i, j) \in E\}$, and is the set of all nodes that are directly-connected to $i$, so that $k_{i} = |\Gamma_{G}(i)|$. The \emph{density} $d(G) = 2 m / n (n-1)$ is the number of edges in $G$ relative to the maximum possible number of edges in a graph with $n$ nodes: it ranges from 0 (a set of isolated nodes) to 1 (an $n$-\emph{complete} graph $K_{n}$, which has all possible edges).

A graph \emph{isomorphism} from a simple graph $G$ to a simple graph $H$ is a bijective mapping $f:V(G) \to V(H)$ such that $(i, j) \in E(G)$ if and only if $(f(i), f(j)) \in E(H)$. We use $G \cong H$ to denote that $G$ and $H$ are isomorphic. A graph \emph{automorphism} is an isomorphism of a graph with itself.\footnote{Isomorphic graphs on the same set of nodes have the same topology but will generally have different adjacency matrices.} Let $G' = (V', E') \subseteq G$ denote a subgraph of $G$, such that $V' \subseteq V$ and $E' \subseteq E$. If $G' \subseteq G$ and $G' \neq G$, then $G'$ is a \emph{proper subgraph} of $G$, which we write as $G' \subset G$. A \emph{cyclic} (or \emph{acyclic}) subgraph contains some (or no) simple cycles. If $G' \subseteq G$ and all the edges $(i,j) \in E$ such that $i,j \in V'$ are in $E'$, then $G'$ is an \emph{induced} subgraph of $G$ (otherwise it is \emph{non-induced}). We use the notation $M_{a}^{(b)}$ of \cite{lawford20, lawford_mehmeti20} to denote a specific non-induced subgraph, where $b$ is the number of nodes in the subgraph and $a$ is the decimal representation of the smallest binary number derived from a row-by-row reading of the upper triangles of each adjacency matrix $g$ across the set of all isomorphic subgraphs on the same $b$ nodes. Let $\widetilde{M}_{a}^{(b)}$ be an induced subgraph. The non-induced and induced subgraph counts in $G$ are denoted $|M_{a}^{(b)}|$ and $|\widetilde{M}_{a}^{(b)}|$. There are twenty-nine undirected, connected, and non-isomorphic subgraphs on three, four, or five nodes (see Table \ref{fig:motif_notation}). We illustrate using the 4-star:

\begin{example}[Notation]\label{ex:notation}

    \[
        \motif{}{}
        = 
        \begin{blockarray}{ccccc}
            & 1 & 2 & 3 & 4 \\
            \begin{block}{c(cccc)}
                1 & 0 & 1 & 1 & 1 \\
                2 & 1 & 0 & 0 & 0\\
                3 & 1 & 0 & 0 & 0\\
                4 & 1 & 0 & 0 & 0 \\
            \end{block}
        \end{blockarray}\enskip
        \cong
        \begin{blockarray}{ccccc}
            & 1 & 2 & 3 & 4 \\
            \begin{block}{c(cccc)}
                1 & 0 & 1 & 0 & 0 \\
                2 & 1 & 0 & 1 & 1\\
                3 & 0 & 1 & 0 & 0\\
                4 & 0 & 1 & 0 & 0 \\
            \end{block}
        \end{blockarray}\enskip
        \cong
        \begin{blockarray}{ccccc}
            & 1 & 2 & 3 & 4 \\
            \begin{block}{c(cccc)}
                1 & 0 & 0 & 1 & 0 \\
                2 & 0 & 0 & 1 & 0\\
                3 & 1 & 1 & 0 & 1\\
                4 & 0 & 0 & 1 & 0 \\
            \end{block}
        \end{blockarray}\enskip
        \cong
        \begin{blockarray}{ccccc}
            & 1 & 2 & 3 & 4 \\
            \begin{block}{c(cccc)}
                1 & 0 & 0 & 0 & 1 \\
                2 & 0 & 0 & 0 & 1\\
                3 & 0 & 0 & 0 & 1\\
                4 & 1 & 1 & 1 & 0 \\
            \end{block}
        \end{blockarray}\enskip.
    \]

\noindent We choose an arbitrary labelling of the nodes of the subgraph (from 1 to $b$), to give a subgraph $G'$. We then find all subgraphs that are isomorphic to $G'$, list their adjacency matrices, and then use the upper-triangular elements of each adjacency matrix, including leading zeros, to give binary representations, e.g., $111000_{2}$ and $100110_{2}$ and $010101_{2}$ and $001011_{2}$, respectively. We find the decimal representation of each binary number, and set $b$ equal to the minimum of these, e.g., we have $56_{10}$ and $38_{10}$ and $21_{10}$ and $11_{10}$, and so we use $M_{11}^{(4)}$ to denote the 4-star.
\end{example}
A $b$-complete subgraph is also called a \emph{clique}. A \emph{maximal clique} in a graph is a clique that cannot be made any larger by the addition of another node (and its edges) while preserving the complete-connectivity of the clique. A \emph{maximum clique} is a maximal clique with the largest possible number of nodes in the graph, and the \emph{clique number} $w(G)$ is the number of nodes in the maximum clique. Special graphs are the Erd{\H{o}}s-R{\'{e}}nyi random graph $G(n, p)$ with nodes $V = \{1,\ldots,n\}$ and edges that arise independently with constant probability $p$; the \emph{star} graph $S_{1,n-1}$ with \emph{center} node $i_{1}$ that is directly-connected to every other node (these edges are called \emph{spokes}), and that has no other edges; and the \emph{circle} graph $C_{n}$ with edges $(i, i+1) \in E$ for $i=1,\ldots,n-1$, and $(1,n) \in E$.

\subsection{Exact subgraph enumeration}\label{sec:counting_subgraphs}
We count each of the subgraphs in Table \ref{fig:motif_notation} using exact formulae. We emphasize here an important distinction between induced and non-induced subgraphs. A non-induced subgraph $H'$ with $b$ nodes and $c$ edges is allowed to be part of a ``larger'' subgraph $H$ on the same $b$ nodes, in the sense that $H' \subset H$, so that $H$ has more edges than $H'$. For example, the tadpole $M_{15}^{(4)}$ can be nested in the diamond $M_{31}^{(4)}$ and the 4-complete $M_{63}^{(4)}$, but it cannot be nested in the 4-star $M_{11}^{(4)}$ or the 4-path $M_{13}^{(4)}$ or the 4-circle $M_{30}^{(4)}$. Conversely, an induced subgraph $H'$ cannot be part of any larger subgraph $H$ on the same $b$ nodes, in the above sense. So, the set of all induced subgraphs of a given type, e.g., 3-stars, is a subset of all subgraphs of the same type, i.e., some 3-stars in a graph might be nested in triangles, while others are not. By definition, the triangle and 4-complete and 5-complete subgraphs must be induced. Throughout the paper, we allow arbitrary overlapping of nodes and edges between two subgraphs (this corresponds to the $\mathcal{F}_{1}$ ``frequency concept'' in \cite{schreiber_schwobbermeyer05}).

The analytical techniques that we use to derive exact enumeration results on the eight connected non-induced subgraphs with \emph{three or four nodes} are very well known and first appeared in \cite{alon_etal97}. For ease of reference, we collect these results in Theorem \ref{thm:analytic_count}, and provide a purely combinatorial and elementary proof for each formula, based on the number of closed walks and the description of more complicated subgraphs in terms of simpler ones, with no explicit mention of the moments of the spectral density (the relationship between eigenvalues and structural graph properties is discussed by \cite{harary_schwenk79}). It is also well understood that induced subgraph counts follow generally as linear combinations of non-induced counts, and we count induced subgraphs on three and four nodes using the exact formulae in Theorem \ref{thm:analytic_count_not_nested}, which we prove using elementary combinatorial methods. In our implementation, we use memoization of non-induced counts to give a very fast computation of induced counts. There has been much less work on \emph{five-node} subgraphs, and complete results on non-induced and induced subgraph counts were given only recently by \cite[Theorem 2.1, Theorem A.1]{lawford20} and, using a different method of proof, by \cite{pinar_etal17}. In this paper, we use the results of \cite{lawford20} to count non-induced and induced five-node subgraphs.\footnote{There is a rich and fascinating literature in computer science on fast matrix multiplication, subgraph counting, listing of subgraphs and maximal cliques, and motif detection, with development of exact and approximate algorithms that work well on very large graphs. However, it is not the aim of our paper to provide more efficient routines for massive datasets, or to develop the fastest possible algorithms, when very accurate exact methods have good practical runtime performance. For a brief history of fast matrix multiplication, see \cite[Section 1]{vassilevskawilliams14}. Efficient algorithms for listing all triangles in a graph are given by \cite{bjorklund_etal14}, while \cite{chu_cheng12} develop an exact triangle listing algorithm based on iterative partitioning of the input graph $G$, and survey other triangle listing algorithms. Fast algorithms for finding some 4-node subgraphs are presented in \cite{vassilevskawilliams_etal15}. For a short discussion of the \emph{k-clique problem} see \cite{vassilevska09}. Subgraph enumeration on large graphs is discussed by \cite{itzhack_etal07, kashtan_etal04b}.  State-of-the-art network motif detection algorithms are surveyed by \cite{khakabimamaghani_etal13, tran_etal14, wong_etal11}, who report experimental evidence on the runtime of eleven software tools.} See Appendix \ref{sec:bounds_on_complete} for weak bounds on the number of complete subgraphs.

\subsection{Real-world network data}\label{sec:network_data}
Our network data is constructed from the U.S. Department of Transportation's DB1B Airline Origin and Destination survey over the period 1999Q1 to 2013Q4.\footnote{The data is publicly-available, and can be downloaded from \url{http://www.transtats.bts.gov/}\ } The source provides quarterly information on a 10\% random sample of all tickets that were sold for domestic U.S. airline travel, and has been widely used in the economics literature, e.g., \cite{aguirregabiria_ho12, ciliberto_tamer09, dai_etal14, goolsbee_syverson08}. In this paper, we focus on one carrier, Southwest Airlines, which appears in every quarter of the full sample, and is the largest (number of nodes and edges) and densest ($d(G)$) network available in the dataset. We drop any tickets that were sold under a codeshare agreement, or that had unusually high or low fares. We retain coach class tickets, unless more than 75\% of the carrier's tickets in a particular quarter were reported as either business or first class, in which case we keep all tickets for that carrier. We aggregate individual tickets to unidirectional route-level observations, and drop routes that have very few passengers, or that do not have a constant number of passengers on each segment. We refer to airports using the official three-letter IATA designators. For further details on the data treatment, see \cite{lawford20, lawford_mehmeti20}. For each quarter, we build the associated simple unweighted and undirected graph (or ``route map'') $G$ as follows: (node) the set of nodes $V$ are all airports that served as an origin or destination on some route for Southwest in that quarter; (edges) the set of edges $E$ are all non-directional airport--airport routes for which a sufficient number of passengers bought tickets for direct travel.

\begin{figure}
    \begin{subfigure}{.5\textwidth}
      \centering
       \includegraphics[width=.8\linewidth]{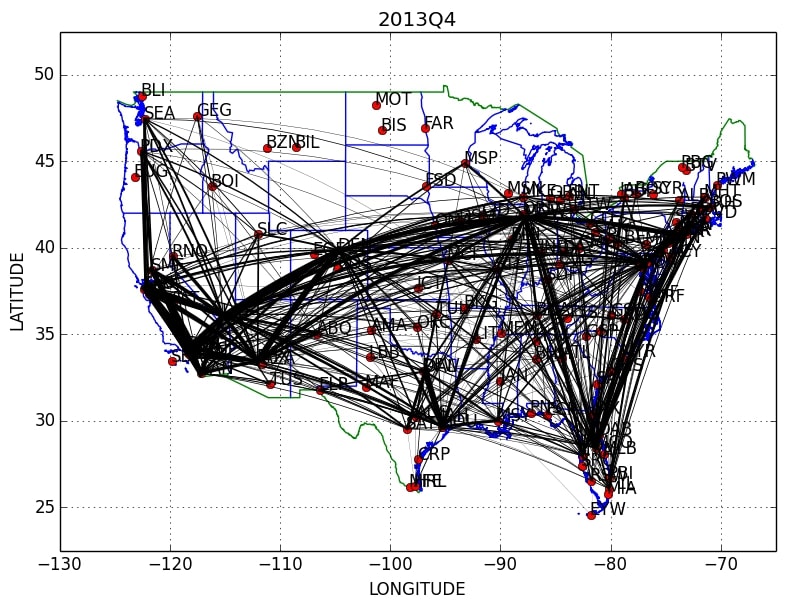}
      \caption{Spatial network.}
      \label{fig:southwest_network_2013q4}
    \end{subfigure}
    \begin{subfigure}{.5\textwidth}
      \centering
      \includegraphics[width=.8\linewidth, height=.24\textheight, trim=0cm 0cm 0cm 0cm, clip]{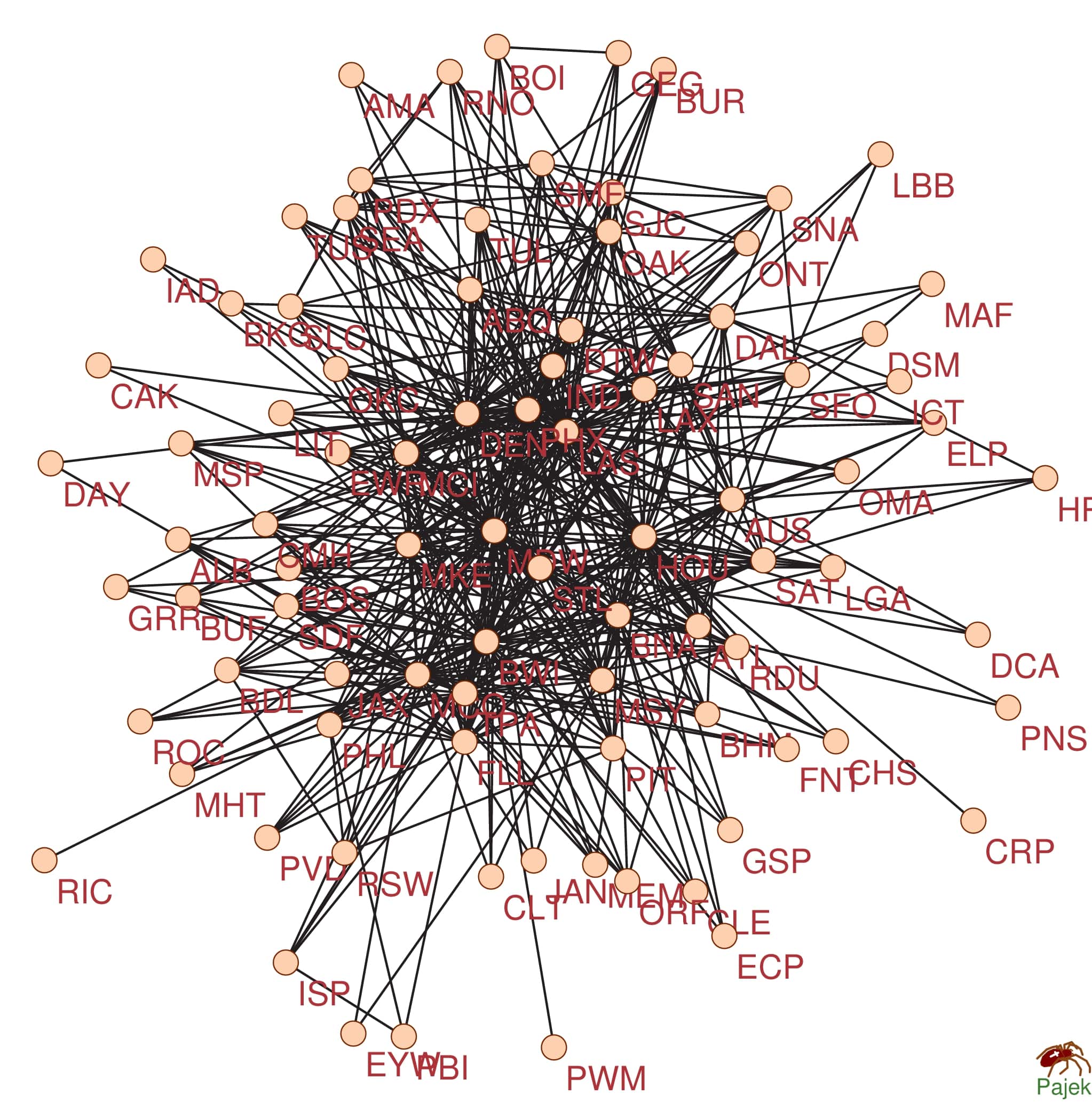}
      \caption{Topological network.}
      \label{fig:southwest_network_2013q4_topology}
    \end{subfigure}
    	\caption{Southwest's network in 2013Q4, computed using nondirectional nonstop round-trip coach class tickets. Routes in (\ref{fig:southwest_network_2013q4}) are plotted as minimum-distance paths between directly-connected origin and destination airports; the line width is proportional to the number of passengers on each route, from the U.S. Department of Transportation's Airline Origin and Destination Survey (DB1B). The topological network in (\ref{fig:southwest_network_2013q4_topology}) was plotted using Pajek's \cite{mrvar_batagelj16} Kamada-Kawai visualization algorithm.}
    	\label{fig:illustrative_networks}
\end{figure}

In Figure \ref{fig:illustrative_networks}, we give two representations of the 2013Q4 network. The spatial plot shows that passenger activity is highly concentrated between particular geographical areas, creating clearly visible traffic ``corridors'', while other regions have little or no service. The topological plot gives some insight into the degree distribution, with low-degree and high-degree nodes both apparent. Southwest's network has grown steadily over the sample period, from $(n, m) = (54, 251)$ in 1999Q1 to $(n, m) = (88, 522)$ in 2013Q4. Figure \ref{fig:network_dynamics} displays some properties of the network, and the salient features of the numerical data are as follows:
\begin{itemize}
\item The number of edges increases almost linearly, while the number of nodes increases slowly until the last two years of the sample, followed by a more rapid increase (Figure \ref{fig:nodes_and_edges_dynamic}).
\item The density was stable from 2001 to 2009, at around 0.20, but fell sharply thereafter, to below 0.14, as the increase in nodes was not matched proportionally by new edges (Figure \ref{fig:density_WN_2013_4}).
\item The diameter and average path length of Southwest's network are, respectively, 3--4 and roughly 2. These values are very close to those of the corresponding $G(n, p)$, when the edge-formation probability $p$ is set equal to the density of Southwest's network.\footnote{We computed the statistics for $G(n, p)$ using 1,000 replications for each time period, except for expected overall and average clustering, for which we used 100 replications. See \cite{barabasi16} for a non-technical introduction to random graphs.} Viewed through this global lens, Southwest's network might appear to behave very much like the random $G(n, p)$ (Figure \ref{fig:diameter_apl}). However, we see later that there are various critical differences, and Southwest's network cannot be considered as an Erd{\H{o}}s-R{\'{e}}nyi random graph.
\item The overall clustering coefficient measures the fraction of connected triples of nodes that have their third edge connected to form a triangle; the average clustering coefficient computes this measure on a node-by-node basis and then averages across nodes. There is considerably more clustering (both overall and average) in Southwest's network than in the random $G(n, p)$, for which expected overall and average clustering are identical and equal to the density $p$ (Figure \ref{fig:clustering}).
\item Despite the global stability of Southwest's network that is suggested by diameter and average path length, there is considerable dynamic variation at the local route level: on average, 2.5\% of routes in a quarter were not served in the previous quarter, and 1.2\% of routes that were served in the previous quarter were closed in the subsequent one (Figure \ref{fig:plus_minus_routes}).
\item There is substantial heterogeneity in the \emph{degree centrality} $DC_{i} = k_{i} / (n - 1)$ across different nodes. We illustrate this with Denver (DEN), Detroit Metropolitan (DTW), Las Vegas McCarran (LAS), Chicago Midway (MDW) and Phoenix Sky Harbor (PHX). Midway has experienced several discrete jumps in its activity. Denver entered the network in 2006Q1, with direct links to 8\% of other nodes, a figure that rose to a maximum of 72\% of other nodes in 2012Q4. Denver and Midway are the two airports that have seen the largest change in degree centrality over the sample period (Figure \ref{fig:centrality_nodes}).
\end{itemize}

\begin{figure}
    \begin{subfigure}{.5\textwidth}
      \centering
       \includegraphics[width=.8\linewidth]{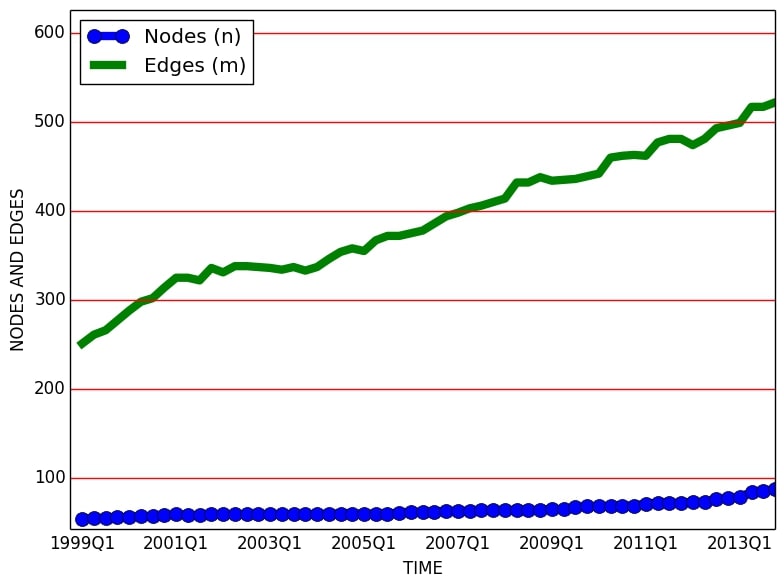}
      \caption{Number of nodes and edges.}
      \label{fig:nodes_and_edges_dynamic}
    \end{subfigure}
    \begin{subfigure}{.5\textwidth}
      \centering
      \includegraphics[width=.8\linewidth]{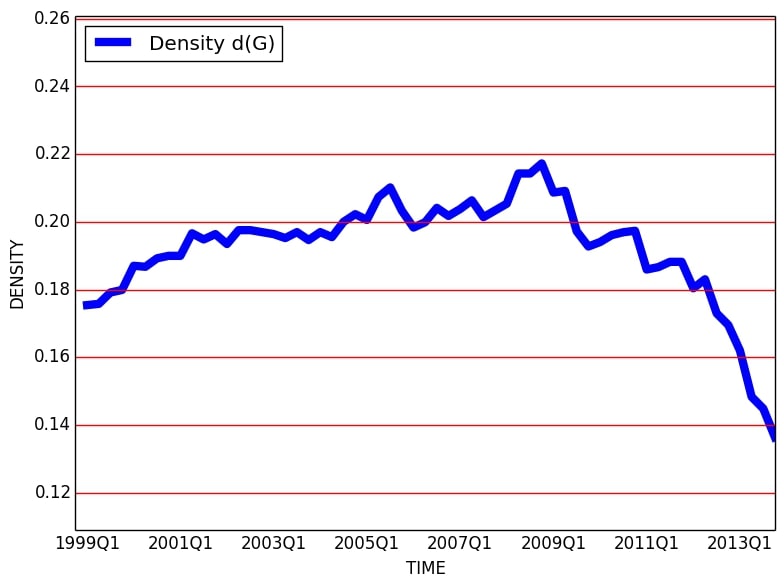}
      \caption{Density.}
      \label{fig:density_WN_2013_4}
    \end{subfigure}
    \begin{subfigure}{.5\textwidth}
      \centering
       \includegraphics[width=.8\linewidth]{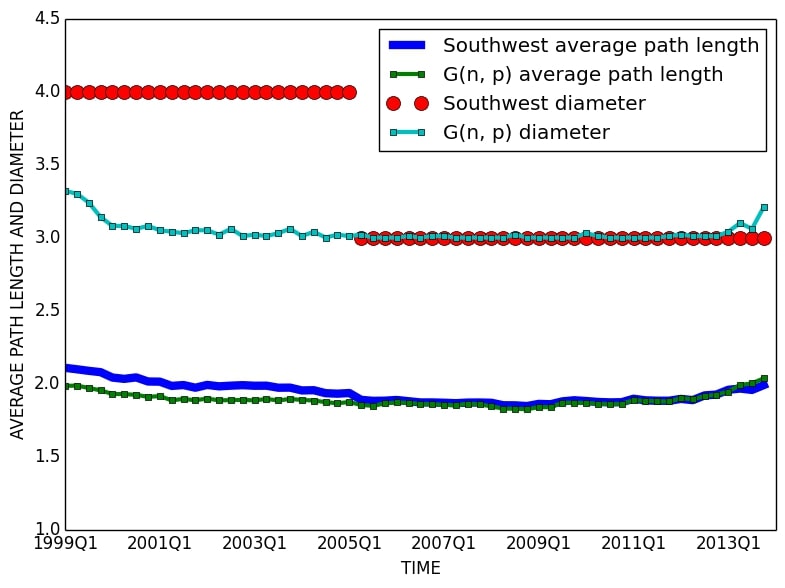}
      \caption{Diameter and average path length.}
      \label{fig:diameter_apl}
    \end{subfigure}
    \begin{subfigure}{.5\textwidth}
      \centering
      \includegraphics[width=.8\linewidth]{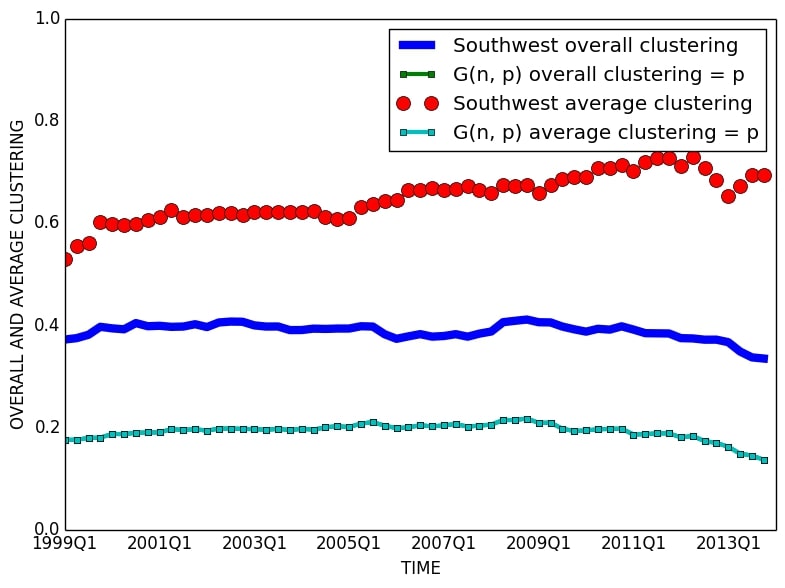}
      \caption{Overall and average clustering.}
      \label{fig:clustering}
     \end{subfigure}
    \begin{subfigure}{.5\textwidth}
      \centering
       \includegraphics[width=.8\linewidth]{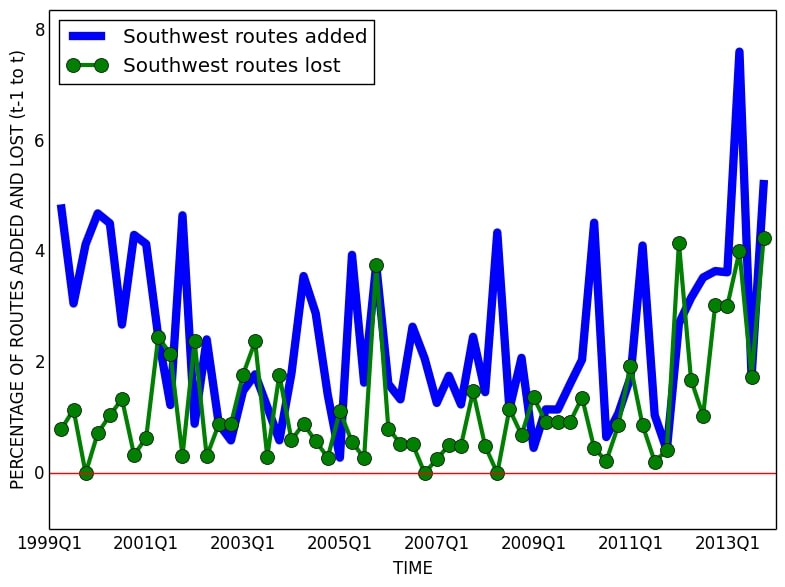}
      \caption{Percentage of routes added and lost.}
      \label{fig:plus_minus_routes}
    \end{subfigure}
    \begin{subfigure}{.5\textwidth}
      \centering
      \includegraphics[width=.8\linewidth]{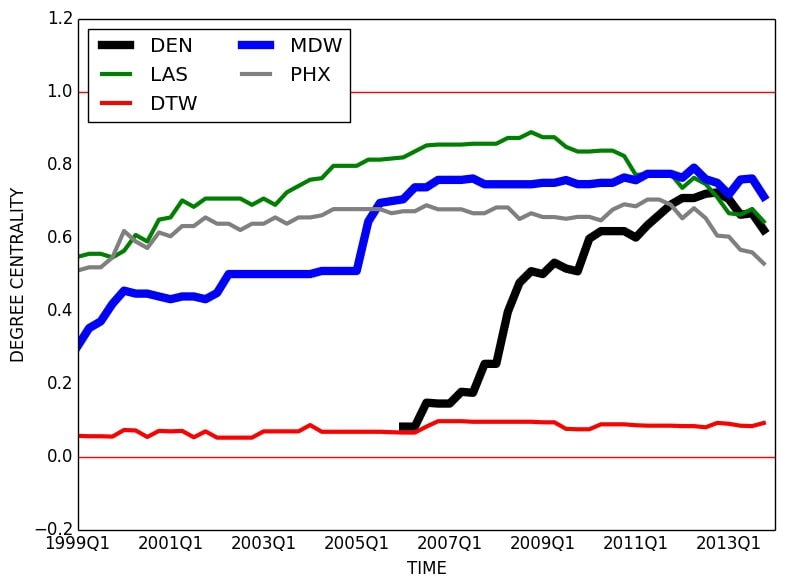}
      \caption{Degree centrality, selected nodes.}
      \label{fig:centrality_nodes}
    \end{subfigure}
    	\caption{Global and local properties of Southwest's network from 1999Q1 to 2013Q4. (\ref{fig:nodes_and_edges_dynamic} and \ref{fig:density_WN_2013_4}) The number of edges ($m$) increases approximately linearly across the sample period, while the density falls sharply after 2012 due to a rapid increase in the number of new airports ($n$) that was not matched proportionally by new routes; (\ref{fig:diameter_apl}) the diameter and average shortest path length compared to the Erd{\H{o}}s-R{\'{e}}nyi random graph $G(n, p)$ with density $p$ equal to the density of Southwest's network; (\ref{fig:clustering}) the overall (average) clustering coefficient for Southwest is about two (three) times the clustering of $G(n, p)$; (\ref{fig:plus_minus_routes}) there is generally a net increase in the number of routes between successive quarters; (\ref{fig:centrality_nodes}) heterogeneity in degree centrality over time, for different airports.}
    	\label{fig:network_dynamics}
\end{figure}

The tendency for many sparse real-world networks to have average path lengths close to those of a random graph but with much higher local clustering (nodes have many mutual neighbours) is called the \emph{small-world} property. An elegant theoretical explanation for this is given by \cite{watts_strogatz98} who show that the presence of a small number of ``short-cut'' edges, which connect nodes that would otherwise be farther apart than the average path length in a random network, can lead to a rapid fall in average path length while having very little impact on local clustering.\footnote{A ``sparse'' network is defined by \cite{watts_strogatz98} as one which satisfies, in our notation, $n \gg (m/n) \gg \log(n) \gg 1$. For Southwest's 2013Q4 network, we have $(n, m) = (88, 522)$, whereupon $n = 88 > (m/n) \approx 5.93 > \log(n) \approx 4.48 > 1$. Throughout, $\log(\cdot)$ refers to the natural logarithm. Evidence that Southwest's network is small-world is also presented by \cite[Section 3.1]{lawford_mehmeti20} and \cite[Section II.A]{wuellner_etal10}.} This is consistent with the presence of a small number of high degree ``hub'' nodes in Southwest's network. For a longer theoretical treatment of the small-world property, with a focus on social networks, see \cite{watts99}.

\subsection{Scaling of subgraph counts}\label{sec:scaling_properties}
There has been much interest in the statistical physics literature on rules for the scaling of subgraph counts with measures of network size and structure, e.g., \cite{itzkovitz_alon05, itzkovitz_etal03}. We report log non-induced subgraph counts over time in Table \ref{fig:non_induced_counts}. We observe that there is substantial variation in the count across subgraph type, e.g., the counts in 2013Q4 of the triangle $M_{7}^{(3)}$ and the 4-path $M_{13}^{(4)}$ differ by several orders of magnitude. For much of the sample, the slope of the log count is close to linear, indicating approximately constant percentage growth in subgraph counts over time. In part, the similar slopes across subgraphs reflect the correlation in non-induced counts, e.g., adding a triangle to a network will also increase the non-induced 3-star count. Because the number of edges increases almost linearly over time, we focus on the relationship between the subgraph count and the number of edges in the network.

Figure \ref{fig:log_log_slope_all_subgraphs} summarizes the estimated slope coefficient and coefficient of determination $R^{2}$ from a least squares regression of $\log |M_{a}^{(b)}|$ on a constant and $\log(m)$, across the 60 quarters of the sample, for each of the subgraphs in Table \ref{fig:motif_notation}. Numerical results are reported in Table \ref{tab:log_log_results}. Together, these results suggest that each non-induced subgraph count can be well-approximated by a power-law $|M_{a}^{(b)}| = A \, m^{\beta}$, where $A$ is a constant, $m$ is the number of edges in the graph, and $\beta$ is the slope coefficient from the log-log regression. Hence, as $m$ increases by a multiplicative factor $\kappa$, the subgraph count will increase by a factor $\kappa^{\beta}$. We would not expect this scaling to arise by tautology.\footnote{For excellent surveys of research on empirical power-laws in economics and finance (including firm and city sizes, and CEO compensation), with discussion of theoretical mechanisms such as random growth that result in scaling behaviour, and of economic complexity more generally, see \cite{durlauf05, gabaix09, gabaix16}. Power-laws in a variety of real-world datasets are described by \cite{clauset_etal09}, who discuss statistical methods that can distinguish between power-laws and alternative models; and \cite{gabaix_etal03} present a model of power-law movements in stock prices, and in the volume and number of financial trades.} These results lead us to make two observations, which we address below:

\begin{itemize}
\item It appears that the slope $\beta$ is closely related to the number of nodes $b$ \emph{in each subgraph}, since $\beta \approx b - 1$ in Figure \ref{fig:log_log_slope_all_subgraphs} and Table \ref{tab:log_log_results}. We observe little evidence of any power-law scaling for induced subgraphs.

\item The scaling seems to hold across a wide range of network sizes ($m$), and appears robust to the large fall in network density that is observed during the last couple of years of the sample (Figure \ref{fig:density_WN_2013_4}). We see later that this apparent robustness misses important signs of regime-switching scaling (Section \ref{sec:regime_switch_scaling}).
\end{itemize}

\begin{figure}\centering
    	\includegraphics[width=.65\linewidth, keepaspectratio]{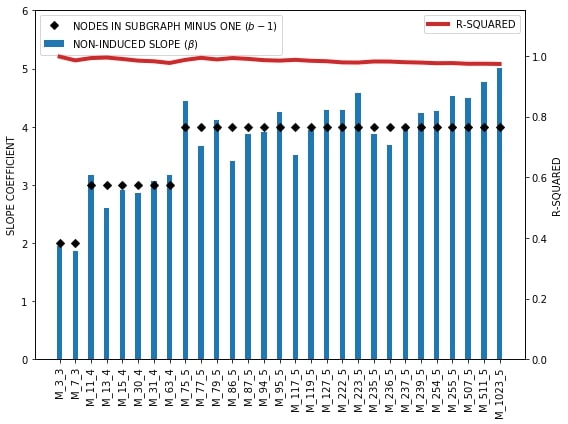}
    	\caption{A summary of log-log least squares regressions of non-induced subgraph count on number of edges, for Southwest's network, over the period 1999Q1 to 2013Q4 (60 datapoints). The vertical blue bars represent the estimated slope coefficient $\beta$. The red line plots the computed coefficient of determination $R^{2}$. The black diamonds correspond to $b-1$, where $b$ is the number of nodes in each respective subgraph. There is some evidence that $\beta \approx b-1$ for each 3-node, 4-node and 5-node subgraph, indicating a power-law relationship between subgraph count and edges.}
    	\label{fig:log_log_slope_all_subgraphs}
\end{figure}

\noindent To start with, is it surprising that there is \emph{any} scaling behaviour in Southwest's network? To build some intuition, we consider the implied scaling from three standard random and deterministic theoretical network models, each of which imposes a specific relationship between the number of edges $n$ and the number of edges $m$.

\begin{example}[Subgraphs in the Erd{\H{o}}s-R{\'{e}}nyi random graph]
Let $G = G(n,p)$ be an Erd{\H{o}}s-R{\'{e}}nyi random graph with expected degree $\mathbb{E}(k) = (n-1) \, p$, where $n$ is allowed to vary, but the edge-formation probability $p > 0$ is fixed. A general scaling result for the count of a non-induced subgraph, with $b$ nodes and $c$ edges, is given by \cite[eqn. (7)]{itzkovitz_alon05}
\begin{equation}\label{eq:erdos_renyi_scaling}
    \mathbb{E}(|M_{a}^{(b)}|) \sim n^{b-c} (\mathbb{E}(k))^{c} \sim n^{b} \, p^{c},
\end{equation}
as $n \to \infty$. It is straightforward to verify (\ref{eq:erdos_renyi_scaling}) for particular non-induced subgraphs, by combinatorial methods. Using linearity of expectation, it is sufficient to find the number of each subgraph $S$ of interest in $K_{n}$, and then to multiply this count by $p^{|E(S)|}$ to give the count in $G(n, p)$. Consider the expected tadpole count $\mathbb{E}(|M_{15}^{(4)}|)$. There are $\binom{n}{3}$ triangles in $K_{n}$. Let one of the 3 corners of a triangle be the ``center'' (degree 3 node) of a tadpole. A tadpole is formed by linking the center to one of $(n - 3)$ other nodes in $K_{n}$. Hence, $\mathbb{E}(|M_{15}^{(4)}|) = 3 \, \binom{n}{3} \, (n-3) \, p^{4} = 12 \, \binom{n}{4} \, p^{4} \sim (1/2) \, n^{4} \, p^{4}$. Alternatively, we can specialize the analytic tadpole count (\ref{eq:m_15_4}) from Theorem \ref{thm:analytic_count} to $K_{n}$:
\begin{equation*}
|M_{15}^{(4)}| = \frac{1}{2} \sum_{i:k_{i} > 2}(g^{3})_{ii} \, (k_{i}-2),
\end{equation*}
noting that $k_{i} = (n-1)$ for all $i$. Upon applying \cite[Lemma 3.1]{lawford20} to give $(g^{3})_{ii} = (n-1)(n-2)$, the result follows immediately. It is easy to derive analytic formulae for expected induced subgraph counts in $G(n, p)$ by
\begin{equation}\label{eq:erdos_renyi_scaling_induced}
    \mathbb{E}(|\widetilde{M}_{a}^{(b)}|) = \mathbb{E}(|M_{a}^{(b)}|) \times (1-p)^{\binom{b}{2} - c}.
\end{equation}
Continuing our example, $\mathbb{E}(|\widetilde{M}_{15}^{(4)}|) = 12 \, \binom{n}{4} \, p^{4} \, (1-p)^{2} \sim (1/2) \, n^{4} \, p^{4} \, (1-p)^{2}$. We note that $\mathbb{E}(m) = \binom{n}{2} \, p$ implies that
\begin{equation*}
n = \frac{1}{2}\left(1 + \sqrt{1 + \frac{8 \mathbb{E}(m)}{p}}\right)
\end{equation*}
exactly for $p \neq 0$. Combining this with $\mathbb{E}(|M_{a}^{(b)}|) \sim n^{b} \, p^{c}$, and noting that $\mathbb{E}(m) \to \infty$ as $n \to \infty$, we have $\mathbb{E}(|M_{a}^{(b)}|) \sim 2^{b/2} \, p^{c-b/2} \, \mathbb{E}(m)^{b/2}$ and so
\begin{equation*}
    \log(\mathbb{E}(|M_{a}^{(b)}|)) \sim \mathrm{constant} + \frac{b}{2} \, \log(\mathbb{E}(m))
\end{equation*}
as $\mathbb{E}(m) \to \infty$. Hence, we would expect the slope of a log-log regression of non-induced subgraph count on number of edges, for a \emph{given} realization of $G(n,p)$, to be close to $\beta = b/2$ as the number of edges increases. From (\ref{eq:erdos_renyi_scaling_induced}), we would also expect to see the same slope for \emph{induced} subgraph scaling. We verified both observations numerically.

\end{example}

\begin{example}[Star subgraphs in the deterministic star graph]
Now consider the deterministic model based upon the star $S_{1, n-1}$ and count the number of $b$-star subgraphs. Since $m = n - 1$, it follows immediately from (\ref{eq:m_3_3}) and (\ref{eq:m_11_4}) that (3-star) $|M_{3}^{(3)}| \sim (1/2) \,  m^{2}$ and (4-star)  $|M_{11}^{(4)}| \sim (1/6) \, m^{3}$, with implied log-log slopes of $\beta = 2$ and $\beta = 3$ respectively. In general, the number of $b$-star subgraphs in $S_{1,n-1}$ equals
\begin{equation*}
\binom{n-1}{b-1} \sim \frac{n^{b-1}}{(b-1)!} \sim \frac{m^{b-1}}{(b-1)!}
\end{equation*}
as $m \to \infty$, with an implied slope $\beta = b-1$ for both non-induced and induced stars.

\end{example}

\begin{example}[Star and triangle subgraphs in the deterministic circle graph]
While it is tempting to conjecture that $\beta$ always increases in the number of nodes $b$ in the subgraph, an easy counterexample shows that this is not true. Consider the number of $b$-star subgraphs in the circle graph $C_{n}$. Since each node has degree $k_{i} = 2$, it follows from (\ref{eq:m_3_3}) that $|M_{3}^{(3)}| = n = m$, with implied slope $\beta = 1$. However, there are \emph{no} $b$-star subgraphs in the circle $C_{n}$, and so the implied slope is $\beta = 0$ for all $b > 3$. Similarly, there will be no triangles in any circle $C_{n}$.

\end{example}

It seems that any scaling behaviour in real-world networks will generally depend upon (i) the number of nodes $b$ in the subgraph, (ii) the topology of the subgraph for a given $b$, e.g., the 3-star or the triangle, and (iii) the properties of the graph $G$ in which these subgraphs are contained, and the way in which the topology of $G$ evolves as $n$ increases, including the implied relationship between the numbers of nodes and edges in $G$. Moreover, the underlying model of evolution might change over time. It is unclear whether we would find $\beta \approx b-1$ in other real-world networks of interest. Clearly, the scaling behaviour of Southwest's network across the full sample is very different to that of Erd{\H{o}}s-R{\'{e}}nyi, even though the graphs have some common topological features such as similar average path length.

\subsubsection{Robustness of scaling to changes in network evolution in a toy regime-switching model}\label{sec:scaling_regime_switch}
We now focus on the observation that the scaling in Figure \ref{fig:log_log_slope_all_subgraphs} appears to be robust to a significant change in network evolution: in 2012 and 2013, the net number of routes increased at a much slower rate, relative to the net number of airports, than it did before 2012, with a resulting fall in network density (Figure \ref{fig:density_WN_2013_4}). We obtain analytic results for a toy regime-switching model of network evolution, and show how apparently robust scaling can appear in non-induced subgraph counts, despite significant underlying changes in the dynamics.

Consider a deterministic dynamic network model that starts with two connected nodes and adds one additional node in each subsequent time period. There are two regimes, where $\ell$ is the total number of nodes in the network in a given time period:

\begin{itemize}
\item (\textbf{Regime 1}) For $\ell \leq n^{\star}$, the network evolves as an $n$-star. One of the initial two nodes is chosen to be the (fixed) center, and each subsequent node links only to the center node.

\item (\textbf{Regime 2}) For $\ell > n^{\star}$, each subsequent node links to \emph{all} existing nodes.
\end{itemize}

\noindent So, $n^{\star}$ is the network size at which the model of evolution switches from Regime 1 to Regime 2. The network will evolve as an $n$-star for $n \leq n^{\star}$ and will become increasingly like a complete graph as $n > n^{\star}$ and $n$ becomes large. See Figure \ref{fig:regime_switch} for an illustration, with $n^{\star} = 5$. Consider the number of 3-stars in the combined network described by Regimes 1 and 2.

    \begin{figure}\centering
    	\begin{subfigure}{0.32\textwidth}
    		\centering
    		\includegraphics[width=.5\linewidth, keepaspectratio]{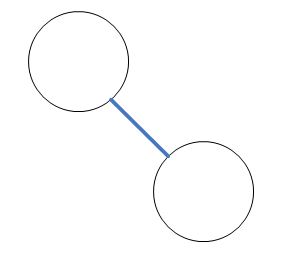}
    		\caption{$\ell=2$ nodes.}
    		\label{figregime_switch_step_2}
    	\end{subfigure}
    	\begin{subfigure}{0.32\textwidth}
    		\centering
    		\includegraphics[width=.68\linewidth, keepaspectratio, height=.165\textheight]{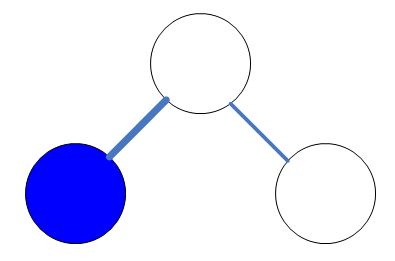}
    		\caption{$\ell=3$ nodes.}
    		\label{fig:regime_switch_step_3}
    	\end{subfigure}
    		\begin{subfigure}{0.32\textwidth}
    		\centering
    		\includegraphics[width=.5\linewidth, keepaspectratio]{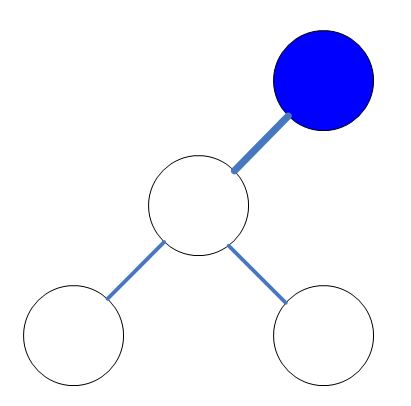}
    		\caption{$\ell=4$ nodes.}
    		\label{fig:regime_switch_step_4}
    	\end{subfigure}
    	\begin{subfigure}{0.32\textwidth}
    		\centering
    		\includegraphics[width=.5\linewidth, keepaspectratio]{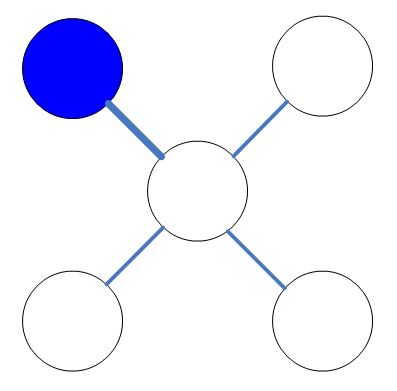}
    		\caption{$\ell=n^{\star}=5$ nodes.}
    		\label{fig:regime_switch_step_5}
    	\end{subfigure}
    	\begin{subfigure}{0.32\textwidth}
    		\centering
    		\includegraphics[width=.58\linewidth, keepaspectratio, height=.175\textheight]{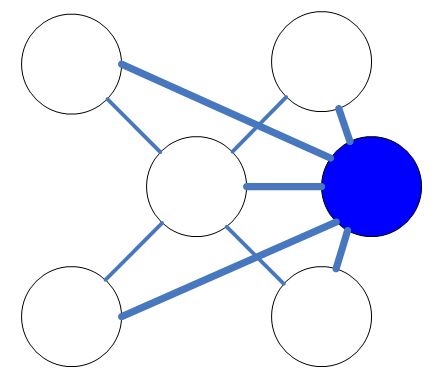}
    		\caption{$\ell=6$ nodes.}
    		\label{fig:regime_switch_step_6}
    	\end{subfigure}
    		\begin{subfigure}{0.32\textwidth}
    		\centering
    		\includegraphics[width=.6\linewidth, keepaspectratio]{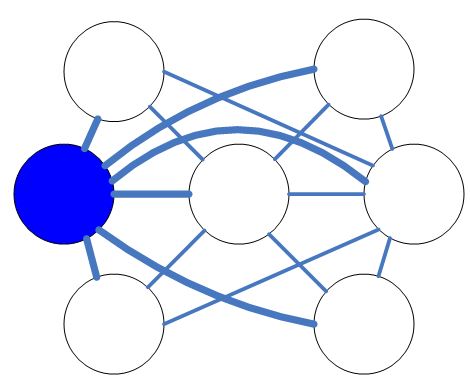}
    		\caption{$\ell=7$ nodes.}
    		\label{fig:regime_switch_step_7}
    	\end{subfigure}
    	\caption{An illustrative regime-switching model on $\ell \geq 2$ nodes, with regime change after $\ell = n^{\star} = 5$ (Section \ref{sec:scaling_regime_switch}). In each step, the new node and new edges are highlighted in bold. The model initially evolves as a star graph. After the regime switch, new nodes connect to all existing nodes, and the graph approaches a complete graph as $\ell \to \infty$.}
    	\label{fig:regime_switch}
    \end{figure}

\begin{itemize}

\item (\textbf{Regime 1}) Here, $\ell = 2, 3, \ldots, n^{\star}$ and $m = \ell - 1$. From (\ref{eq:m_3_3}), the number of non-induced 3-stars is given by $|M_{3}^{(3)}| = \binom{m}{2} = (1/2) \, m(m-1) \sim (1/2) \, m^{2}$ as $n^{\star}$ increases.

\item (\textbf{Regime 2}) Here, $\ell = n^{\star} + 1, n^{\star} + 2, \ldots, n$. When $\ell = n^{\star} + 1$, we obtain $m = (n^{\star} - 1) + n^{\star} = 2n^{\star} - 1$. When $\ell = n^{\star} + 2$, we have $m = (n^{\star} - 1) + n^{\star} + (n^{\star} + 1) = 3n^{\star}$. In general, and setting $a = \ell - n^{\star}$, we can show that the number of edges is given by

\begin{equation*}
    m = (a + 1)n^{\star} + \sum_{j=-1}^{a-1}j = (a + 1)n^{\star} + \sum_{j=1}^{a+1}(j-2) \nonumber = (a + 1)\left(n^{\star} + \frac{a}{2} - 1\right),
\end{equation*}

and so $m \sim (1/2) \, a^{2}$ as $a$ becomes large (so that $n \gg n^{\star}$). From (\ref{eq:m_3_3}), the non-induced 3-star count is $|M_{3}^{(3)}| = (1/2)\sum_{k \in \mathcal{D}}k(k-1)$, where $\mathcal{D}$ is the set of node degrees. In Regime 1, $\mathcal{D} = \{1, 1, \ldots, 1, \ell - 1\}$, where $\ell - 1$ nodes have degree 1. In Regime 2, $\mathcal{D} = \{(a + 1), \ldots, (a + 1), (n^{\star} + a - 1), \ldots, (n^{\star} + a - 1)\}$, where $n^{\star} - 1$ nodes have degree $a + 1$, and $a + 1$ nodes have degree $n^{\star} + a - 1$. Putting this all together, it follows that
\begin{equation*}
|M_{3}^{(3)}| = \frac{1}{2}(a + 1)\left(a(n^{\star} - 1) + (n^{\star} + a - 1)(n^{\star} + a - 2)\right) \sim \frac{1}{2}a^{3},
\end{equation*}
as $a$ becomes large. Using $m \sim (1/2) \, a^{2}$ and $|M_{3}^{(3)}| \sim (1/2) \, a^{3}$, we have $|M_{3}^{(3)}| \sim 2^{1/2} \, m^{3/2}$ in Regime 2.

\end{itemize}

\noindent Hence, there will be a transition in the implied scaling slope $\beta$ as the regime changes, from 2 in Regime 1 (if $n^{\star}$ is large) to 3/2 in Regime 2 (if $n$ is large relative to $n^{\star}$). Although there is a different degree of scaling in each regime, and a very different model of evolution before and after $n^{\star}$, if we ignore this and apply least squares to the entire sample ($\ell = 1, \ldots, n$) then the regression slope will be a weighted average of the slopes in the individual regimes.\footnote{In Figure \ref{fig:regime_smoothing} we illustrate the toy model using simulated data, with $\ell = 4, \ldots, n^{\star} = 20, \ldots, n = 30$: the slopes of 2 (in Regime 1) and 1.5 (in Regime 2) are averaged by the regression to 1.56, with a very high $R^{2}$ of 0.983.}

\begin{figure}\centering
    	\includegraphics[width=.55\linewidth, keepaspectratio]{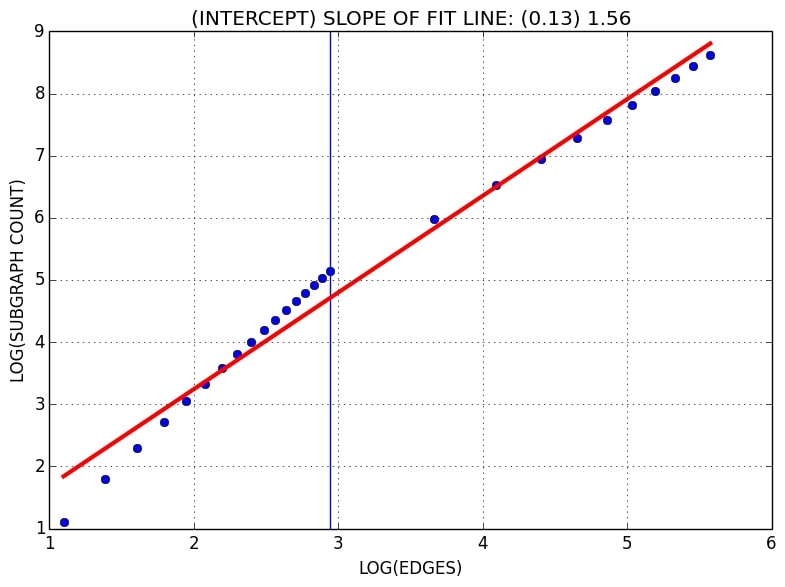}
    	\caption{Averaging of the scaling factor by log-log regression fit, for a toy regime-switching model (Section \ref{sec:scaling_regime_switch}), illustrated using simulated data, with $\ell = 4, \ldots, n = 30$, and breakpoint $n^{\star} = 20$ between Regimes 1 and 2.}
    	\label{fig:regime_smoothing}
\end{figure}

\subsubsection{Evidence for a regime-switch in power-law scaling for Southwest's network}\label{sec:regime_switch_scaling}
It is possible that Southwest's network evolution changed substantially after about 2009, when the growth in network density appeared to stall, and that the regressions are averaging the scaling in two or more regimes. Indeed, the regression errors do seem to be consistently larger at the end (and start) of the sample, when $n$ is largest (smallest). We investigate by fitting log-log regressions to subsamples, from 1999 to 2008 (40 datapoints) and from 2009 to 2013 (20 datapoints). This provides surprising support for strong \emph{but different} power-law scaling before and after 2009, suggesting a significant change in the model of network evolution at or around that time. There is slightly weaker but still good evidence that $\beta \approx b - 1$ over 1999 to 2008 compared to the full sample (Figure \ref{fig:log_log_slope_all_subgraphs_regime1}) and quite convincing evidence that $\beta \approx b / 2$ over 2009 to 2013 (Figure \ref{fig:log_log_slope_all_subgraphs_regime2}). Of course, this is not proof that Southwest's network actually evolves as an Erd{\H{o}}s-R{\'{e}}nyi random graph after 2009 even though the implied scaling is the same. 

\subsection{Which subgraphs are motifs?}\label{sec:identify_motifs}
Do any induced subgraphs arise more (or less) often than we would expect at random? We consider the significance of subgraph counts against two randomized null networks, to detect: (a) 3-node motifs relative to Erd{\H{o}}s-R{\'{e}}nyi $G(n, p)$, (b) 3-node motifs relative to a degree-preserving rewiring of the original network, (c) 4-node motifs relative to a distribution that controls for the number of 3-node induced subgraphs in the network, and (d) 5-node motifs relative to a distribution that controls for the number of 3-node and 4-node induced subgraphs in the network. It is well-known that the choice of null distribution is of critical importance in determining which subgraphs are identified as motifs. Certainly, the null should have some of the properties of the real-world network.\footnote{See \cite[Appendix A]{itzkovitz_etal05} for some discussion of randomized ensembles subject to constraints. We also experimented with variants of the \emph{erased configuration model}, with and without some clustering, e.g., \cite{angel_etal17, newman09, schlauch_zweig15}, but found that these gave some self-loops and many multiple-edges. Since our networks are quite small, we cannot make use of the observation that these issues are not important asymptotically. The resulting randomized graphs have rather different properties to the real-world networks and so we did not use this approach.} It is also important to search for \emph{induced} rather than non-induced subgraph motifs, for two reasons. First, non-induced star subgraph counts depend only on $P(k)$. It follows that $|M_{7}^{(3)}|$ and $|M_{11}^{(4)}|$ and $|M_{75}^{(5)}|$ will be invariant to a degree-preserving rewiring, and so we will not be able to use that approach to find non-induced $b$-star motifs in general. Second, a motif is naturally interpreted as a specific (unique) topology on a given set of $b$ nodes, and it makes intuitive sense to look for induced subgraphs: a given set of nodes form a motif (or an anti-motif) when they do not have any more complicated topological interrelationship, and their topology is statistically over-represented (or under-represented) in the network. We perform inference using the z-score of each subgraph count.

\subsubsection{3-node motifs relative to $G(n, p)$}\label{sec:motif_g_n_p}
We note a result of Ruci{\'{n}}ski that gives the asymptotic distribution of the z-score relative to $G(n, p)$:

\begin{theorem}[Asymptotic normality of the z-score \cite{rucinski88}]\label{thm:z_statistic}
Let $J(n, p)$ be a random graph with nodes $V = \{1,\ldots,n\}$ and edges that arise independently with probability $p(n)$. Let $X_{n}$ denote the count of subgraphs of $J(n, p)$ that are isomorphic to a graph $G$. Define $\gamma := \max \{|E(G')| \big / |V(G')|: G' \subseteq G\}$, and let $\mathbb{E}(X)$ and $\var(X)$ be the expectation and variance of a random variable $X$. Then,
\begin{equation*}
    Z_{n} = \frac{X_{n} - \mathbb{E}(X_{n})}{\sqrt{\var(X_{n})}} \dto \rN(0,1),
\end{equation*}

as $n \to \infty$, if and only if $n \, (p(n))^{\gamma} \to \infty$ and $n^{2}(1 - p(n)) \to \infty$; and $\rN(0,1)$ is the standard normal.
\end{theorem}

\begin{remark}
We do not require the full strength of the result. Note that $\gamma>0$ for each of the subgraphs that we consider, where $\gamma$ is one half of the largest average degree across all subgraphs $G'$ of $G$. If we assume that $p=p(n)$ is $O(1)$ and does not go to zero with $n$, which is supported by our real-world data (Figure \ref{fig:density_WN_2013_4}), then the theorem holds. When $p = p(n)$ is constant in $n$, then $J(n, p)$ reduces to $G(n, p)$. If $p = 0$ (a set of isolated nodes) or $p = 1$ (a complete graph), there is no variation in the subgraph count, and the result does not hold.
\end{remark}

\begin{figure}\centering
    	\includegraphics[width=.65\linewidth, keepaspectratio]{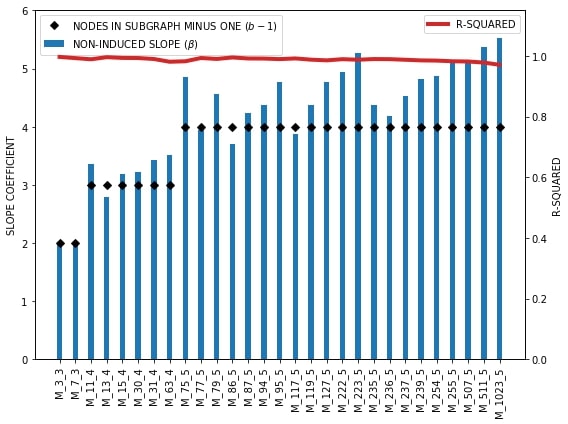}
    	\caption{A summary of log-log least squares regressions of non-induced subgraph count on number of edges, for Southwest's network, over the period 1999Q1 to 2008Q4 (40 datapoints). There is evidence that $\beta \approx b-1$ for each 3-node, 4-node and 5-node subgraph, indicating a power-law relationship between subgraph count and edges.}
    	\label{fig:log_log_slope_all_subgraphs_regime1}
\end{figure}

\begin{figure}\centering
    	\includegraphics[width=.65\linewidth, keepaspectratio]{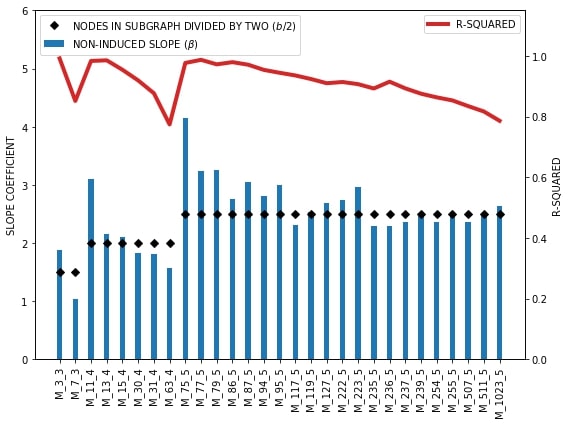}
    	\caption{A summary of log-log least squares regressions of non-induced subgraph count on number of edges, for Southwest's network, over the period 2009Q1 to 2013Q4 (20 datapoints). There is evidence that $\beta \approx b/2$ for each 3-node, 4-node and 5-node subgraph, indicating a different power-law relationship between subgraph count and edges to Figure \ref{fig:log_log_slope_all_subgraphs_regime1}.}
    	\label{fig:log_log_slope_all_subgraphs_regime2}
\end{figure}

We compute the expected number of subgraphs in $G(n, p)$ using:
\begin{equation*}
\mathbb{E}(|M_{3}^{(3)}|) = 3 \binom{n}{3} p^{2}, \quad \mathbb{E}(|M_{7}^{(3)}|) = \binom{n}{3} p^{3}, \quad \mathbb{E}(|\widetilde{M}_{3}^{(3)}|) = 3 \binom{n}{3} p^{2} (1-p),
\end{equation*}
and simulate the variance of the count by 1,000 replications from $G(n, p)$, with edge-probability $p$ set equal to the density $d(G)$ of the real network, discarding any realizations that are not connected. For each quarter in the full sample, we compute the z-score for the induced 3-star and the triangle (Figure \ref{fig:z_scores}); we include the z-score for the \emph{non-induced} 3-star for reference. It is only strictly correct to search for 3-node motifs with reference to Erd{\H{o}}s-R{\'{e}}nyi, since $G(n, p)$ only matches the number of 1-node and 2-node subgraphs (nodes and expected edges) in the real-world network. We observe that (a) the z-scores increase over time, which corresponds to a general increase in the size of the network ($n$), and they are correlated across subgraphs, (b) the non-induced 3-star is highly significant in every period, which might lead us to conclude (incorrectly) that the induced 3-star is a motif too --- this illustrates the importance of searching for induced motifs directly, (c) the triangle is a motif across the full sample, and (d) the induced 3-star is a motif from 2003 onwards.\footnote{See \cite{itzkovitz_alon05}, who study the occurrence of subgraphs in geometric network models, with nodes arranged on a lattice, and edges arising at random with a probability that decreases in the distance between nodes. Relative to Erd{\H{o}}s-R{\'{e}}nyi, they show that all subgraphs with at least as many edges as nodes (in the subgraph) will be motifs as $n \to \infty$, if the real-world and Erd{\H{o}}s-R{\'{e}}nyi networks have the same expected degree $\mathbb{E}(k)$. They give a similar result for heavy-tailed random networks.} We can interpret these results as follows:
\begin{itemize}
\item There is more clustering (triangles) in the real-world network than in $G(n, p)$, and this increases over time. While clustering coefficients indicate the higher clustering, they do not show the dynamic increase relative to $G(n, p)$ that is suggested by the triangle motif (Figure \ref{fig:clustering}).

\item There are more ``spokes'' (induced 3-stars) than in $G(n, p)$, from 2003 onwards. Nevertheless, Southwest's network has similar average path lengths to an Erd{\H{o}}s-R{\'{e}}nyi random network (Figure \ref{fig:diameter_apl}).

\end{itemize}

\begin{figure}\centering
    	\includegraphics[width=.6\linewidth, keepaspectratio]{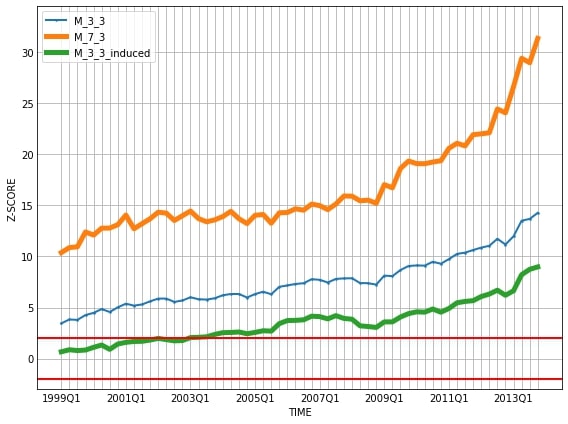}
    	\caption{The z-scores for the non-induced (thin line) and induced 3-star, and the triangle, relative to $G(n, p)$. Mean subgraph counts are computed using analytic formulae. The variance of each subgraph count is computed numerically, by 1,000 draws from $G(n, p)$. The horizontal red lines represent approximate 95\% critical values, $\pm 2$. (Section \ref{sec:motif_g_n_p})}
    	\label{fig:z_scores}
    	\vspace{1cm}
    	\includegraphics[width=.6\linewidth, keepaspectratio]{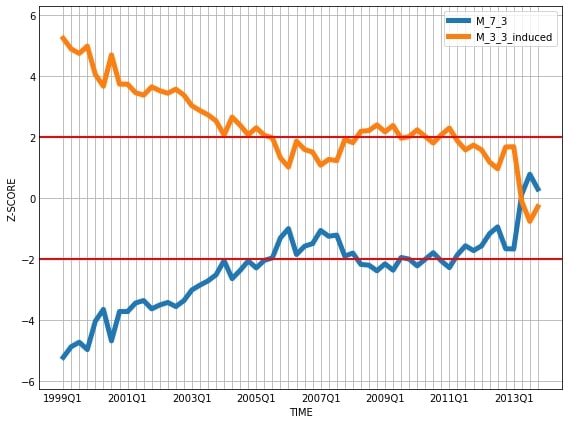}
    	\caption{The z-scores for the induced 3-star and triangle, relative to a degree-preserving rewiring of the real-world network $G$. The mean and variance of the subgraph count are computed numerically, by 1,000 draws from the randomized ensemble. Bootstrap p-values are used for inference, with 100 bootstrap replications. Since bootstrap and standard normal p-values are similar, we refer to the horizontal red lines, which represent the approximate 95\% critical values, $\pm 2$. (Section \ref{sec:motif_rewiring})}
    	\label{fig:z_scores_rewiring}
\end{figure}

\noindent While some authors consider motifs relative to $G(n, p)$, e.g., \cite{prill_etal05}, this null only matches the number of nodes and expected edges, and so we now also match the degree distribution $P(k)$ of the real-world network. Theorem \ref{thm:z_statistic} no longer applies, and the asymptotic distribution of the z-score is not generally known, so we use bootstrap p-values to assess statistical significance.

\subsubsection{3-node motifs relative to a degree-preserving rewiring}\label{sec:motif_rewiring}
In Figure \ref{fig:degree_distribution}, we plot the degree distributions of Southwest's network (kernel density estimates) and the corresponding $G(n, p)$. While $G(n, p)$ matches $n$ and the density $d(G)$ of the real-world network, it cannot generate realizations that capture the ``hub-and-spoke'' nature of the observed degree distribution $P(k)$. In this section, we use a null distribution that matches $P(k)$, by a Markov-chain degree-preserving rewiring of $G$. This method, and that described in the following sections, is based upon \cite{milo_etal02}. Starting from the observed network $G$, we select one pair of edges $(x_{1}, y_{1}) \in E$ and $(x_{2}, y_{2}) \in E$ at random, such that the nodes are all distinct, and both $(x_{1}, y_{2}) \notin E$ and $(x_{2}, y_{1}) \notin E$. Then, edges $(x_{1}, y_{1})$ and $(x_{2}, y_{2})$ are replaced by edges $(x_{1}, y_{2})$ and $(x_{2}, y_{1})$.
The edge-switching is repeated until $G$ has been sufficiently randomized. The resulting graph will have the same number of nodes $n$ and edges $m$ as the original graph, and the same degree distribution $P(k)$ but, in general, a \emph{different} topology.

In Figure \ref{fig:z_scores_rewiring}, we plot the z-scores of the induced 3-star and the triangle. Normal and bootstrap p-values give similar results, and so we refer to the $\rN(0,1)$ critical values. The results are strikingly different to those for the Erd{\H{o}}s-R{\'{e}}nyi $G(n, p)$ null in Figure \ref{fig:z_scores}. We see that:
\begin{itemize}
\item The induced 3-star is a motif from 1999--2005 and again from 2008--2011; it has become notably less significant from 2012 onwards.

\item The triangle has exactly the opposite interpretation, as an anti-motif. This follows by construction: comparing the z-score ($z_{1}$) of $|\widetilde{M}_{3}^{(3)}|$ and the z-score ($z_{2}$) of $|M_{7}^{(3)}|$, and using $|\widetilde{M}_{3}^{(3)}| = |M_{3}^{(3)}| - 3 \, |M_{7}^{(3)}|$ from Theorem \ref{thm:analytic_count_not_nested}, and the fact that $|M_{3}^{(3)}|$ is invariant to rewiring, which gives $\mathbb{E}(|M_{3}^{(3)}|) = |M_{3}^{(3)}|$ and $\var(|M_{3}^{(3)}|) = 0$, it is easy to see that $z_{1} = - z_{2}$. Hence, if one of these subgraphs is a motif, then the other will be an anti-motif; however, both subgraphs can be insignificant (not motifs) together.
\end{itemize}

\noindent Together, these results tell us that triangles (clustering) have become much more prevalent over time, while the importance of 3-stars (spokes) has decreased. This is surprising given the fall in network density over the same period (Figure \ref{fig:density_WN_2013_4}).

\subsubsection{4-node motifs relative to a degree-preserving rewiring that controls for 3-node subgraphs}\label{sec:motif_simulated_annealing}
There might be a large number of 4-node subgraphs simply because there is a significant number of 3-node subgraphs in the network. Hence, we search for 4-node motifs, controlling for the number of 1-node, 2-node and 3-node induced subgraphs. The null distribution is generated as follows, starting from the observed graph $G$. We first perform a degree-preserving rewiring as described in Section \ref{sec:motif_rewiring}, until a sufficient degree of randomness has been attained. We then use simulated annealing \cite{eglese90}, with successive edge-pair switches, to match the number of 3-node induced subgraphs to those in the original graph $G$. Simulated annealing attempts to avoid local optima by sometimes accepting a rewiring which increases the value of the optimization function.\footnote{Specifically, we minimize the function $\mathrm{Energy} = \sum_{i} |(\theta_{\mathrm{real}})_{i} - (\theta_{\mathrm{rand}})_{i}| \big / ((\theta_{\mathrm{real}})_{i} + (\theta_{\mathrm{rand}})_{i})$ by performing edge-pair switches on the already randomized graph, where the induced 3-node subgraph counts in the real and randomized (rand) data are given by $\theta_{\cdot} = (|\widetilde{M}_{3}^{(3)}|, |M_{7}^{(3)}|)^\intercal$. In our notation, we suppress the dependence of Energy and $\theta_{\cdot}$ on the current ``time'' $t$ spent in the optimization. We define the slowly-decaying \emph{temperature} function $\Psi(t+1) = \Psi(t) / \log(t+1)$, with initial value $\Psi(1)=100$. At each time step, a random edge-switch is accepted if it reduces the current Energy, and is otherwise accepted with probability $e^{-|\Delta \, \mathrm{Energy}| / \Psi(t)}$, where $\Delta \, \mathrm{Energy}$ is the difference in Energy before and after the edge-switch. One edge-switch is performed at each temperature level, and the stopping criterion is achieved when $\mathrm{Energy} < 0.00001$.} In Figure \ref{fig:z_scores_sa}, we plot the z-score for each 4-node induced subgraph. We observe that:
\begin{itemize}
    \item The induced 4-star is a strong motif for most of the sample, although it becomes less significant over time.
    
    \item The induced 4-circle and diamond are borderline motifs for much of the sample (2002--2013).
    
    \item The induced 4-path and tadpole are strong anti-motifs for the entire sample.
    
    \item The 4-complete was an anti-motif over 1999--2006 but has become progressively more significant since then, and was a borderline motif in 2013.
\end{itemize}

\noindent These results suggest that the importance of spoke airports (4-star) in Southwest's network has fallen over time, consistent with the findings of Section \ref{sec:motif_rewiring}, while clustered groups of airports (diamond and 4-complete) have gained or maintained a level of importance. In particular, the rise of the 4-complete subgraph implies that new routes have created cliques among groups of airports that were not previously completely-connected, and gives us some new insight into the decision-making that underlies network evolution (see \cite{lawford_mehmeti20} for discussion of cliques in Southwest's network). The significance of the 4-circle is rather unexpected, since travel between two opposite ``corners'' of the subgraph requires a two-step trip. The under-representation of the tadpole and 4-path makes sense, since both patterns imply two-step or perhaps even three-step trips between some of the airports in the subgraph, with no possible shortcuts (and this would be very inefficient for both the carrier and passengers).

\begin{figure}\centering
    	\includegraphics[width=.65\linewidth, keepaspectratio]{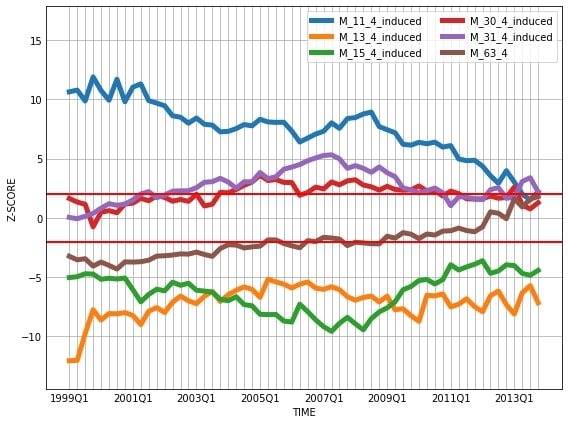}
    	\caption{The z-scores for the induced 4-star, 4-path, tadpole, 4-circle, diamond and 4-complete subgraphs, relative to a degree-preserving rewiring of the real-world network $G$, followed by a simulated annealing optimization that matches the number of induced 3-node subgraphs in $G$, in every time period. The mean and variance of the subgraph counts are computed numerically by 1,000 draws from the randomized ensemble. Bootstrap p-values are used for inference, with 100 bootstrap replications. Since bootstrap and normal p-values are similar, we refer to the horizontal red lines, which represent the approximate 95\% critical values, $\pm 2$. (Section \ref{sec:motif_simulated_annealing})}
    	\label{fig:z_scores_sa}
\end{figure}

\subsubsection{5-node motifs relative to a degree-preserving rewiring that controls for 3-node and 4-node subgraphs}\label{sec:motif_simulated_annealing_5_node}
Here, we search for 5-node motifs, controlling for the number of 1-node, 2-node, 3-node and 4-node induced subgraphs. The null distribution is generated as in Section \ref{sec:motif_simulated_annealing}, with additional control for 4-node subgraph counts. In Table \ref{tab:5_node_motifs_5_node}, we plot the z-score for each 5-node induced subgraph. We observe that:

\begin{itemize}
    \item The induced lollipop, ufo, chevron and hourglass subgraphs are the strongest candidates for motifs. There is some evidence of a transition in significance over the sample, in particular for the ufo subgraph.
    
    \item The induced banner, 5-circle, house and crown subgraphs are the most convincing anti-motifs.
\end{itemize}

Most of these results have reasonable intuitive explanations in a transportation network (cf. Table \ref{fig:motif_notation}). First, the ufo and chevron motifs suggest that some pairs of nodes will have multiple one-step or two-step paths between them. Second, the hourglass motif indicates a central ``hub'' node connected to four spoke nodes, with few short-cuts between the spokes that do not involve the hub. Interestingly, though, the addition of one edge to the hourglass gives the crown, which is an anti-motif. Third, the lollipop motif represents a triangle of nodes that are incident to a multi-step spoke. It is surprising that this is significant over the full sample, because it implies three-step paths between some pairs of airports, although shorter paths might be available if we include nodes that are not in the subgraph. Fourth, the banner, 5-circle and house anti-motifs would seem to be inefficient ways of connecting groups of five nodes, and this contrasts with the borderline 4-circle motif found in Section \ref{sec:motif_simulated_annealing}.

\subsection{Subgraph-based centrality measures}\label{sec:subgraph_centrality}
Node centrality measures are frequently used in applied work to rank nodes by their individual importance in a network. Standard measures include: (a) degree centrality, interpreted as the number of direct neighbours of node $i$, (b) closeness centrality, which characterizes the (inverse of the) average shortest path from a node $i$ to all other nodes, (c) betweenness centrality, which measures the number of times that a node $i$ acts as an ``intermediary'' in the sense of being on shortest paths between other pairs of nodes, and (d) eigenvector centrality, in which a node $i$ is more important when it is directly-connected to other more important nodes; see \cite{jackson08, jackson11}. These measures are highly correlated for many real-world and simulated networks, and thus give very similar rankings of nodes, e.g., \cite{wuchty_stadler03} report high correlations between three geometric centrality measures and the logarithm of node degree, on Erd{\H{o}}s-R{\'{e}}nyi and scale-free random graphs, and \cite{dossin_lawford17} find high linear correlations between degree, closeness, betweenness and eigenvector centralities, on unweighted and weighted real-world networks defined by the domestic route service of various U.S. airlines. The main theoretical treatment of centrality correlation is by \cite{bloch_etal16}, who argue that standard centrality measures are all characterized by the same related axioms.

In an effort to resolve the problem of high correlation, \cite{estrada_rodriguez-velazquez05} introduce \emph{subgraph centrality}
\begin{equation}\label{eq:subgraph_centrality}
    B_{S}(i) = \sum_{\tau = 0}^{\infty} \frac{(g^{\tau})_{ii}}{\tau!} = \sum_{j=1}^{n} (\nu_{j})_{i}^{2} \, e^{\lambda_{j}},
\end{equation}
for node $i$, where $\tau$ is the length of a closed walk, that can contain both cyclic and acyclic subgraphs. The second equality holds for simple graphs, where $\nu_{1},\ldots,\nu_{n}$ are the orthonormal eigenvectors of $g$, with eigenvalues $\lambda_{1},\ldots,\lambda_{n}$. Subgraph centrality measures the number of times that a node $i$ is at the start and end of closed walks of different length, with shorter lengths having greater influence.\footnote{Subgraph centrality is shown to have more discriminative power than standard centrality measures on some real-world networks by \cite{estrada_rodriguez-velazquez05}, and it is further discussed by the same authors in \cite{estrada_rodriguez-velazquez05b}.}

To illustrate, we use (\ref{eq:subgraph_centrality}) to rank nodes in Southwest's 2013Q4 network, and compare the top-ten rankings to degree centrality, in Table \ref{tab:subgraph_centrality}. We also suggest a new measure based on the number of induced subgraphs that are incident to node $i$. Formally, we define \emph{subgraph membership centrality}
\begin{equation}\label{eq:subgraph_membership_centrality}
B_{SM}\left(\widetilde{M}_{a}^{(b)};i\right)=\sum_{j \neq i}\mathbf{1}\left((i,j) \in E\left(\{\widetilde{M}_{a}^{(b)}\}\right)\right),
\end{equation}
where $\{\widetilde{M}_{a}^{(b)}\}$ is the set of all induced subgraphs of type $\widetilde{M}_{a}^{(b)}$ in the graph, and $\mathbf{1}(\cdot)$ is an indicator function. Some $B_{SM}(\cdot)$ require careful interpretation: for instance, a node might have high $B_{SM}(\widetilde{M}_{11}^{(4)})$ because it is the center of many 4-stars, or is frequently a spoke. To avoid ambiguity, we report this measure in Table \ref{tab:subgraph_centrality} for some regular subgraphs: the triangle ($M_{7}^{(3)}$), the induced 4-circle ($\widetilde{M}_{30}^{(4)}$), and the 4-complete ($M_{63}^{(4)}$). We find that:
\begin{itemize}
\item The top-ten rankings for $DC$ and $B_{S}$ and $B_{SM}(M_{7}^{(3)})$ include the same set of nodes, and are very similar. The correlations in 2013Q4 across \emph{all} nodes are 98.8\% (between $DC$ and $B_{S}$) and 98.9\% (between $DC$ and $B_{SM}(M_{7}^{(3)})$). High degree nodes are associated with more closed walks of all lengths, and with triangles.

\item The top-ten rankings for $B_{SM}(\widetilde{M}_{30}^{(4)})$ and $B_{SM}(M_{63}^{(4)})$ display several interesting differences to $DC$. First, Denver (DEN) is the top-ranked node by 4-complete membership. Second, the induced 4-circle $B_{SM}(\widetilde{M}_{30}^{(4)})$ rankings are very different to $DC$, and include eight ``new'' nodes, with Orlando (MCO) top-ranked. In 2013Q4, $DC$ and $B_{SM}(\widetilde{M}_{30}^{(4)})$ are correlated at only 59.5\%. This shows, surprisingly given the results on $B_{S}$, that different nodes can become important when one considers membership of \emph{particular} topologies.

\item The 2013Q4 correlations between degree centrality $B_{SM}$ for \emph{non-induced} subgraphs are all very high (96\%--99\%), and the latter do not give us a more informative measure here than degree centrality.

\end{itemize}

\begin{table}[!htp]
    \begin{flushleft}
        \begin{threeparttable}
            \onehalfspacing
            \caption{Comparison of top-ten node rankings in 2013Q4 by degree centrality ($DC$), Estrada and {Rodr{\'{i}}guez-Vel{\'{a}}zquez}'s subgraph centrality ($B_{S}$), and subgraph membership centrality ($B_{SM}$) based on the triangle $M_{7}^{(3)}$, the 4-circle $\widetilde{M}_{30}^{(4)}$, and the 4-complete $M_{63}^{(4)}$. Calculated values of each centrality measure are reported in parentheses. (Section \ref{sec:subgraph_centrality})}\label{tab:subgraph_centrality}
            \begin{footnotesize}
                \begin{tabular*}{\linewidth}{@{\extracolsep{\fill}}l*{5}{c}}
                    \toprule
                    & \multicolumn{5}{c}{Centrality Measure}\\
                    \cline{2-6}
                    Ranking & $DC$ & $B_{S}$ & $B_{SM}(M_{7}^{(3)})$ & $B_{SM}(\widetilde{M}_{30}^{(4)})$ & $B_{SM}(M_{63}^{(4)})$\\
                    \midrule
                    1. & MDW (0.71) & MDW (2,681,447) & MDW (342) & MCO (446) & DEN (967) \\
                    2. & LAS (0.64) & LAS (2,510,649) & LAS (335) & TPA (353) & LAS (956) \\
                    3. & DEN (0.62) & DEN (2,449,785) & DEN (333) & DAL (295) & MDW (929) \\
                    4. & BWI (0.57) & PHX (2,116,359) & PHX (298) & LAX (194) & PHX (863) \\
                    5. & PHX (0.53) & BWI (2,017,113) & BWI (277) & AUS (190) & BWI (779) \\
                    6. & HOU (0.51) & HOU (1,729,523) & HOU (246) & MKE (181) & HOU (693) \\
                    7. & MCO (0.45) & STL (1,364,403) & STL (201) & FLL (160) & STL (591) \\
                    8. & STL (0.38) & MCO (1,296,749) & BNA (183) & SAN (152) & BNA (577) \\
                    9. & BNA (0.36) & BNA (1,240,896) & MCO (171) & MDW (150) & MCI (438) \\
                    10. & TPA (0.36) & TPA (1,128,861) & TPA (158) & MCI (136) & TPA (427) \\
                    \bottomrule
                \end{tabular*}
            \end{footnotesize}
         \end{threeparttable}
    \end{flushleft}
\end{table}

\noindent To summarize, we have proposed an intuitive new subgraph-based centrality measure based on membership of particular subgraphs, that is informative for particular topologies. For instance, we find that Dallas Love Field (DAL) and Los Angeles (LAX) are very often part of 4-circle groups of airports, but are less likely (relative to other airports) to be part of completely-connected groups. The operational reasons for this structure are still unclear. These results stand in contrast to the subgraph centrality of \cite{estrada_rodriguez-velazquez05}, which is, on this dataset at least, highly correlated with $DC$.\footnote{We do not suggest that $B_{SM}$ will be more informative than $B_{S}$ or $DC$ in general, or for all subgraphs.}

\begin{landscape}
\vspace*{\fill}
\begin{table}[h]
\begin{tabular}{ccccccc}
\includegraphics[scale=0.15]{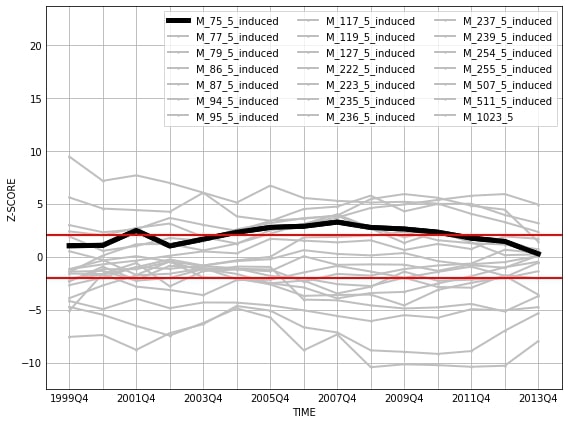} &
\includegraphics[scale=0.15]{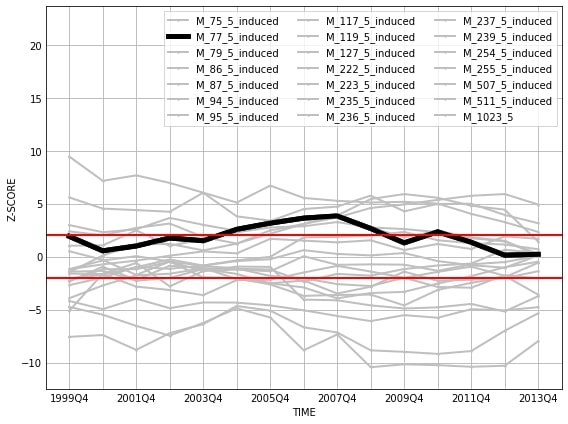} &
\includegraphics[scale=0.15]{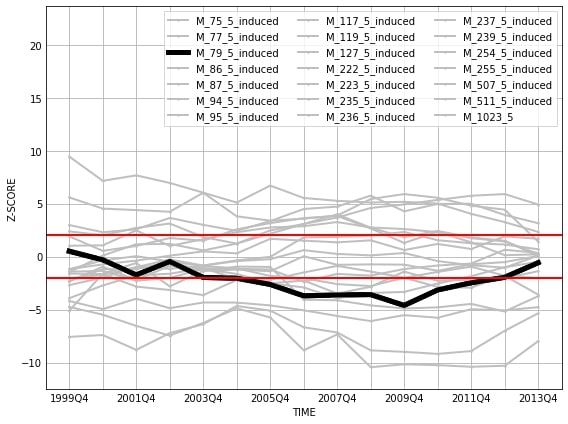} &
\includegraphics[scale=0.15]{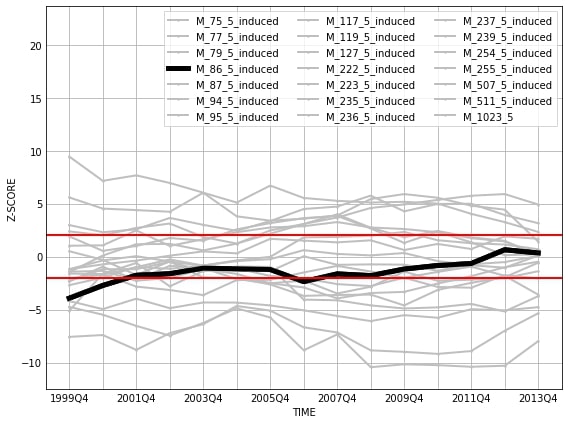} &
\includegraphics[scale=0.15]{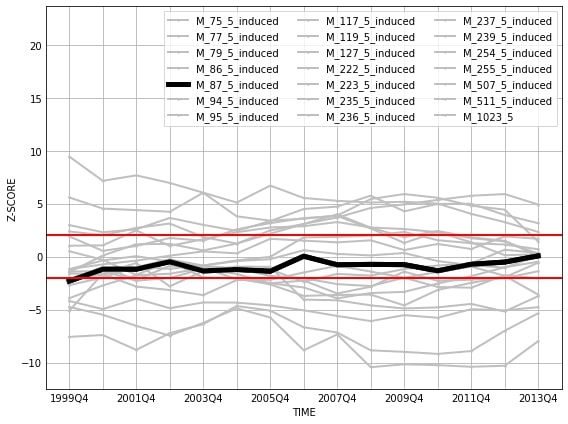} &
\includegraphics[scale=0.15]{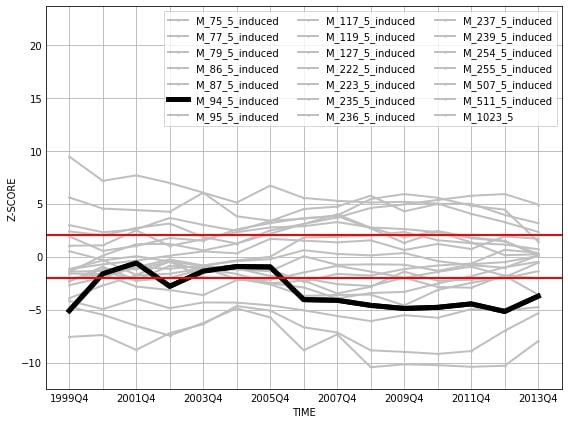} &
\includegraphics[scale=0.15]{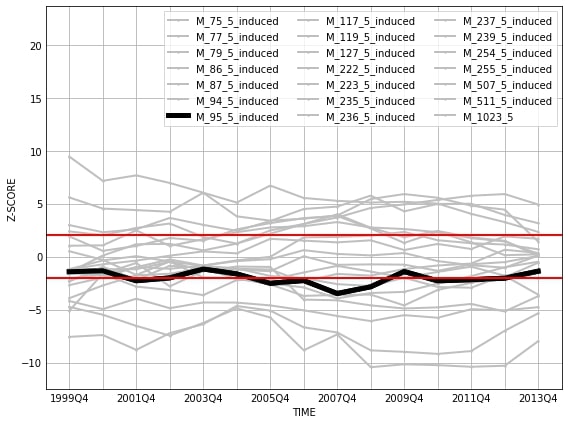} \\
5-star $M_{75}^{(5)}$ & 5-arrow $M_{77}^{(5)}$ & Cricket $M_{79}^{(5)}$ & 5-path $M_{86}^{(5)}$ & Bull $M_{87}^{(5)}$ & Banner $M_{94}^{(5)}$ & Stingray $M_{95}^{(5)}$\\
\\
\\
\includegraphics[scale=0.15]{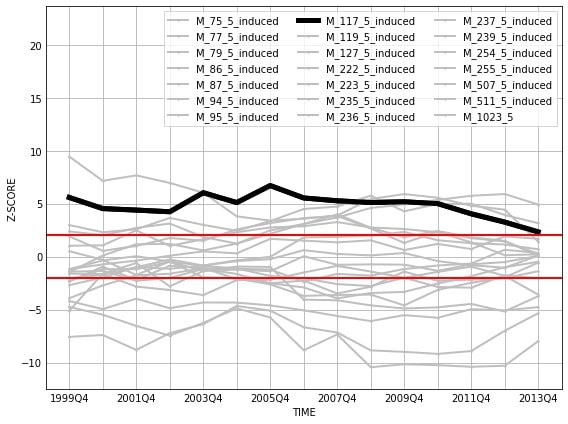} &
\includegraphics[scale=0.15]{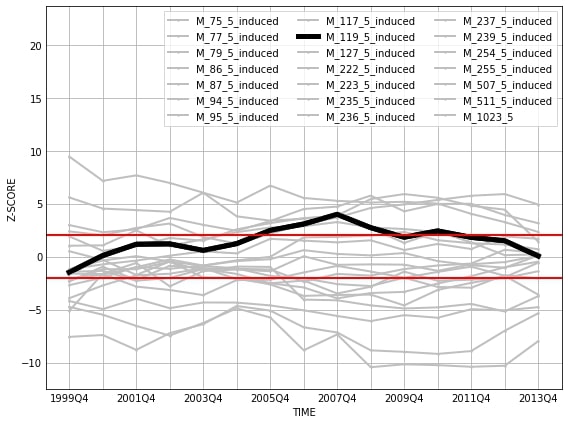} &
\includegraphics[scale=0.15]{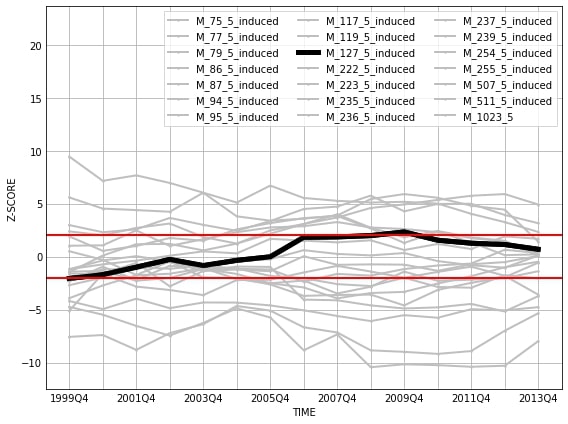} &
\includegraphics[scale=0.15]{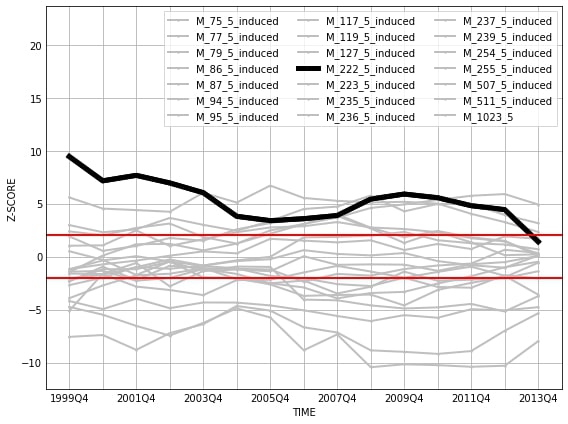} &
\includegraphics[scale=0.15]{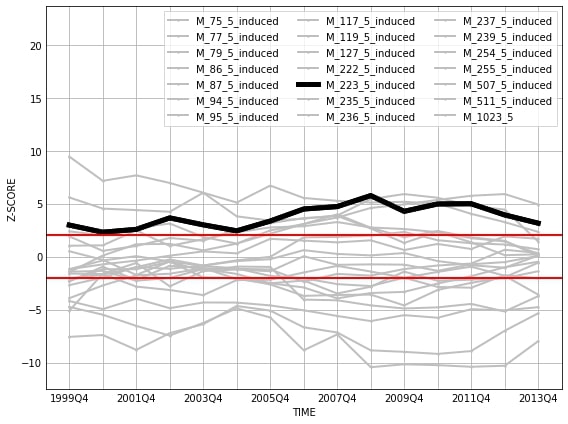} &
\includegraphics[scale=0.15]{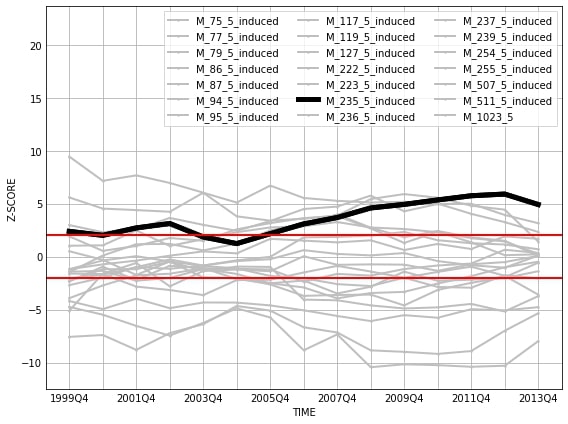} &
\includegraphics[scale=0.15]{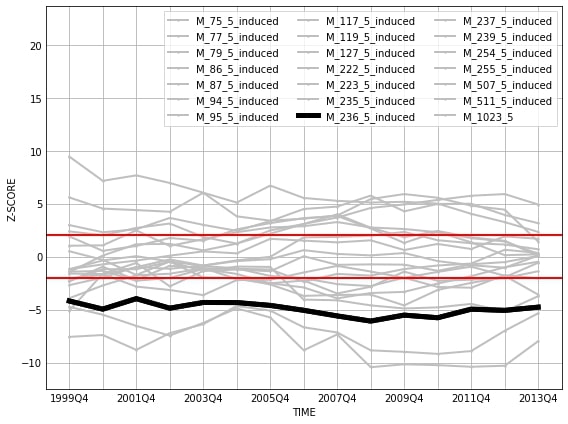} \\
Lollipop $M_{117}^{(5)}$ & Spinning top $M_{119}^{(5)}$ & Kite $M_{127}^{(5)}$ & Ufo $M_{222}^{(5)}$ & Chevron $M_{223}^{(5)}$ & Hourglass $M_{235}^{(5)}$ & 5-circle $M_{236}^{(5)}$\\
\\
\\
\includegraphics[scale=0.15]{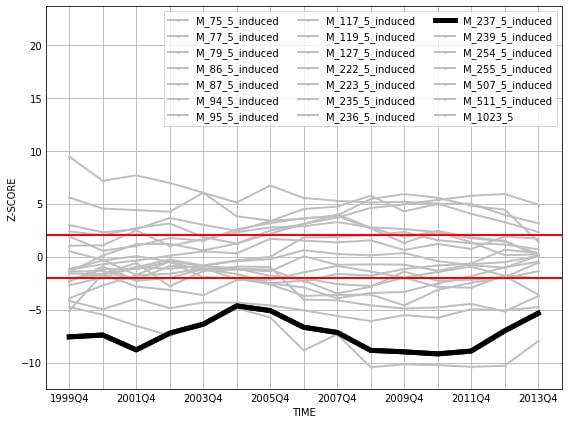} &
\includegraphics[scale=0.15]{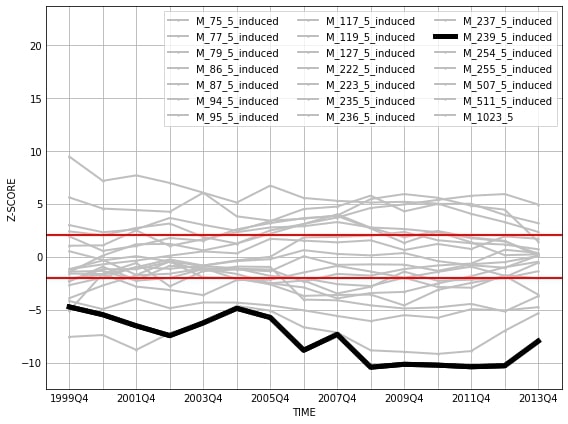} &
\includegraphics[scale=0.15]{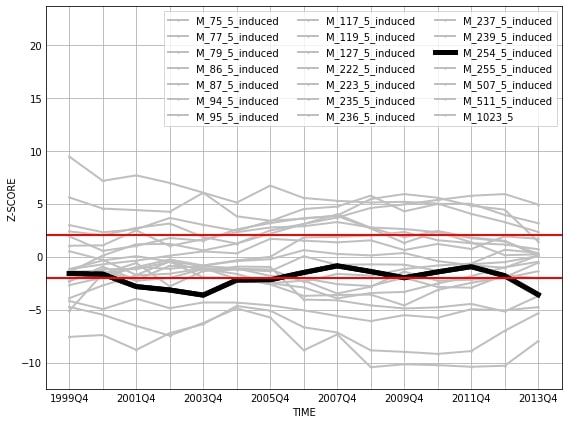} &
\includegraphics[scale=0.15]{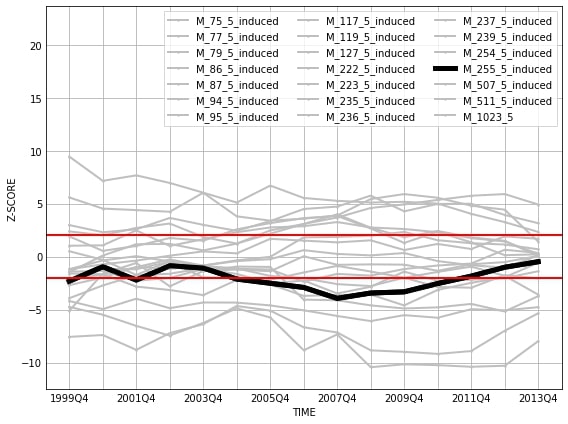} &
\includegraphics[scale=0.15]{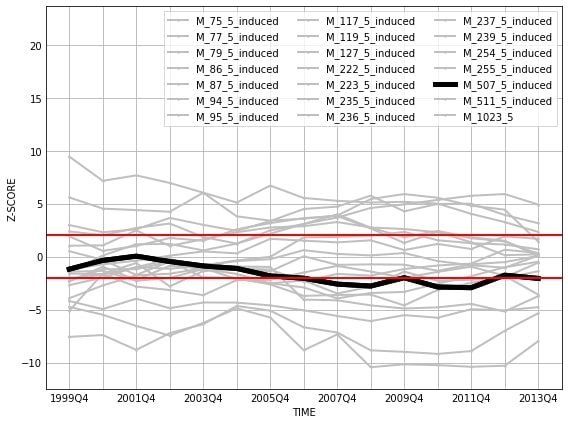} &
\includegraphics[scale=0.15]{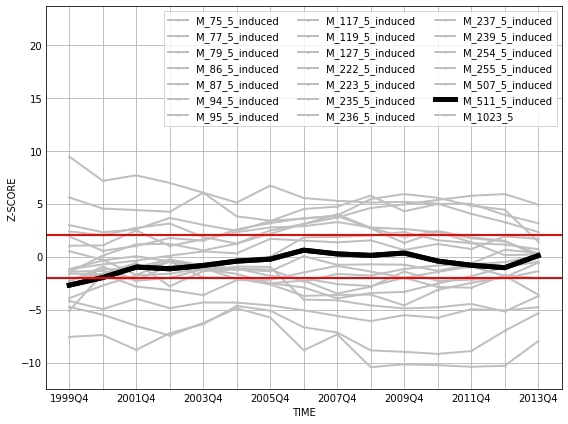} &
\includegraphics[scale=0.15]{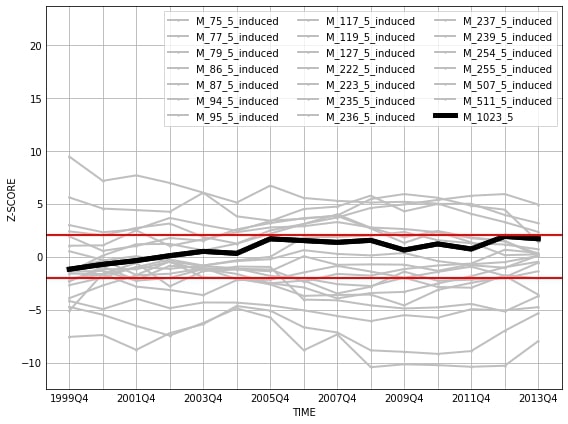} \\
House $M_{237}^{(5)}$ & Crown $M_{239}^{(5)}$ & Envelope $M_{254}^{(5)}$ & Lamp $M_{255}^{(5)}$ & Arrowhead $M_{507}^{(5)}$ & Cat's cradle $M_{511}^{(5)}$ & 5-complete $M_{1023}^{(5)}$\\
\\
\end{tabular}
\caption{The z-scores for the twenty-one induced 5-node subgraphs in Figure \ref{fig:motif_notation}, relative to a degree-preserving rewiring of the real-world network $G$, followed by a simulated annealing optimization that matches the number of induced 3-node and 4-node subgraphs in $G$, in every time period. The mean and variance of the subgraph counts are computed numerically by 1,000 draws from the randomized ensemble. We refer to the horizontal red lines, which represent the approximate 95\% critical values, $\pm 2$. (Section \ref{sec:motif_simulated_annealing_5_node}).}
\label{tab:5_node_motifs_5_node}
\end{table}
\vspace*{\fill}
\end{landscape}

\section{Related work}\label{sec:related}
We now briefly review additional research on the function, interaction, and aggregation of biological motifs, mention the little related work that exists on motifs and topology transitions in transportation networks, and report quantitative results from complex systems applied to air transportation. 

\subsection{Motifs}
\emph{Function of individual motifs.} \cite{alon07} reviews experimental work on motifs in gene regulation and other biological networks, and provides evidence that different families of motifs perform precise and identifiable information-processing functions, at a cellular level, and that these networks have an inherent structural simplicity since they are based on a limited number of basic components. \cite{hayot_jayaprakash05, mangan_alon03} show how motifs in natural networks may express evolved computational operations. Various authors have developed mathematical models for motifs, in an attempt to show how interaction patterns are related to biological function, although there is evidence that structural information may not be sufficient to distinguish between multiple potentially-useful functions in some cases \cite{ingram_etal06, isihara_etal05}.\\
\\
\emph{Interrelationships and aggregation.} It is clear that natural motifs are likely to interact, and another strand of research investigates how this will affect their function. \cite{dobrin_etal04} investigate the aggregation of motifs into clusters, in the transcriptional regulatory network of the bacterium \emph{E. coli}. They find that most individual motifs overlap (by sharing at least one link and/or node), to create \emph{homologous motif clusters}, and that these clusters will themselves merge into a \emph{motif supercluster}, which has similar properties to the whole network; they suggest that this hierarchical interaction of sets of motifs, rather than isolated function, is a general property of cellular networks. \cite{kashtan_etal04} propose topological subgraph (motif) generalizations, created by duplicating nodes (and their edges) that have the same function, which will give larger subgraphs (motifs). Formally, a pair of nodes has the same \emph{role} if there is an automorphism that maps one of the nodes to the other, and all nodes with the same role form a structurally equivalent class. For the undirected subgraphs that we consider (Table \ref{fig:motif_notation}), a node's role just corresponds to its degree \emph{within} the subgraph. There are various ways in which nodes can be duplicated, e.g., duplicating the ``center-node'' role of the 3-star gives a 4-circle, while duplicating both of the ``spoke'' nodes of the 3-star gives a 5-star. These role-preserving generalizations will, as noted by \cite{kashtan_etal04}, tend to have similar functionality to the original motif.\\
\\
\emph{Motifs in airline networks.} Although nearly all of the literature focuses on natural and technological networks, some very interesting work by Bounova \cite{bounova09} applies motifs to airline route maps and investigates topology transitions in U.S. airline networks, but on monthly data over the period 1990--2007. She finds that most airlines have similar networks, but that Southwest is topologically distinct. Using a null randomized ensemble that matches the number of nodes and the degree sequence of the real-world network, and that we would expect to give similar results to Section \ref{sec:motif_rewiring}, for \emph{3-node} subgraphs, she comes to a very different conclusion (\cite[Page 126]{bounova09}, our emphasis added):
\begin{quote}
Southwest brings a surprise in motif finding. \emph{There are no significant motifs, compared to random graphs}, though we tested a few snapshots of the airline's history (1/1990, 8/1997, 8/2007) \ldots Mathematically, this says that Southwest is no different from a random network.
\end{quote}
The z-scores reported for August 2007 in \cite[Figure 3-41]{bounova09} are very low (roughly 0.04--0.16), which contradicts our findings in Figures \ref{fig:z_scores_rewiring} and \ref{fig:z_scores_sa} and Table \ref{tab:5_node_motifs_5_node}. However, by augmenting the topological graph with departure frequency data used as edge-weights in a weighted graph (i.e., an edge is present only if the frequency is greater than some threshold), she finds  evidence that the 4-star is a motif, and remarks that ``hub-spoke motifs are only a recent phenomenon in Southwest'' \cite[Page 127]{bounova09}. The 4-star motif is in essential agreement with the results presented in Figure \ref{fig:z_scores_sa}, although we find that hub-spoke motifs (i.e., 3-star and 4-star) were significant motifs from at least 1999 to 2005 and have, if anything, become \emph{less} significant over time. These differences may be due to (a) use of a different dataset and/or data treatment and filtering steps, and (b) sensitivity to the choice of null random ensemble.

\subsection{Complex systems applied to air transportation}
There is a substantial literature on the application of complex systems to airline networks. See \cite{lin_ban14}, \cite[Table I]{lordan_etal14} and \cite{roucolle_etal20} for nice surveys. \cite{wuellner_etal10} investigate the resilience of airline networks to a targeted removal of nodes and a random removal of edges, and find that graph connectedness and ``travel times'' (based on spatial geodesic lengths and intermediate airport penalties) are generally preserved. The $k$\emph{-core} is defined by iterative removal of nodes (and their edges) with degree less than $k$, until all nodes have degree greater than or equal to $k$ (the final network is called the $k$-core). Using data for 2007, the authors find that Southwest is a special case, with a large $k$-core structure and extreme resilience to node or edge deletion, and conclude that (\cite[Page 056101-1]{wuellner_etal10}, \{.\} our addition): ``\{Southwest\} has essentially built a core network, comprising more than half of its overall destinations, which is a dense mesh of interconnected high-degree (i.e., ``hub'') airports.'' They also report \cite[Table I]{wuellner_etal10} an average path length of 1.542, that is slightly lower than the average path lengths reported in Figure \ref{fig:diameter_apl}. Two related papers \cite{verma_etal14, verma_etal16} analyze the core of the World Airline Network (WAN). This network is made up of more than 3,200 nodes and 18,000 edges. However, unlike the results reported above for Southwest, they find that the WAN has a very small core (containing about 2.5\% of all nodes), that is almost fully connected and is surrounded by a nearly tree-like periphery; upon removal of the core, they find that most of the WAN network is still connected.

\section{Conclusions}\label{sec:conclusions}
We have used exact methods to explore the behaviour of small subgraphs in a dynamic transportation network defined by the route service of Southwest Airlines. The topology has much in common with random graphs, exhibiting ``small-world'' characteristics.\footnote{On the other hand, similarities between Southwest's monthly network and Erd{\H{o}}s-R{\'{e}}nyi were noticed by \cite[e.g., Figure 4]{bounova_weck12}, who describe patterns of linear correlation (``heat maps'') between pairs of graph topology metrics.}  Short path lengths and clustering reflect the need for a carrier to provide efficient service, with few connections between airport pairs. We present new evidence of a regime-switch in power-law scaling between subgraph counts and the number of edges in the network. In a sense, the emergence of any scaling is curious, because transportation networks result from careful route-level planning and strategic decision-making, based on the spatial distribution of demand and competition, and subject to operational and regulatory constraints inherent in providing passenger service (e.g., availability of fleet and crew, scheduling, legal restrictions such as the Wright Amendment of 1979, etc.). The network has evolved by design, from an initial state given by Southwest's network when the U.S. air transport sector was deregulated by the Airline Deregulation Act of 1978, and not at random. Nevertheless, macroscopic regularity emerges from this microscopic behaviour.

We identify motifs and anti-motifs on three, four and five nodes, some of which present substantial dynamic variation, and have a rather different interpretation to motifs in natural networks. Our results on topology evolution provide new insights into the structure of a transportation network, \emph{that are not observable by standard measures such as node centrality and clustering}. In particular, Southwest's network has become less ``starlike'' over time, despite a fall in network density, but also favours unexpected local structure (e.g., circle, diamond). We illustrate how a simple new subgraph-based centrality measure can be used to identify important nodes based on membership of specific topological structures. Subgraph-based tools are useful in giving new qualitative and quantitative understanding of the behaviour of real-world networks, and can potentially be used as diagnostic tools for economic or mathematical models of network evolution.\\
\\
\noindent \emph{Directions for future work.} By focusing on small motifs, we are able to remain within an exact analytic framework for subgraph enumeration. The primary limitation of this method is that it will not be applicable to larger motifs, of the size that are regularly considered in biology (e.g., 10 nodes and above), when the very rapid increase in the number of possible topological subgraphs necessitates the use of very fast approximate sampling methods. Some objectives for extensions of our results include (1) investigating whether our results apply to other economic or transportation networks, and exploring ways to characterize the strategic behaviour of different networks based on their topology; (2) using subgraph-based methods and econometrics, possibly incorporating edge-weights or information on the spatial location of nodes into the graph, to explain the observed strategic and dynamic decisions on market entry and exit in an economic network; (3) developing a better understanding of which classes of theoretical and real-world graph models give rise to the scaling behaviour that we have observed, e.g., \cite{barrat_etal04} report power-law decay, as a function of node degree, in the degree distribution, in the total (and average) traffic handled by each airport, and in the average clustering coefficient, using data on the World Airline Network, while \cite{angelesserrano_etal08, song_etal05} use renormalization to show that scale-free and small-world behaviour can arise naturally in real complex networks that are invariant/self-similar under length-scale transformations; and (4) applying state-of-the-art computational tools to search for larger motifs in such networks. We leave these interesting problems for future work.

\clearpage

\appendix

\counterwithin{figure}{section}

\section{Proofs}\label{sec:proofs}

\begin{theorem}[Analytic formulae for non-induced subgraph enumeration on three and four nodes \cite{alon_etal97}]\label{thm:analytic_count}
\begin{equation}\label{eq:m_3_3}
|M_{3}^{(3)}| = \sum_{i}\binom{k_{i}}{2} = \frac{1}{2} \, \sum_{i}k_{i}(k_{i}-1).
\end{equation}
\begin{equation}\label{eq:m_7_3}
|M_{7}^{(3)}| = \frac{1}{6} \, \tr \, (g^{3}).
\end{equation}
\begin{equation}\label{eq:m_11_4}
|M_{11}^{(4)}| = \sum_{i}\binom{k_{i}}{3} = \frac{1}{6} \, \sum_{i}k_{i}(k_{i}-1)(k_{i}-2).
\end{equation}
\begin{equation}\label{eq:m_13_4}
|M_{13}^{(4)}| = \sum_{(i, j) \in E}(k_{i}-1)(k_{j}-1) - 3 \, |M_{7}^{(3)}|.
\end{equation}
\begin{equation}\label{eq:m_15_4}
|M_{15}^{(4)}| = \frac{1}{2} \sum_{i}(g^{3})_{ii} \, (k_{i}-2).
\end{equation}
\begin{equation}\label{eq:m_30_4}
|M_{30}^{(4)}| = \frac{1}{8} \, (\tr \, (g^{4}) - 4 \, |M_{3}^{(3)}| - 2 \, m).
\end{equation}
\begin{equation}\label{eq:m_31_4}
|M_{31}^{(4)}| = \frac{1}{2} \, \sum_{i,j}\binom{(g^{2})_{ij} \, (g)_{ij}}{2} = \frac{1}{2} \, \sum_{(i,j) \in E}(g^{2})_{ij}((g^{2})_{ij} - 1).
\end{equation}
\begin{equation}\label{eq:m_63_4}
|M_{63}^{(4)}| = \frac{1}{24} \sum_{i} \tr \, (g_{-i}^{3}).
\end{equation}
\end{theorem}

\begin{proof}[Proof of Theorem \ref{thm:analytic_count}] Full proofs are given in \cite{alon_etal97}. To give a self-contained treatment in this paper, we provide elementary combinatorial proofs of each count formula. For a full treatment of exact enumeration of \emph{five-node} subgraphs, see \cite{lawford20}. We denote the binomial coefficient $\binom{n}{r}$ and $\tr \, (g)$ is the trace of a square matrix.

\begin{enumerate}[label=(\alph*)]

    \item $|M_{3}^{(3)}|$: Node $i$ has edges to $k_{i}$ neighbours, and any pair of those edges will form a star, centered on $i$. The result (\ref{eq:m_3_3}) follows immediately. See also \cite[$n_{G}(H_{2})$]{alon_etal97}. In general, it is straightforward to count the number of $b$-node stars in a graph using $\sum_{i}\binom{k_{i}}{b-1}$, where a summand is set to zero if $b-1 > k_{i}$.
    
    \item $|M_{7}^{(3)}|$: The elements of $g^{3}$ are the number of walks of length 3 from node $i$ to node $j$, and so $\tr \, (g^{3})$ gives the total number of closed walks of length 3 in $G$, each of which must involve three distinct nodes $i$, $j$ and $x$. Since there are six ways to traverse a given triangle (starting at any corner, and moving clockwise or counterclockwise), e.g., $\{(i, j), (j, x), (x, i)\}$, we divide by six to give result (\ref{eq:m_7_3}). See also \cite[$n_{G}(C_{3})$]{alon_etal97}.
    
    \item $|M_{11}^{(4)}|$: Result (\ref{eq:m_11_4}) follows directly from the generalization of the argument used for $|M_{3}^{(3)}|$. See \cite[$n_{G}(H_{4})$]{alon_etal97}.
    
    \item $|M_{13}^{(4)}|$: Consider any edge $(i, j) \in E$, as the central edge in a 4-path $\{(x, i), (i, j), (j, y)\}$. Node $i$ has $k_{i} - 1$ possible neighbours (for node $x$), and node $j$ has $k_{j} - 1$ possible neighbours (for node $y$). There are $(k_{i} - 1)(k_{j} - 1)$ ways in which a neighbour of $i$ can be paired with a neighbour of $j$, which gives a total of $\sum_{(i, j) \in E}(k_{i}-1)(k_{j}-1)$ across all possible central edges. This sum includes the unwanted case $x = y$, which forms the triangle with corners $i$, $j$ and $x$. Since any of the three edges of a given triangle can be a candidate central edge $(i, j)$ of a 4-path, we subtract $3 \, |M_{7}^{(3)}|$ to give result (\ref{eq:m_13_4}). See also \cite[$n_{G}(H_{3})$]{alon_etal97}.
    
    \item $|M_{15}^{(4)}|$: The tadpole subgraph can be thought of as a triangle on nodes $i$, $j$ and $x$, with the addition of an extra edge $(i, y)$, where $k_{i} > 2$. The element $(1/2) \, (g^{3})_{ii}$ is the number of triangles incident to node $i$, where the division by two corrects for double-counting due to the two possible directions of travel around a given triangle. Hence, there are $(1/2) (g^{3})_{ii} \, (k_{i} - 2)$ tadpoles ``centered on'' node $i$. Result (\ref{eq:m_15_4}) follows immediately. See also \cite[$n_{G}(H_{5})$]{alon_etal97}.
    
    \item $|M_{30}^{(4)}|$: The elements of $g^{4}$ are the number of walks of length 4 from node $i$ to node $j$, and so $\tr \, (g^{4})$ gives the total number of closed walks of length 4 in $G$. We proceed to prove (\ref{eq:m_30_4}) indirectly, by expressing $\tr \, (g^{4})$ in terms of the number of circles and other walks of length 4 on a circle. Consider four distinct nodes $i$, $j$, $x$ and $y$, and the circle $M_{30}^{(4)}$ with edges $\{(i, j), (j, x), (x, y), (y, i)\}$. There are \emph{eight} ways to traverse this circle (starting at any corner, and moving clockwise or counterclockwise). However, there are two additional ways to walk from one of these nodes and back to itself in four steps, using only the four edges of the circle:
    
    \begin{itemize}
        \item First, there are four possible 3-stars $M_{3}^{(3)}$, i.e., (I) $\{(i, j), (j, x)\}$, (II) $\{(j, x), (x, y)\}$, (III) $\{(x, y), (y, i)\}$, and (IV) $\{(y, i), (i, j)\}$. Starting from node $i$, it is possible to build three of these: (I), (III), and (twice) (IV). Across the four nodes, each of (I)--(IV) will appear \emph{four} times.
        
        \item Second, there are four edges in the circle. Starting from node $i$, it is possible to build two of these: $(i, j)$ and $(i, y)$. Across the four nodes, each edge in the circle will appear \emph{twice}.
    \end{itemize}
    
    Hence, we can write $\tr \, (g^{4}) = 8 |M_{30}^{(4)}| + 4 |M_{3}^{(3)}| + 2m$, and result (\ref{eq:m_30_4}) follows directly. See also \cite[$n_{G}(C_{4})$]{alon_etal97}.
    
    \item $|M_{31}^{(4)}|$: We can think of a diamond on nodes $i$, $j$, $x$, and $y$ as two distinct triangles with a common edge $(i, j)$. Given this common edge, $(g^{2})_{ij} \, (g)_{ij}$ represents the number of walks of length 2 between $i$ and $j$, i.e., the number of distinct triangles in $G$ that contain $(i, j)$. A diamond is formed by any two of these triangles, and so $\binom{(g^{2})_{ij} \, (g)_{ij}}{2}$ gives the number of distinct diamonds that can be built from a common edge $(i, j)$. Summing across all pairs of nodes $i$ and $j$ will give twice the number of diamonds in $G$, since the edge $(i, j)$ has two endpoints. We divide the sum by two to give result (\ref{eq:m_31_4}). See also \cite[$n_{G}(H_{6})$]{alon_etal97}.
    
    \item $|M_{63}^{(4)}|$: Consider a triangle subgraph $M_{7}^{(3)}$ comprised of nodes $j$, $x$ and $y$. Let each node be in the neighbourhoood $\Gamma_{G}(i)$ of some node $i$ such that $i \neq j \neq x \neq y$. Hence, the four nodes $i$, $j$, $x$ and $y$, and the edges between them, form a 4-complete subgraph $M_{63}^{(4)}$. The quantity $(1/6) \tr \, (g_{-i}^{3})$ gives the number of 4-complete subgraphs that are incident to node $i$, where $g_{-i}$ is the adjacency matrix formed by the neighbourhood of $i$. By symmetry, summing across all nodes $i$ will give four times the total count of 4-complete subgraphs in the graph, and so we divide the sum by four to give result (\ref{eq:m_63_4}). See also \cite[Page 222]{alon_etal97}.
    
\end{enumerate}
\end{proof}

\begin{theorem}[Analytic formulae for induced subgraph enumeration on three and four nodes]\label{thm:analytic_count_not_nested}

\begin{flalign*}
|\widetilde{M}_{3}^{(3)}| &= |M_{3}^{(3)}| - 3 \, |M_{7}^{(3)}|.&\\
|\widetilde{M}_{11}^{(4)}| &= |M_{11}^{(4)}| - |M_{15}^{(4)}| + 2 \, |M_{31}^{(4)}| - 4 \, |M_{63}^{(4)}|.&\\
|\widetilde{M}_{13}^{(4)}| &= |M_{13}^{(4)}| - 2 \, |M_{15}^{(4)}| - 4 \, |M_{30}^{(4)}| + 6 \, |M_{31}^{(4)}| - 12 \, |M_{63}^{(4)}|.&\\
|\widetilde{M}_{15}^{(4)}| &= |M_{15}^{(4)}| - 4 \, |M_{31}^{(4)}| + 12 \, |M_{63}^{(4)}|.&\\
|\widetilde{M}_{30}^{(4)}| &= |M_{30}^{(4)}| - |M_{31}^{(4)}| + 3 \, |M_{63}^{(4)}|.&\\
|\widetilde{M}_{31}^{(4)}| &= |M_{31}^{(4)}| - 6 \, |M_{63}^{(4)}|.&
\end{flalign*}

\end{theorem}

\begin{proof}[Proof of Theorem \ref{thm:analytic_count_not_nested}] It is well-known that induced subgraph counts $\widetilde{y}=(|\widetilde{M}_{a}^{(b)}|)$ can be obtained from non-induced subgraph counts $y=(|M_{a}^{(b)}|)$ by a linear combination $\widetilde{y} = y - A\, \widetilde{y}$, where $A = (A)_{ij}$ gives the number of copies of each subgraph $i$ in subgraph $j$, and $(A)_{ii} = 0$. The elements of $A$ can be computed directly using Theorem \ref{thm:analytic_count}. Then, $\widetilde{y} = (I+A)^{-1}\,y$, where invertibility of $I+A$ follows from the properties of a unit upper-triangular matrix. Here we present an alternative combinatorial proof. Each induced subgraph count of $H'$ is found by considering all subgraphs $H$ on the same set of $b$ nodes as $H'$, such that $H' \subset H$, and noting the number of ways in which edges can be removed from $H$ to obtain $H'$. For example, the 4-star $M_{11}^{(4)}$ is nested in the tadpole $M_{15}^{(4)}$, the diamond $M_{31}^{(4)}$ and the 4-complete $M_{63}^{(4)}$ (Figure \ref{fig:nesting_subgraphs}). We treat each subgraph separately but, for convenience of exposition, not in order.

\begin{figure}\centering
    	\begin{subfigure}{0.24\textwidth}
    		\centering
    		\includegraphics[width=.6\linewidth, keepaspectratio]{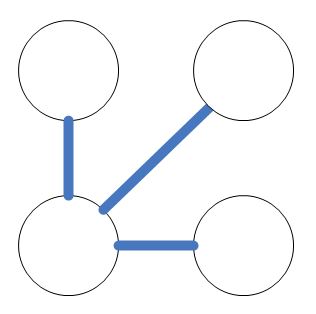}
    		\caption{4-star $M_{11}^{(4)}$.}
    		\label{fig:m_11_4_nested}
    	\end{subfigure}
    	\begin{subfigure}{0.24\textwidth}
    		\centering
    		\includegraphics[width=.6\linewidth, keepaspectratio]{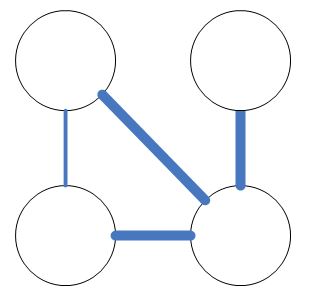}
    		\caption{Tadpole $M_{15}^{(4)}$.}
    		\label{fig:m_15_4_nested}
    	\end{subfigure}
    		\begin{subfigure}{0.24\textwidth}
    		\centering
    		\includegraphics[width=.6\linewidth, keepaspectratio]{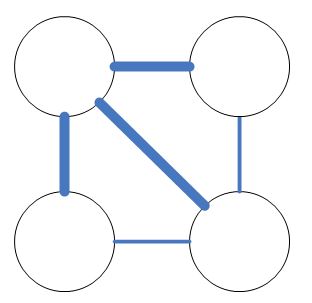}
    		\caption{Diamond $M_{31}^{(4)}$.}
    		\label{fig:m_31_4_nested}
    	\end{subfigure}
    	    \begin{subfigure}{0.24\textwidth}
    		\centering
    		\includegraphics[width=.6\linewidth, keepaspectratio]{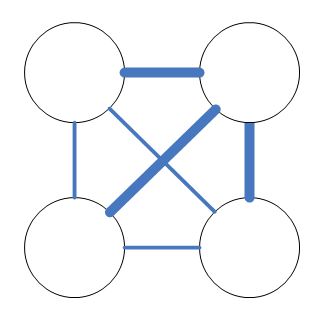}
    		\caption{4-complete $M_{63}^{(4)}$.}
    		\label{fig:m_63_4_nested}
    	\end{subfigure}
    	\caption{Illustrative nestings of the 4-star subgraph. There is one way in which one edge can be removed from the tadpole to create a 4-star, two ways in which two edges can be removed from the diamond, and four ways in which three edges can be removed from the 4-complete.}
    	\label{fig:nesting_subgraphs}
    \end{figure}

\begin{enumerate}[label=(\alph*)]

    \item $|\widetilde{M}_{3}^{(3)}|$: To create a 3-star, there are three ways to remove one edge from the triangle.
    
    \item $|\widetilde{M}_{31}^{(4)}|$: To create a diamond, there are six ways to remove one edge from the 4-complete.
    
    \item $|\widetilde{M}_{30}^{(4)}|$: To create a 4-circle, there is one way to remove one edge from the diamond, and three ways to remove two edges (with no common nodes) from the 4-complete. Hence,
    \begin{equation*}
        |\widetilde{M}_{30}^{(4)}| = |M_{30}^{(4)}| - |\widetilde{M}_{31}^{(4)}| - 3 \, |M_{63}^{(4)}|.
    \end{equation*}
    
    \item $|\widetilde{M}_{15}^{(4)}|$: To create a tadpole, there are four ways to remove one edge from the diamond, and twelve ways to remove two edges (with a common node) from the 4-complete. Hence,
    \begin{equation*}
        |\widetilde{M}_{15}^{(4)}| = |M_{15}^{(4)}| - 4 \, |\widetilde{M}_{31}^{(4)}| - 12 \, |M_{63}^{(4)}|.
    \end{equation*}
    
    \item $|\widetilde{M}_{13}^{(4)}|$: To create a 4-path, there are two ways to remove one edge from a tadpole, four ways to remove one edge from a 4-circle, six ways to remove two edges from a diamond, and twelve ways to remove three edges from a 4-complete. So,
    \begin{equation*}
        |\widetilde{M}_{13}^{(4)}| = |M_{13}^{(4)}| - 2 \, |\widetilde{M}_{15}^{(4)}| - 4 \, |\widetilde{M}_{30}^{(4)}| - 6 \, |\widetilde{M}_{31}^{(4)}| - 12 \, |M_{63}^{(4)}|.
    \end{equation*}
    
    \item $|\widetilde{M}_{11}^{(4)}|$: As illustrated in Figure \ref{fig:nesting_subgraphs}, to create a 4-star, there is one way to remove one edge from a tadpole, two ways to remove two edges from a diamond, and four ways to remove three edges (a triangle) from the 4-complete. Then,
    \begin{equation*}
        |\widetilde{M}_{11}^{(4)}| = |M_{11}^{(4)}| - |\widetilde{M}_{15}^{(4)}| - 2 \, |\widetilde{M}_{31}^{(4)}| - 4 \, |M_{63}^{(4)}|.
    \end{equation*}
    
\end{enumerate}

\end{proof}

\newpage

\section{Additional Tables and Figures}\label{sec:additional_figures}

\setcounter{table}{0}
\renewcommand{\thetable}{B.\arabic{table}}

\counterwithin{figure}{section}

\begin{landscape}
\vspace*{\fill}
\begin{table}[h]
\begin{tabular}{cccccccc}
\includegraphics[height=0.65in]{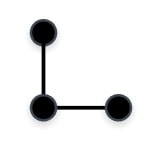} &
\includegraphics[height=0.65in]{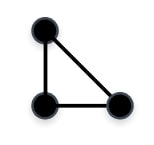} &
\includegraphics[height=0.65in]{M_11_4} &
\includegraphics[height=0.65in]{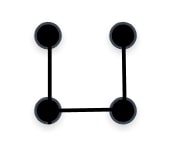} &
\includegraphics[height=0.65in]{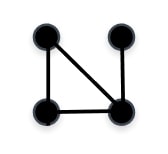} &
\includegraphics[height=0.65in]{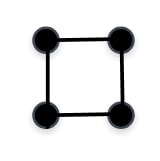} &
\includegraphics[height=0.65in]{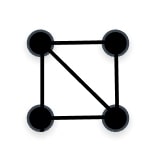} &
\includegraphics[height=0.65in]{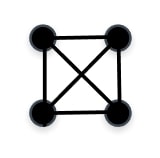}\\
3-star $M_{3}^{(3)}$ & Triangle $M_{7}^{(3)}$ & 4-star $M_{11}^{(4)}$ & 4-path $M_{13}^{(4)}$ & Tadpole $M_{15}^{(4)}$ & 4-circle $M_{30}^{(4)}$ & Diamond $M_{31}^{(4)}$ & 4-complete $M_{63}^{(4)}$\\
\\
\\
\includegraphics[height=0.65in]{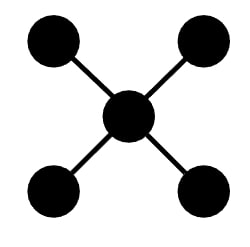} &
\includegraphics[height=0.65in]{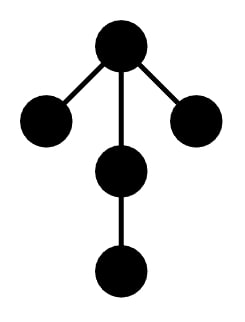} &
\includegraphics[height=0.65in]{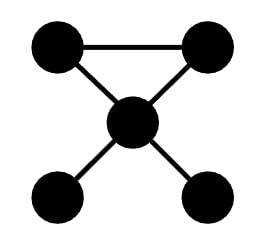} &
\includegraphics[height=0.65in]{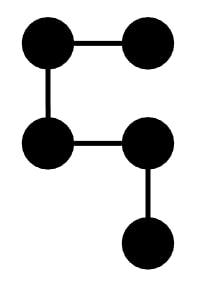} &
\includegraphics[height=0.65in]{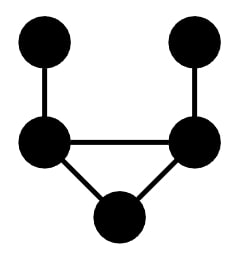} &
\includegraphics[height=0.65in]{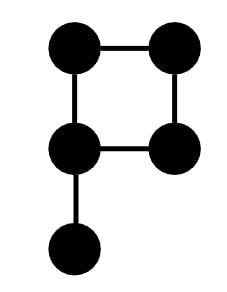} &
\includegraphics[height=0.65in]{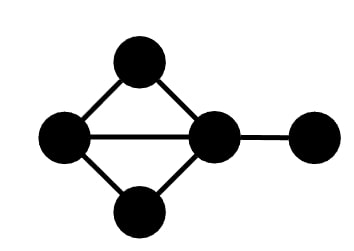} &
\includegraphics[height=0.65in]{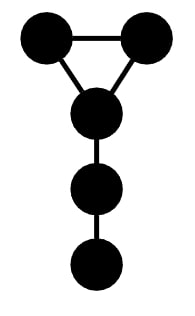}\\
5-star $M_{75}^{(5)}$ & 5-arrow $M_{77}^{(5)}$ & Cricket $M_{79}^{(5)}$ & 5-path $M_{86}^{(5)}$ & Bull $M_{87}^{(5)}$ & Banner $M_{94}^{(5)}$ & Stingray $M_{95}^{(5)}$ & Lollipop $M_{117}^{(5)}$\\
\\
\\
\includegraphics[height=0.65in]{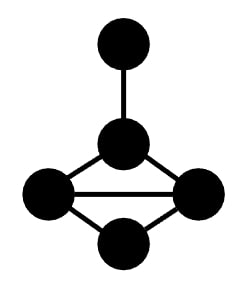} &
\includegraphics[height=0.65in]{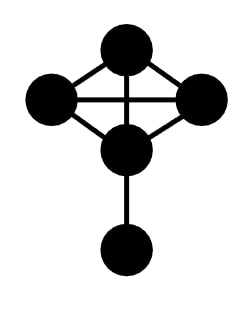} &
\includegraphics[height=0.65in]{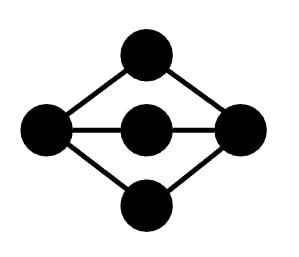} &
\includegraphics[height=0.65in]{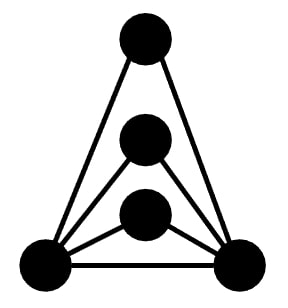} &
\includegraphics[height=0.65in]{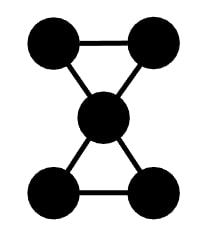} &
\includegraphics[height=0.65in]{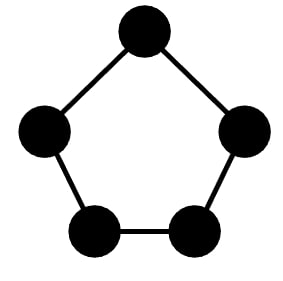} &
\includegraphics[height=0.65in]{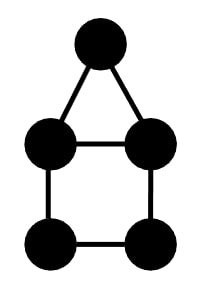} &
\includegraphics[width=0.65in]{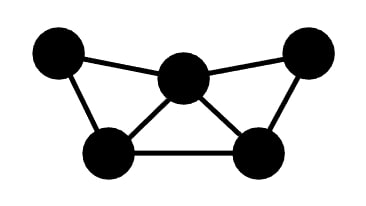}\\
Spinning top $M_{119}^{(5)}$ & Kite $M_{127}^{(5)}$ & Ufo $M_{222}^{(5)}$ & Chevron $M_{223}^{(5)}$ & Hourglass $M_{235}^{(5)}$ & 5-circle $M_{236}^{(5)}$ & House $M_{237}^{(5)}$ & Crown $M_{239}^{(5)}$\\
\\
\\
\includegraphics[height=0.65in]{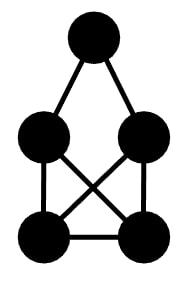} &
\includegraphics[height=0.65in]{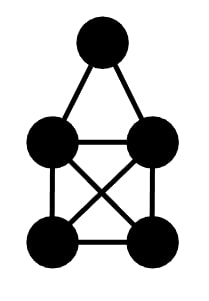} &
\includegraphics[height=0.65in]{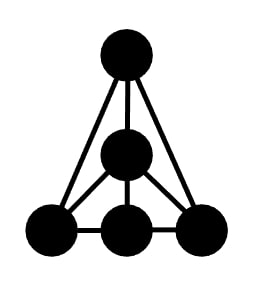} &
\includegraphics[height=0.65in]{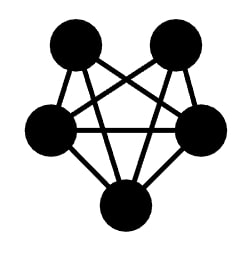} &
\includegraphics[height=0.65in]{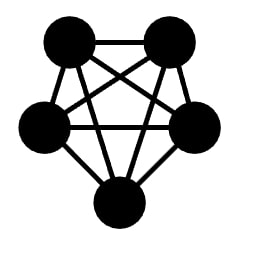} &
&
&
\\
Envelope $M_{254}^{(5)}$ & Lamp $M_{255}^{(5)}$ & Arrowhead $M_{507}^{(5)}$ & Cat's cradle $M_{511}^{(5)}$ & 5-complete $M_{1023}^{(5)}$ & & &\\
\\
\end{tabular}
\caption{The twenty-nine undirected connected subgraphs on three, four, or five nodes, denoted using the notation of \cite{lawford20, lawford_mehmeti20}.}
\label{fig:motif_notation}
\end{table}
\vspace*{\fill}
\end{landscape}

\begin{landscape}
\vspace*{\fill}
\begin{table}[h]
\begin{tabular}{cccccccc}
\includegraphics[height=0.65in]{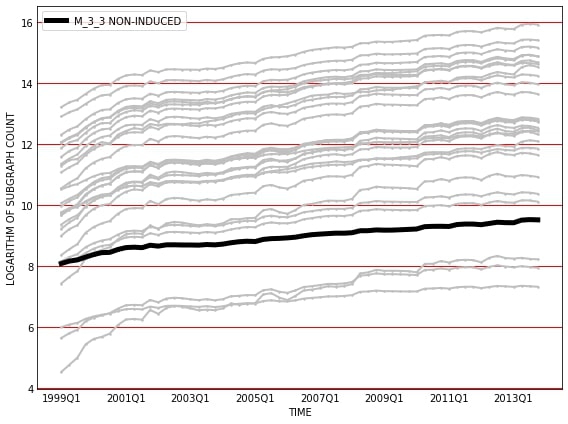} &
\includegraphics[height=0.65in]{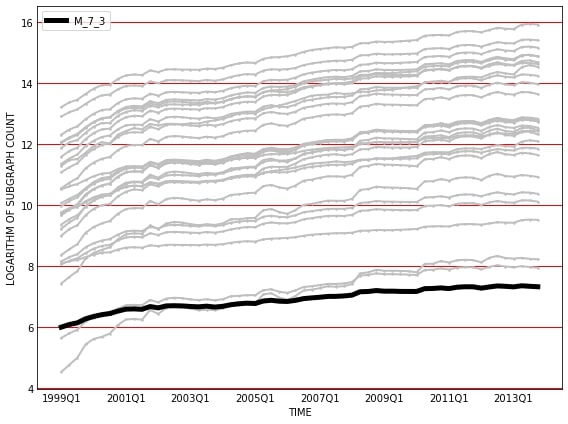} &
\includegraphics[height=0.65in]{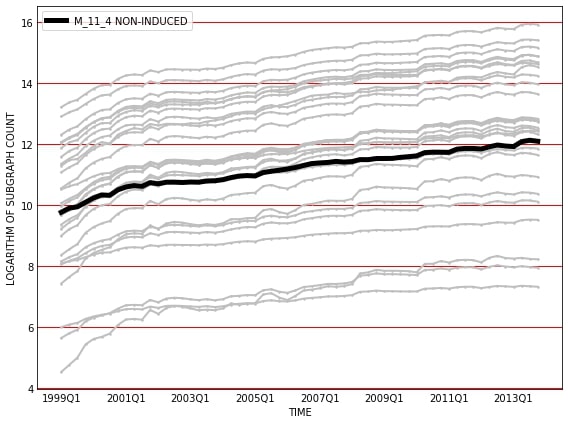} &
\includegraphics[height=0.65in]{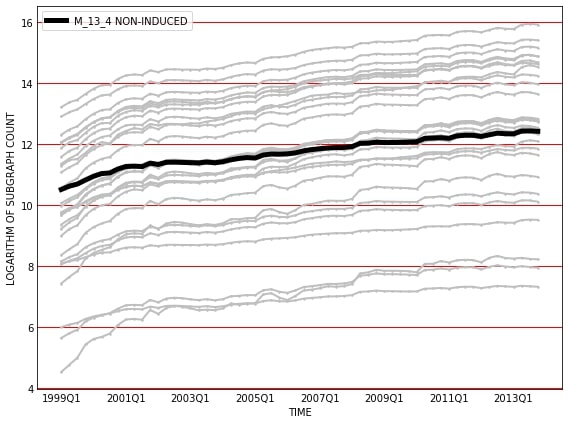} &
\includegraphics[height=0.65in]{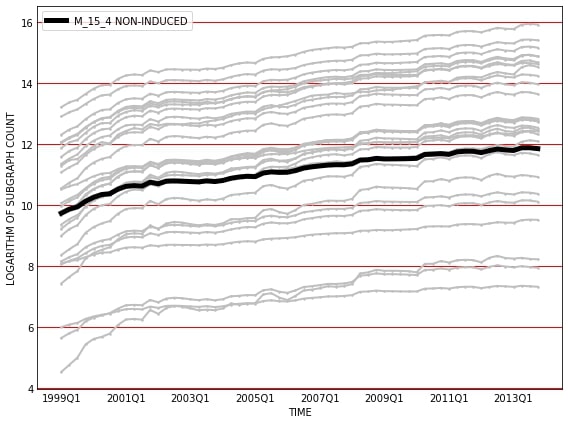} &
\includegraphics[height=0.65in]{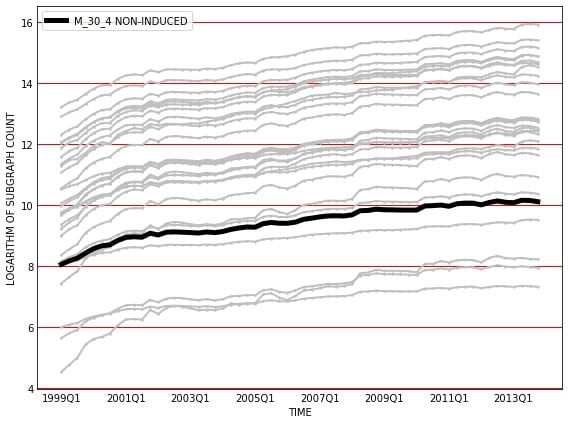} &
\includegraphics[height=0.65in]{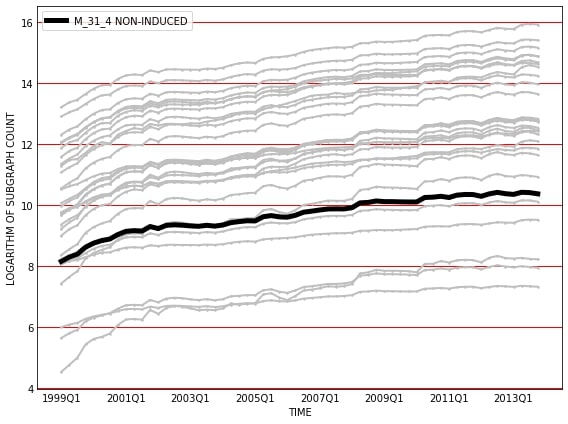} &
\includegraphics[height=0.65in]{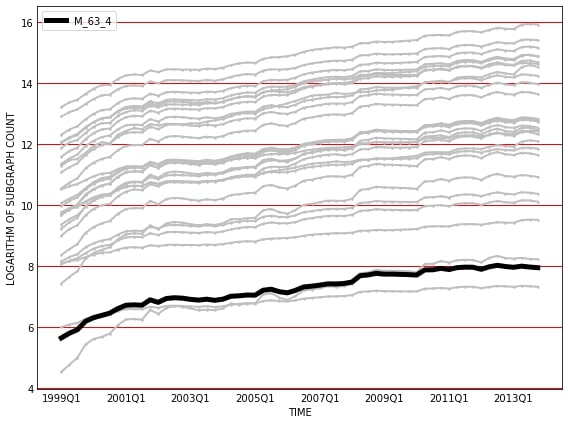}\\
3-star $M_{3}^{(3)}$ & Triangle $M_{7}^{(3)}$ & 4-star $M_{11}^{(4)}$ & 4-path $M_{13}^{(4)}$ & Tadpole $M_{15}^{(4)}$ & 4-circle $M_{30}^{(4)}$ & Diamond $M_{31}^{(4)}$ & 4-complete $M_{63}^{(4)}$\\
\\
\includegraphics[height=0.65in]{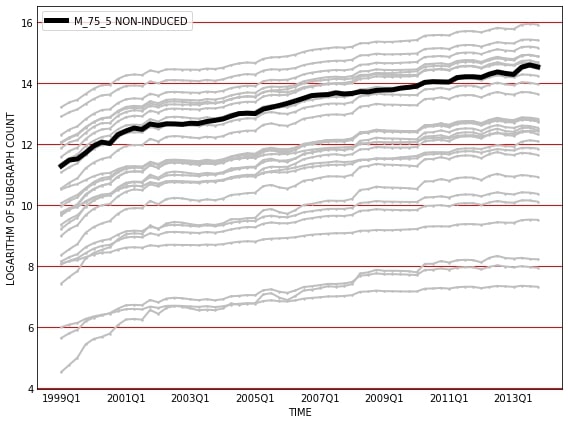} &
\includegraphics[height=0.65in]{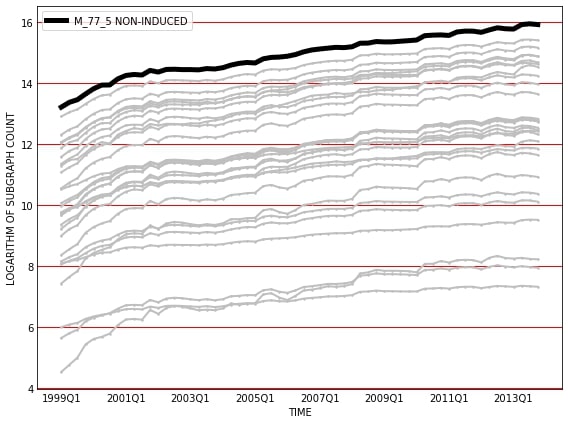} &
\includegraphics[height=0.65in]{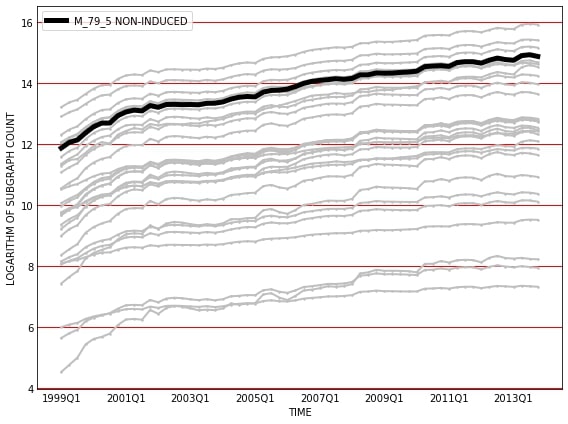} &
\includegraphics[height=0.65in]{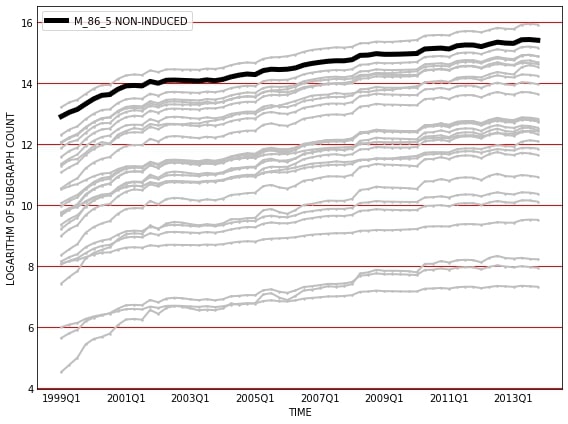} &
\includegraphics[height=0.65in]{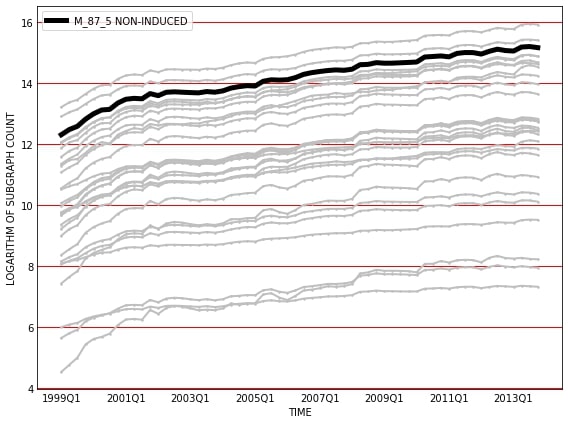} &
\includegraphics[height=0.65in]{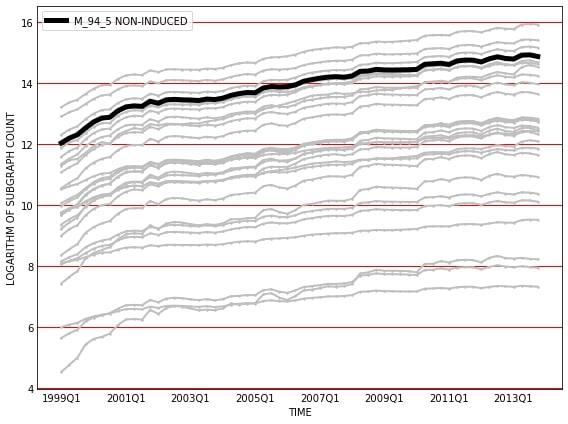} &
\includegraphics[height=0.65in]{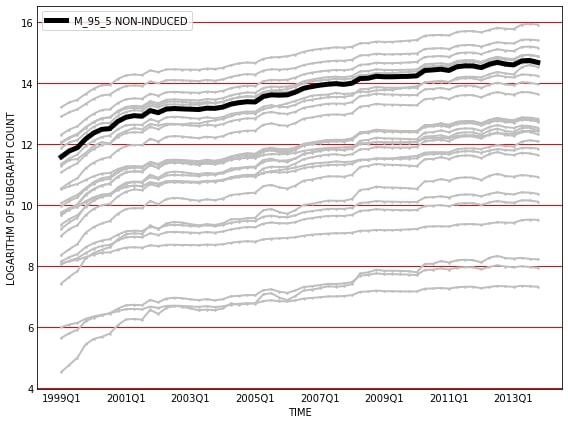} &
\includegraphics[height=0.65in]{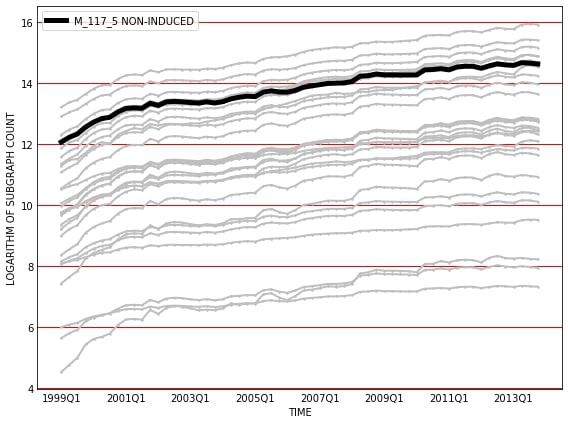}\\
5-star $M_{75}^{(5)}$ & 5-arrow $M_{77}^{(5)}$ & Cricket $M_{79}^{(5)}$ & 5-path $M_{86}^{(5)}$ & Bull $M_{87}^{(5)}$ & Banner $M_{94}^{(5)}$ & Stingray $M_{95}^{(5)}$ & Lollipop $M_{117}^{(5)}$\\
\\
\includegraphics[height=0.65in]{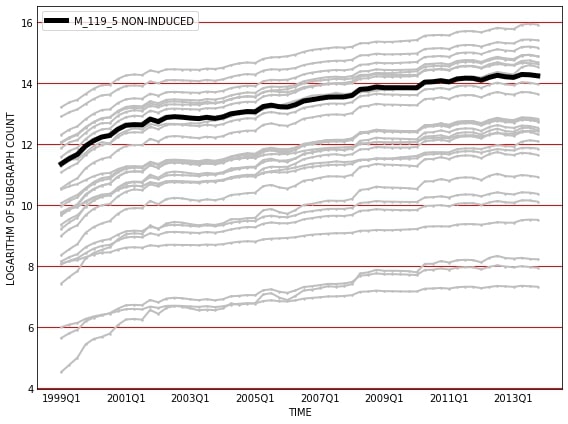} &
\includegraphics[height=0.65in]{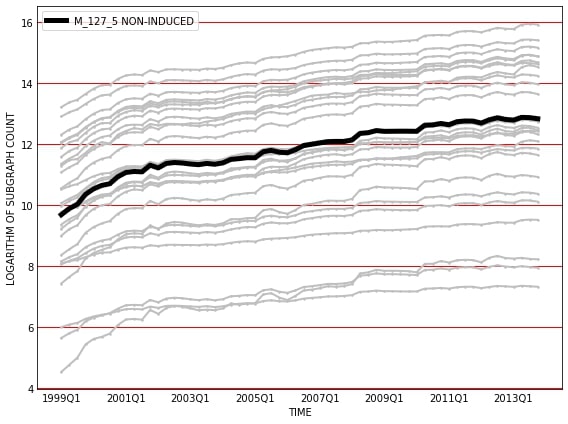} &
\includegraphics[height=0.65in]{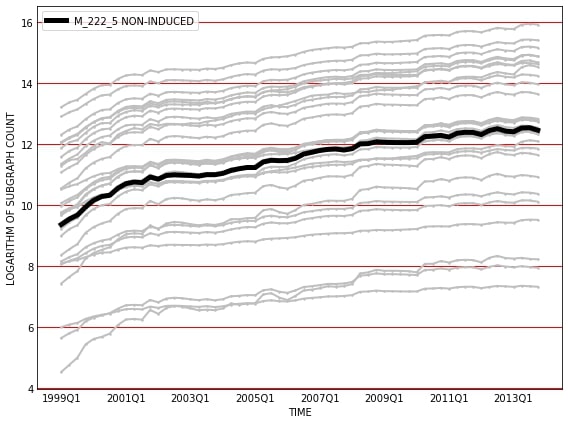} &
\includegraphics[height=0.65in]{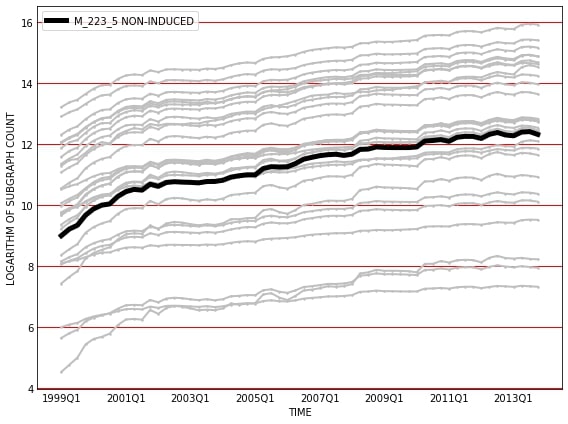} &
\includegraphics[height=0.65in]{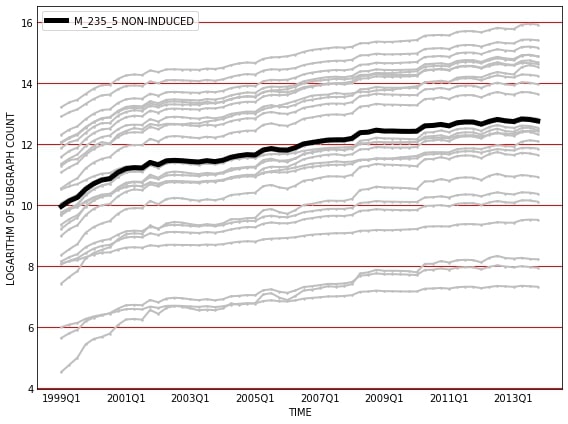} &
\includegraphics[height=0.65in]{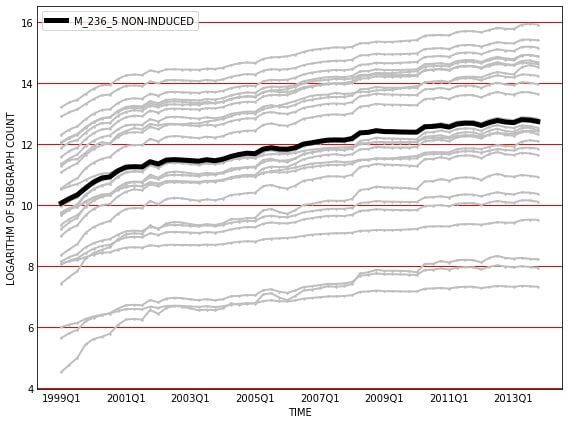} &
\includegraphics[height=0.65in]{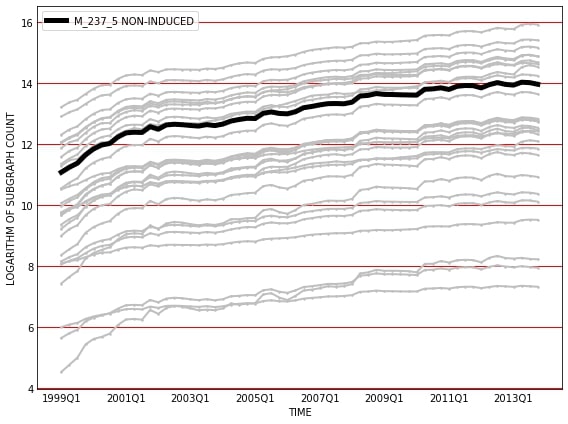} &
\includegraphics[height=0.65in]{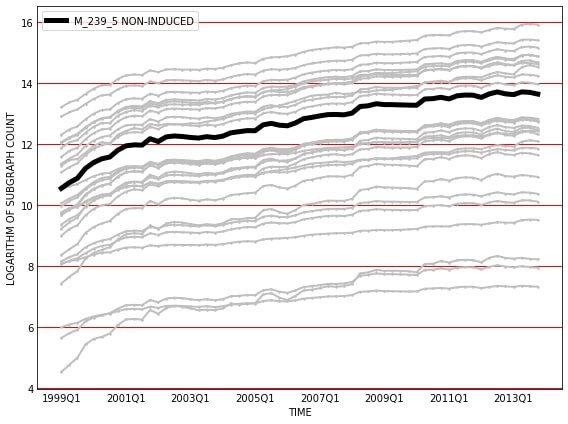}\\
Spinning top $M_{119}^{(5)}$ & Kite $M_{127}^{(5)}$ & Ufo $M_{222}^{(5)}$ & Chevron $M_{223}^{(5)}$ & Hourglass $M_{235}^{(5)}$ & 5-circle $M_{236}^{(5)}$ & House $M_{237}^{(5)}$ & Crown $M_{239}^{(5)}$\\
\\
\includegraphics[height=0.65in]{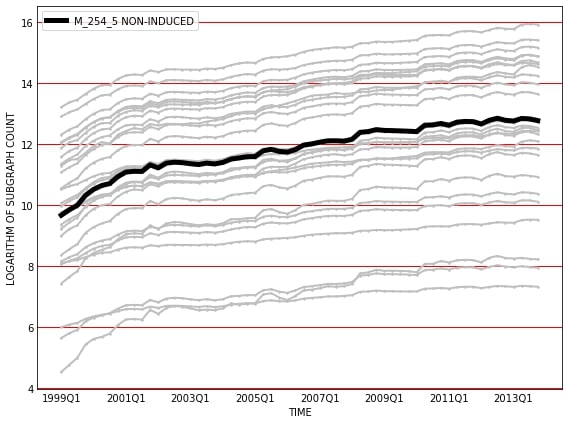} &
\includegraphics[height=0.65in]{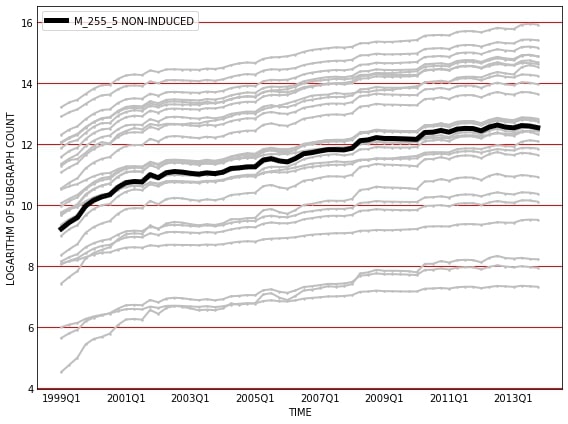} &
\includegraphics[height=0.65in]{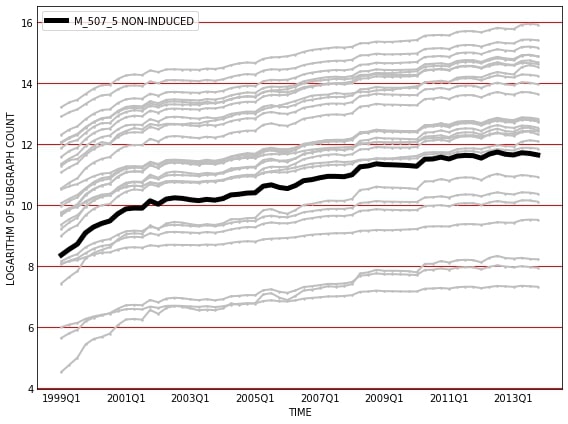} &
\includegraphics[height=0.65in]{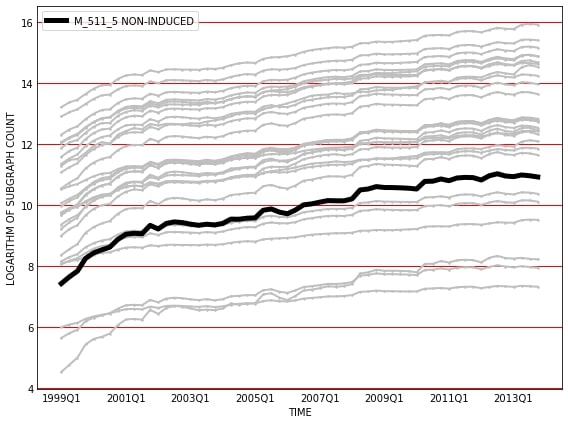} &
\includegraphics[height=0.65in]{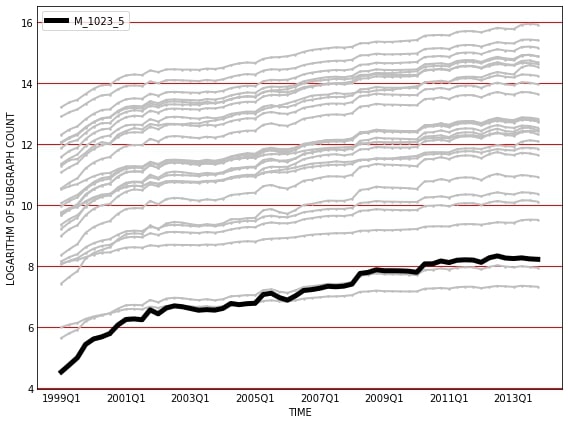} &
&
&
\\
Envelope $M_{254}^{(5)}$ & Lamp $M_{255}^{(5)}$ & Arrowhead $M_{507}^{(5)}$ & Cat's cradle $M_{511}^{(5)}$ & 5-complete $M_{1023}^{(5)}$ & & &\\
\\
\end{tabular}
\caption{Log counts of three-node, four-node and five-node non-induced subgraphs, computed using exact enumeration formulae on Southwest's network from 1999Q1 to 2013Q4 (Section \ref{sec:scaling_properties}).}
\label{fig:non_induced_counts}
\end{table}
\vspace*{\fill}
\end{landscape}

\clearpage

\setlength\extrarowheight{5pt}

\begin{table}[!htp]
    \begin{flushleft}
        \begin{threeparttable}
            \onehalfspacing
            \caption{Summary results of log-log regressions of non-induced subgraph count $|M_{a}^{(b)}|$ on number of edges $m$ (Section \ref{sec:scaling_properties}).}\label{tab:log_log_results}
            \begin{footnotesize}
                \begin{tabular*}{\linewidth}{@{\extracolsep{\fill}}l*{8}{c}}
                    \toprule
                    & \multicolumn{8}{c}{Subgraph}\\
                    \cline{2-9}
                    & $M_{3}^{(3)}$ & $M_{7}^{(3)}$ & $M_{11}^{(4)}$ & $M_{13}^{(4)}$ & $M_{15}^{(4)}$ & $M_{30}^{(4)}$ & $M_{31}^{(4)}$ & $M_{63}^{(4)}$\\
                    \midrule
                    slope ($\beta$) & 1.95 & 1.86 & 3.17 & 2.60 & 2.91 & 2.87 & 3.07 & 3.18\\
                    $R^{2}$ & 0.998 & 0.985 & 0.993 & 0.995 & 0.990 & 0.984 & 0.982 & 0.977\\
                    & \multicolumn{8}{c}{Subgraph}\\
                    \cline{2-9}
                    & $M_{75}^{(5)}$ & $M_{77}^{(5)}$ & $M_{79}^{(5)}$ & $M_{86}^{(5)}$ & $M_{87}^{(5)}$ & $M_{94}^{(5)}$ & $M_{95}^{(5)}$ & $M_{117}^{(5)}$\\
                    \midrule
                    slope ($\beta$) & 4.44 & 3.67 & 4.12 & 3.42 & 3.88 & 3.90 & 4.25 & 3.52\\
                    $R^{2}$ & 0.987 & 0.994 & 0.988 & 0.993 & 0.990 & 0.986 & 0.984 & 0.987\\
                    & \multicolumn{8}{c}{Subgraph}\\
                    \cline{2-9}
                    & $M_{119}^{(5)}$ & $M_{127}^{(5)}$ & $M_{222}^{(5)}$ & $M_{223}^{(5)}$ & $M_{235}^{(5)}$ & $M_{236}^{(5)}$ & $M_{237}^{(5)}$ & $M_{239}^{(5)}$\\
                    \midrule
                    slope ($\beta$) & 3.94 & 4.28 & 4.28 & 4.58 & 3.87 & 3.69 & 3.98 & 4.24\\
                    $R^{2}$ & 0.984 & 0.982 & 0.979 & 0.978 & 0.982 & 0.981 & 0.979 & 0.978\\
                    & \multicolumn{8}{c}{Subgraph}\\
                    \cline{2-9}
                    & $M_{254}^{(5)}$ & $M_{255}^{(5)}$ & $M_{507}^{(5)}$ & $M_{511}^{(5)}$ & $M_{1023}^{(5)}$ & & & \\
                    \midrule
                    slope ($\beta$) & 4.26 & 4.53 & 4.50 & 4.76 & 5.01 & & & \\
                    $R^{2}$ & 0.976 & 0.976 & 0.974 & 0.974 & 0.973 & & & \\
                    \bottomrule
                \end{tabular*}
            \end{footnotesize}
         \end{threeparttable}
    \end{flushleft}
\end{table}

    \clearpage

\begin{figure}\centering
    	\begin{subfigure}{0.48\textwidth}
    		\centering
    		\includegraphics[width=.7\linewidth]{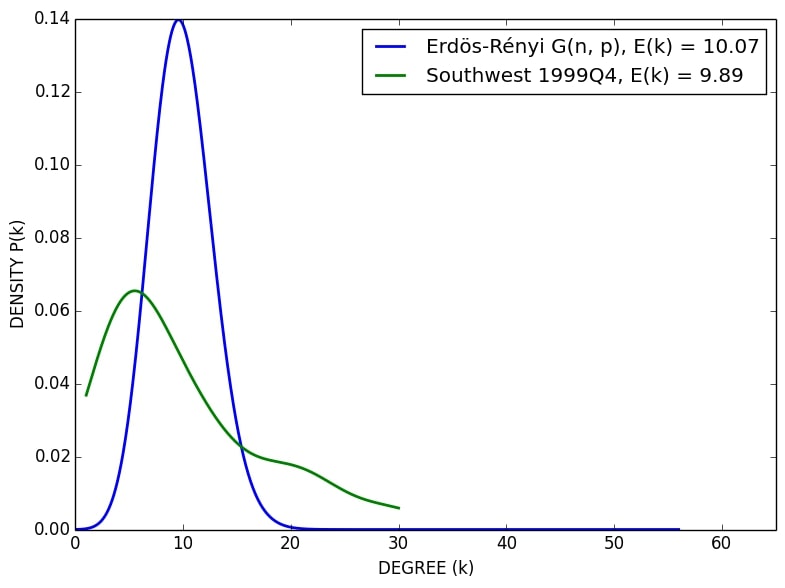}
    		\caption{1999Q4.}
    		\label{fig:degree_distribution_1999_1}
    	\end{subfigure}
    	\begin{subfigure}{0.48\textwidth}
    		\centering
    		\includegraphics[width=.7\linewidth]{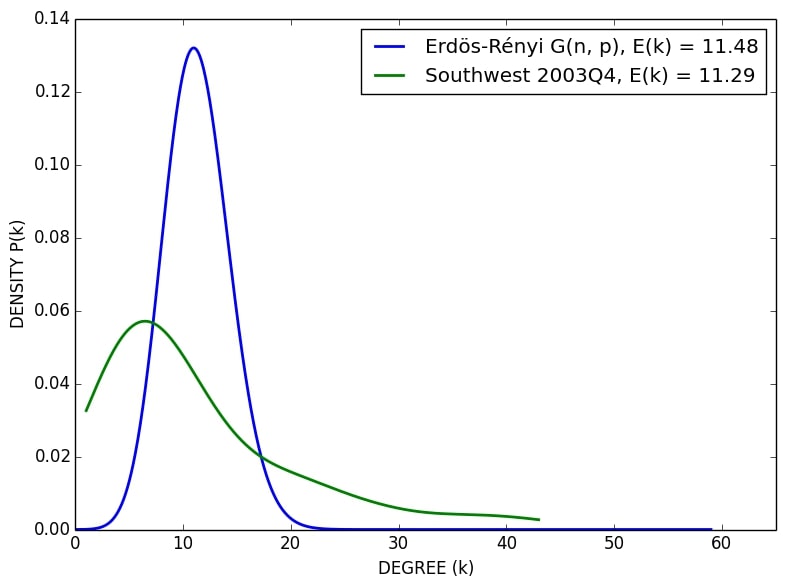}
    		\caption{2003Q4.}
    		\label{fig:degree_distribution_2003_4}
    	\end{subfigure}
    		\begin{subfigure}{0.48\textwidth}
    		\centering
    		\includegraphics[width=.7\linewidth]{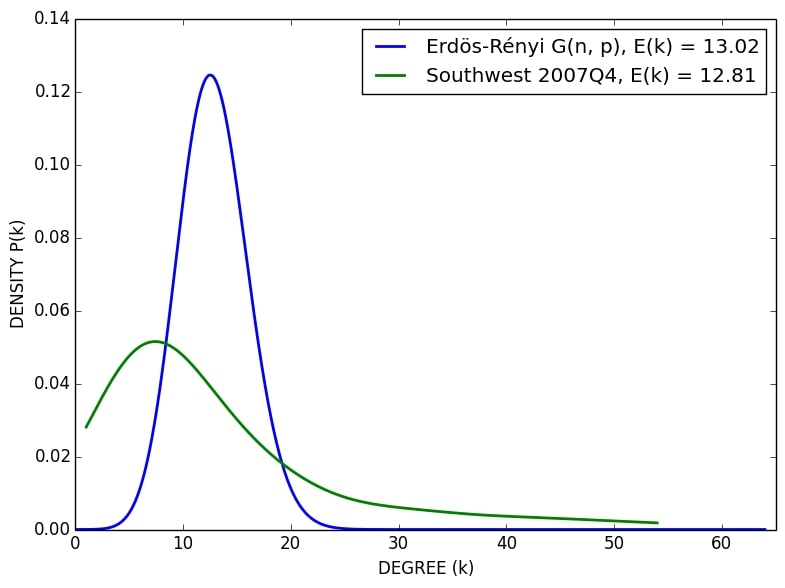}
    		\caption{2007Q4.}
    		\label{fig:degree_distribution_2007_4}
    	\end{subfigure}
    	    \begin{subfigure}{0.48\textwidth}
    		\centering
    		\includegraphics[width=.7\linewidth]{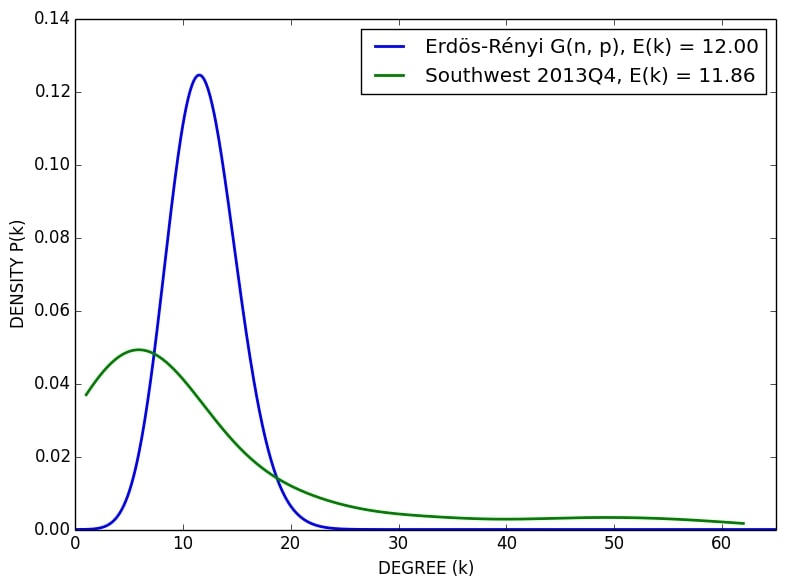}
    		\caption{2013Q4.}
    		\label{fig:degree_distribution_2013_4}
    	\end{subfigure}
    	\caption{Gaussian kernel density estimates corresponding to the degree distribution $P(k)$ of Southwest's network, not correcting for edge-effects at $k=0$; and the binomial degree distribution of the Erd{\H{o}}s-R{\'{e}}nyi graph $G(n, p)$ with edge-formation probability $p = d(G)$ (Section \ref{sec:motif_rewiring}).}
    	\label{fig:degree_distribution}
    \end{figure}

\clearpage

\section{Spatial properties of the triangle subgraph}\label{sec:results_spatial}
One of the distinctive characteristics of airline networks is their spatial nature: unlike many natural or social networks, the nodes (airports) have a fixed geographical location. With this in mind, we explore the dynamic spatial distribution of the triangle subgraph, chosen because ``area'' and ``center'' have a clear meaning in this case. To be precise, the area and barycenter of a triangle subgraph on a curved surface are calculated using the latitude and logitude of each node, with a great circle method. To illustrate, the triangle formed by Baltimore--Washington (BWI), Denver (DEN) and Las Vegas McCarran (LAS), with coordinates $(39.18^{\degree}, -76.67^{\degree})$, $(39.86^{\degree}, -104.67^{\degree})$, and $(36.08^{\degree}, -115.17^{\degree})$, respectively, has area 87,754 square miles, and a center located at $(38.37^{\degree}, -98.84^{\degree})$.

In Figure \ref{fig:barycenters}, we show the general trends in the spatial distribution of triangles between 1999 and 2013. We observe that the triangle centers are evenly-distributed across the U.S. in 1999, but become progressively more concentrated in the east of the country, most notably from 2009 onwards. In Figure \ref{fig:areas}, we plot the density of the triangle subgraph area: this shows that the triangles generally become larger over time. This approach provides a straightforward graphical means of assessing the spatial evolution of clustering in a network over time: clustering (in the sense of connected triples) seems to evolve towards (at least two) nodes that are located in the eastern U.S.

\begin{figure}\centering
    	\begin{subfigure}{0.48\textwidth}
    		\centering
    		\includegraphics[width=.75\linewidth]{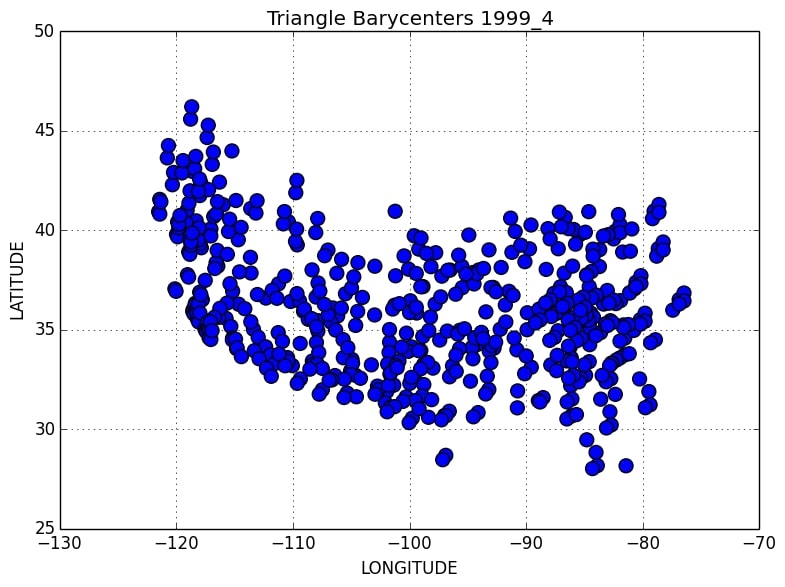}
    		\caption{1999Q4.}
    		\label{fig:barycenters_1999_4}
    	\end{subfigure}
    	\begin{subfigure}{0.48\textwidth}
    		\centering
    		\includegraphics[width=.75\linewidth]{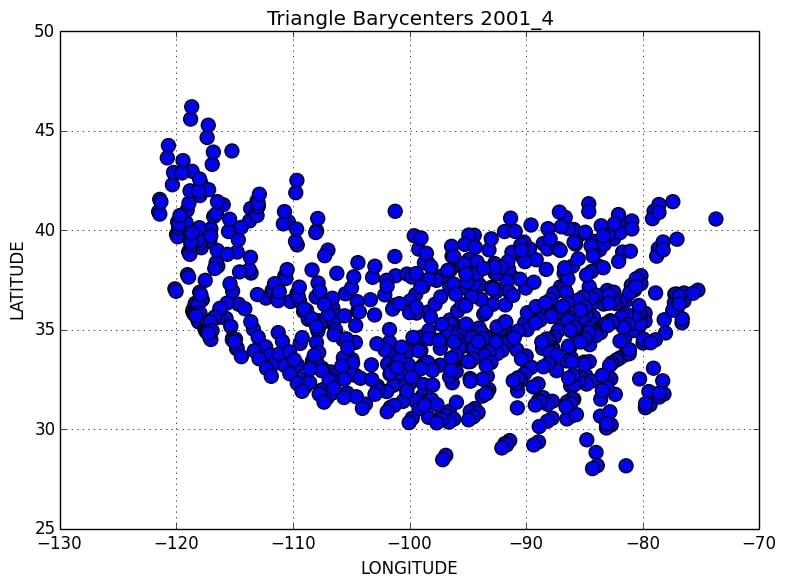}
    		\caption{2001Q4.}
    		\label{fig:barycenters_2001_4}
    	\end{subfigure}
    		\begin{subfigure}{0.48\textwidth}
    		\centering
    		\includegraphics[width=.75\linewidth]{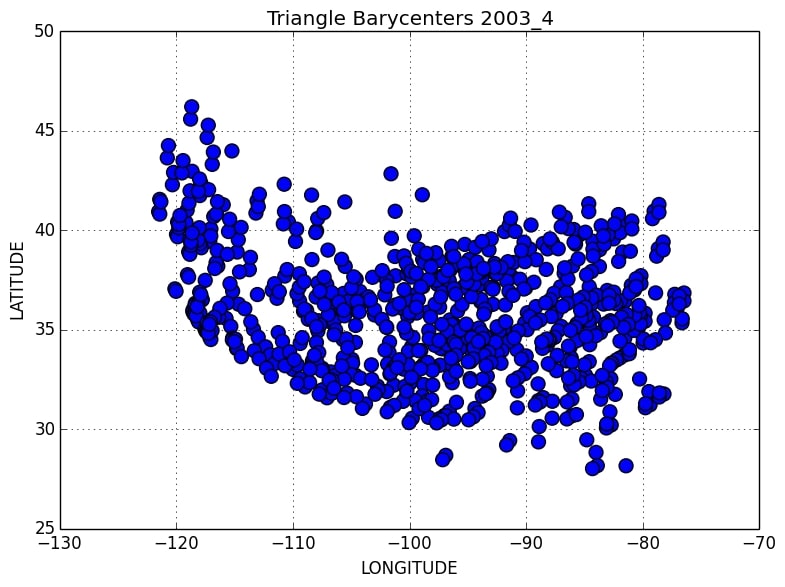}
    		\caption{2003Q4.}
    		\label{fig:barycenters_2003_4}
    	\end{subfigure}
    	    \begin{subfigure}{0.48\textwidth}
    		\centering
    		\includegraphics[width=.75\linewidth]{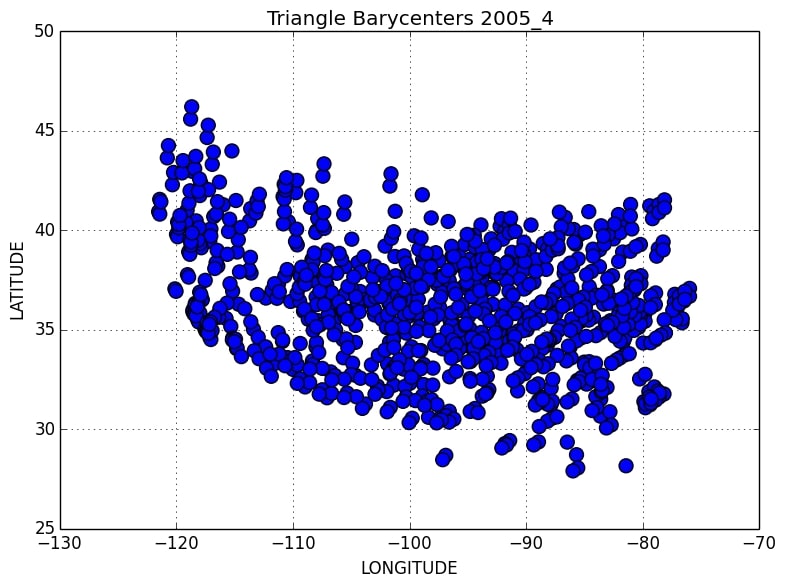}
    		\caption{2005Q4.}
    		\label{fig:barycenters_2005_4}
    	\end{subfigure}
    	\begin{subfigure}{0.48\textwidth}
    		\centering
    		\includegraphics[width=.75\linewidth]{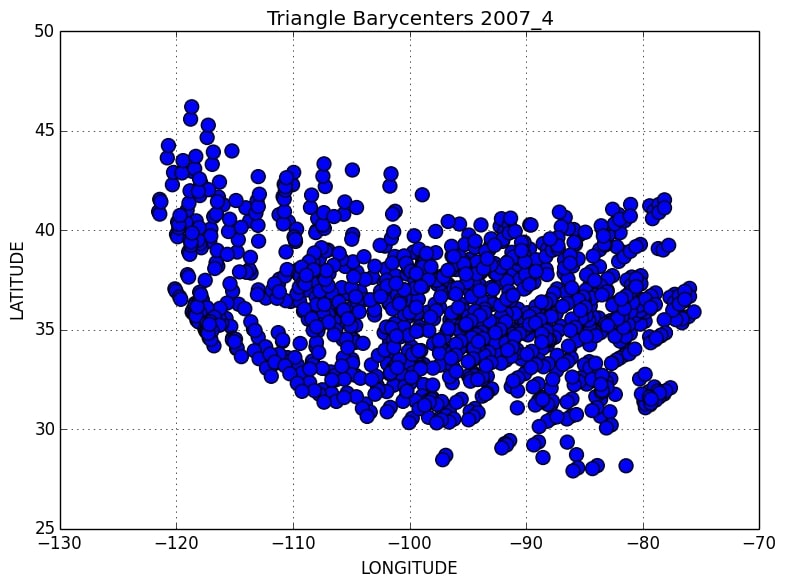}
    		\caption{2007Q4.}
    		\label{fig:barycenters_2007_4}
    	\end{subfigure}
    	\begin{subfigure}{0.48\textwidth}
    		\centering
    		\includegraphics[width=.75\linewidth]{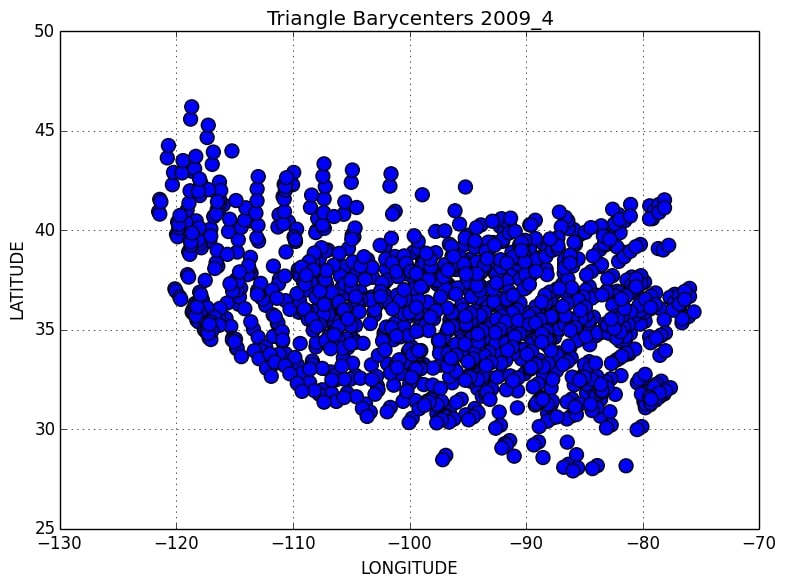}
    		\caption{2009Q4.}
    		\label{fig:barycenters_2009_4}
    	\end{subfigure}
    		\begin{subfigure}{0.48\textwidth}
    		\centering
    		\includegraphics[width=.75\linewidth]{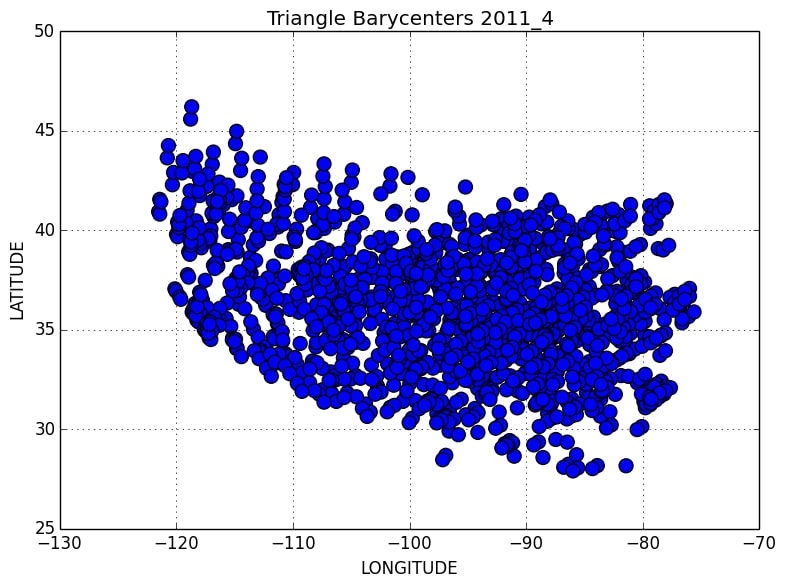}
    		\caption{2011Q4.}
    		\label{fig:barycenters_2011_4}
    	\end{subfigure}
    	    \begin{subfigure}{0.48\textwidth}
    		\centering
    		\includegraphics[width=.75\linewidth]{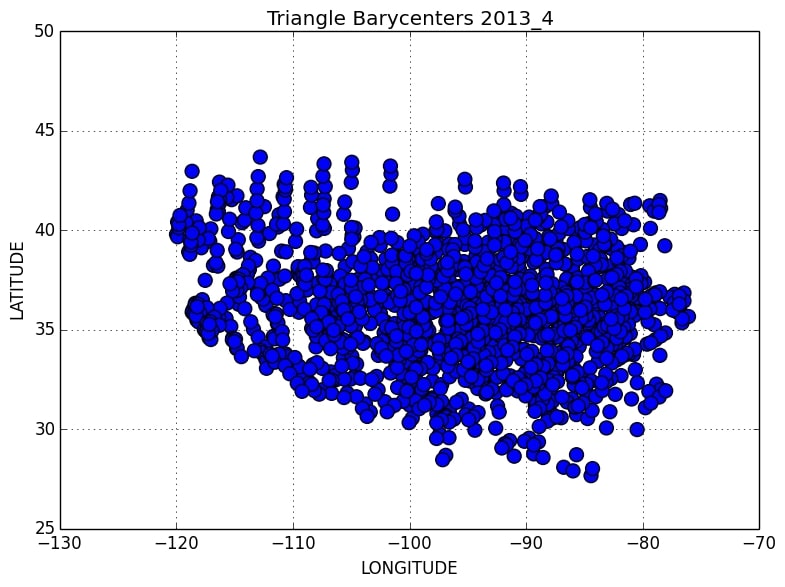}
    		\caption{2013Q4.}
    		\label{fig:barycenters_2013_4}
    	\end{subfigure}
    	\caption{Spatial distribution of the triangle subgraph center for odd-numbered years, last quarter, between 1999 and 2013 (Appendix \ref{sec:results_spatial}).}
    	\label{fig:barycenters}
    \end{figure}

\begin{figure}\centering
    	\begin{subfigure}{0.48\textwidth}
    		\centering
    		\includegraphics[width=.75\linewidth]{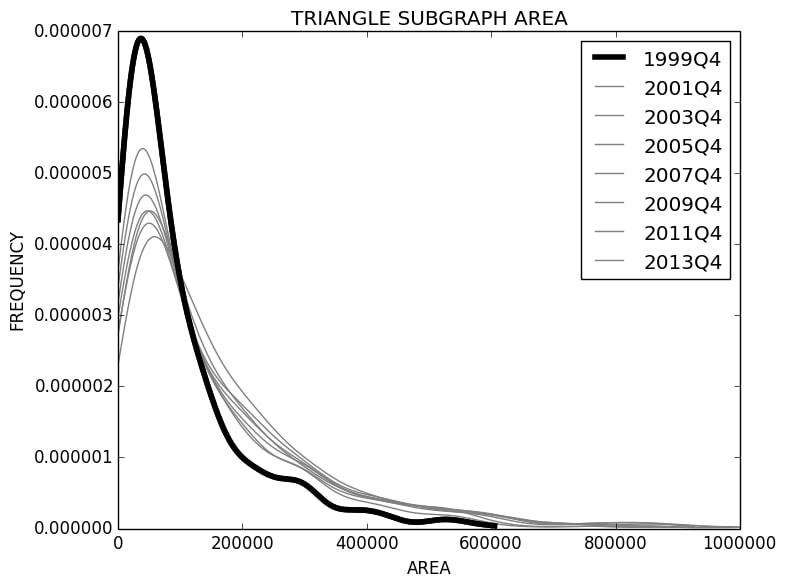}
    		\caption{1999Q4.}
    		\label{fig:triangle_area__1999_4}
    	\end{subfigure}
    	\begin{subfigure}{0.48\textwidth}
    		\centering
    		\includegraphics[width=.75\linewidth]{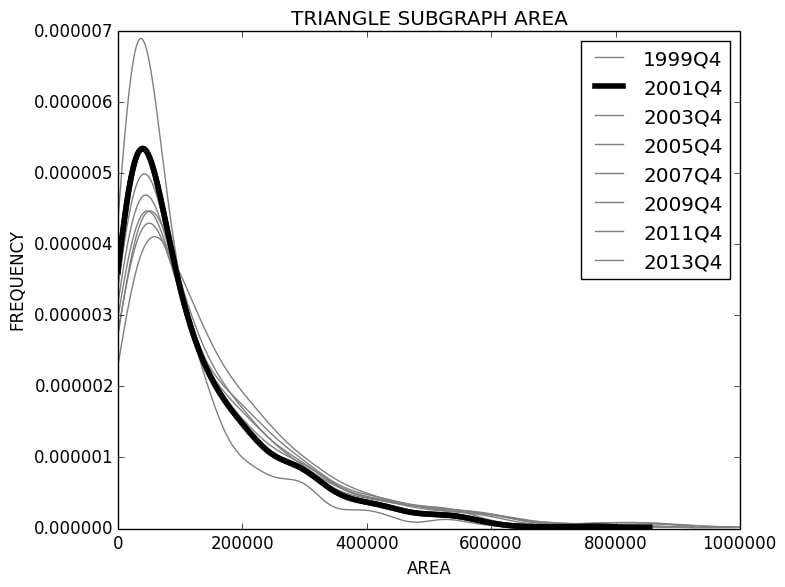}
    		\caption{2001Q4.}
    		\label{fig:triangle_area__2001_4}
    	\end{subfigure}
    		\begin{subfigure}{0.48\textwidth}
    		\centering
    		\includegraphics[width=.75\linewidth]{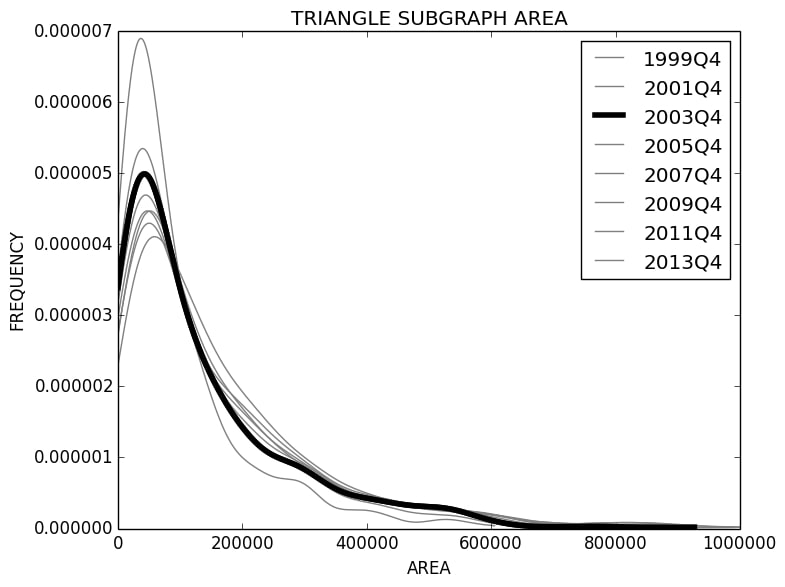}
    		\caption{2003Q4.}
    		\label{fig:triangle_area__2003_4}
    	\end{subfigure}
    	    \begin{subfigure}{0.48\textwidth}
    		\centering
    		\includegraphics[width=.75\linewidth]{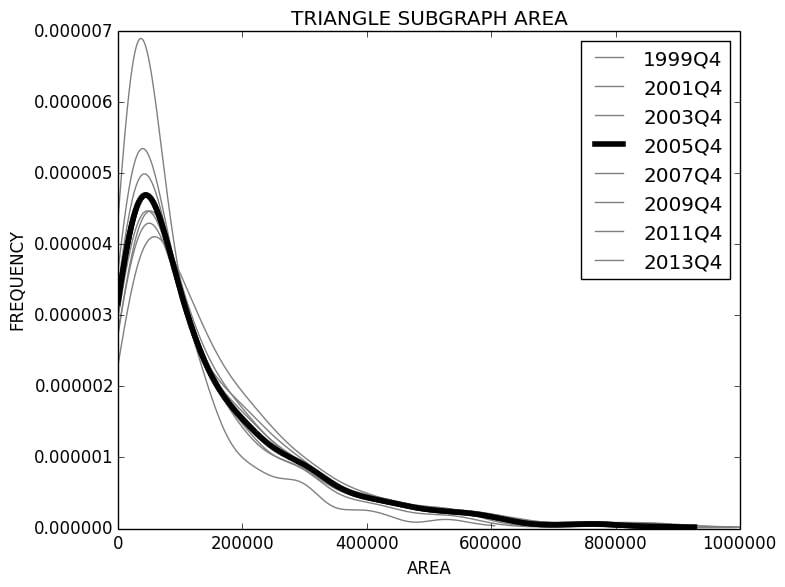}
    		\caption{2005Q4.}
    		\label{fig:triangle_area__2005_4}
    	\end{subfigure}
    	\begin{subfigure}{0.48\textwidth}
    		\centering
    		\includegraphics[width=.75\linewidth]{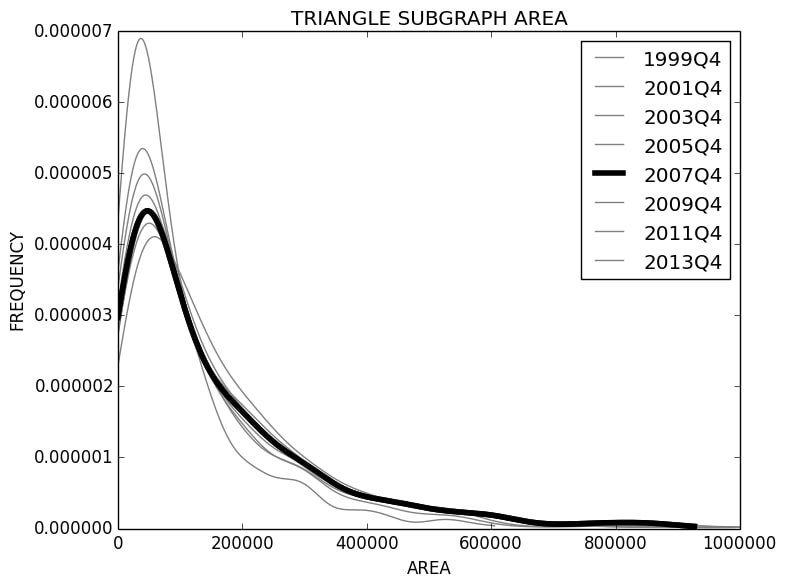}
    		\caption{2007Q4.}
    		\label{fig:triangle_area__2007_4}
    	\end{subfigure}
    	\begin{subfigure}{0.48\textwidth}
    		\centering
    		\includegraphics[width=.75\linewidth]{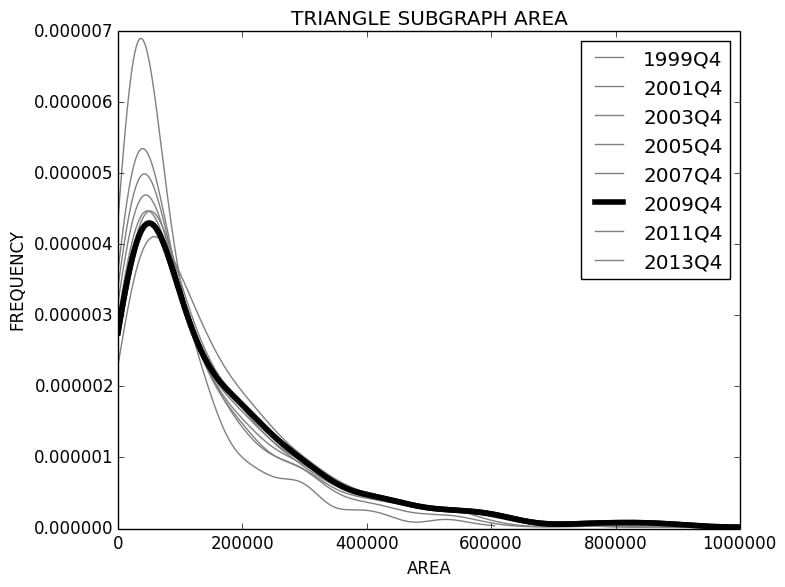}
    		\caption{2009Q4.}
    		\label{fig:triangle_area__2009_4}
    	\end{subfigure}
    		\begin{subfigure}{0.48\textwidth}
    		\centering
    		\includegraphics[width=.75\linewidth]{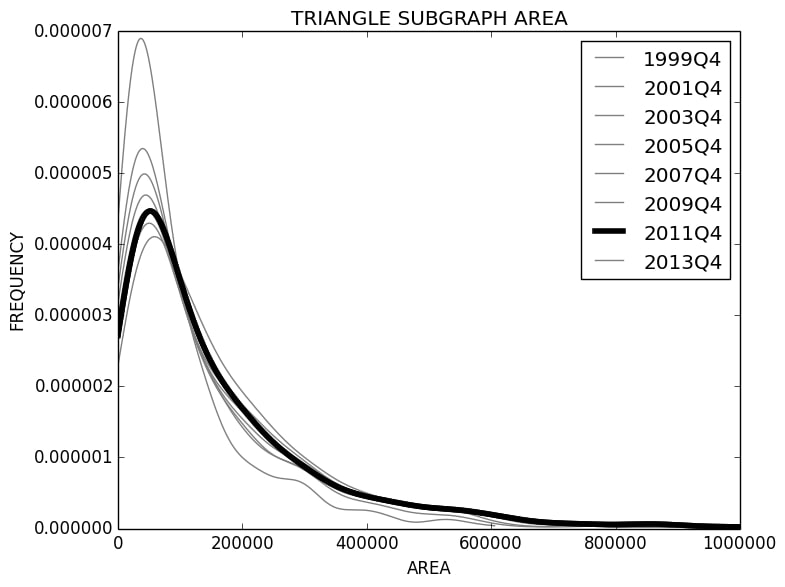}
    		\caption{2011Q4.}
    		\label{fig:triangle_area__2011_4}
    	\end{subfigure}
    	    \begin{subfigure}{0.48\textwidth}
    		\centering
    		\includegraphics[width=.75\linewidth]{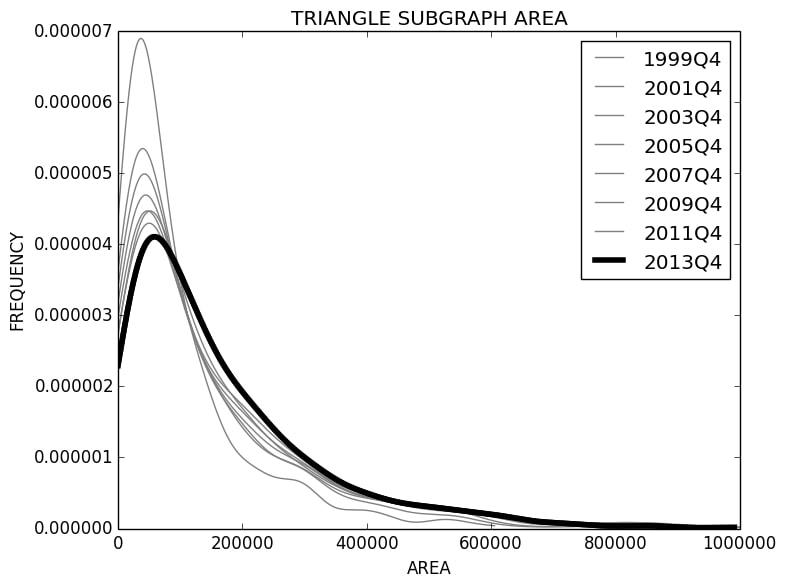}
    		\caption{2013Q4.}
    		\label{fig:triangle_area_s_2013_4}
    	\end{subfigure}
    	\caption{Kernel density estimates of the triangle subgraph area for odd-numbered years, last quarter, between 1999 and 2013, not correcting for edge effects at zero (Appendix \ref{sec:results_spatial}).}
    	\label{fig:areas}
    \end{figure}
    
\clearpage

\section{Weak bounds on complete subgraph counts}\label{sec:bounds_on_complete}
In passing, we mention a result of \cite{fisher_ryan92}, that we give in a symmetrized form, which provides bounds on the number of complete subgraphs (cliques) in a simple graph:

\begin{theorem}[Bounds on number of complete subgraphs \cite{fisher_ryan92}]\label{thm:bounds_complete}
Let $G$ be a simple graph with clique number $w=w(G)$. For $1 \leq h \leq w$, let $T_{h}$ be the number of $h$-complete subgraphs. Then:
\begin{equation}\label{eq:bounds_complete}
\left(\frac{T_{h+1}}{\binom{w}{h+1}}\right)^{\frac{1}{h+1}} \leq \left(\frac{T_{h}}{\binom{w}{h}}\right)^{\frac{1}{h}}.
\end{equation}
\end{theorem}
\noindent The terms $T_{1}$ through to $T_{5}$ refer to the number of nodes and edges, and the counts of the triangle, 4-complete and 5-complete subgraphs. To illustrate the strength of these inequalities, we use the Bron-Kerbosch algorithm \cite{bron_kerbosch73} to find all maximal cliques in Southwest's 2013Q4 network: this gives four maximum cliques with $w(G)=11$.\footnote{Each of the maximum cliques contains a common 9-complete subgraph, on Nashville (BNA), Baltimore--Washington (BWI), Denver (DEN), Houston William P. Hobby (HOU), Las Vegas McCarran (LAS), Kansas City (MCI), Chicago Midway (MDW), Louis Armstrong New Orleans (MSY) and St. Louis Missouri (STL). The 11-complete subgraphs include, in addition, one of the following pairs of nodes: (FLL, TPA), (LAX, PHX), (MCO, PHX) or (PHX, TPA), on Fort Lauderdale--Hollywood (FLL), Los Angeles (LAX), Orlando (MCO), Phoenix Sky Harbor (PHX), and Tampa (TPA). Each maximum subgraph contains 12.5\% of the 88 nodes in the entire 2013Q4 network.} We observe $T_{1}=88$, $T_{2}=522$, $T_{3}=1501$, $T_{4}=2806$ and $T_{5}=3690$. It follows from (\ref{eq:bounds_complete}) that, in our notation,
\begin{equation*}
    \left(\frac{|M_{1023}^{(5)}|}{\binom{w}{5}}\right)^{\frac{1}{5}} \leq \left(\frac{|M_{63}^{(4)}|}{\binom{w}{4}}\right)^{\frac{1}{4}} \leq \left(\frac{|M_{7}^{(3)}|}{\binom{w}{3}}\right)^{\frac{1}{3}} \leq \left(\frac{m}{\binom{w}{2}}\right)^{\frac{1}{2}} \leq \left(\frac{n}{w}\right),
\end{equation*}
which reduces to the rather weak set of inequalities
\begin{equation*}
    m \leq 3520, \quad |M_{7}^{(3)}| \leq 4824, \quad |M_{63}^{(4)}| \leq 6266, \quad |M_{1023}^{(5)}| \leq 6708,
\end{equation*}
based on the observed subgraph counts.

\clearpage

\section*{Acknowledgements}
We are grateful to Karim Abadir, Gergana Bounova, Pascal Lezaud, Chantal Roucolle, Miguel Urdanoz, and participants at the Conference on Complex Systems (CCS2018, Thessaloniki), for helpful comments and suggestions. We also thank Patrick Senac for supporting this project. The visualization, subgraph analysis, and motif detection tools used in this paper were coded by the authors in Python 2.7. The usual caveat applies. This research did not receive any specific grant from funding agencies in the public, commercial, or not-for-profit sectors.\\
\\
Keywords: Airline network, Graph theory, Network motif, Scaling, Subgraph, Topology transitions.\\
\\
PACS numbers: 02.10.Ox (Combinatorics; graph theory), 89.40.Dd (Air transportation), 89.65.Gh (Economics; econophysics; financial markets; business and management), 89.75.-k (Complex systems).

\clearpage

\singlespace

\bibliographystyle{plainnat}

\bibliography{ms}

\begin{thebibliography}{79}
\providecommand{\natexlab}[1]{#1}
\providecommand{\url}[1]{\texttt{#1}}
\expandafter\ifx\csname urlstyle\endcsname\relax
  \providecommand{\doi}[1]{doi: #1}\else
  \providecommand{\doi}{doi: \begingroup \urlstyle{rm}\Url}\fi

\bibitem[Aguirregabiria and Ho(2012)]{aguirregabiria_ho12}
V.~Aguirregabiria and {C.-Y.} Ho.
\newblock A dynamic oligopoly game of the {US} airline industry: {E}stimation
  and policy experiments.
\newblock \emph{Journal of Econometrics}, 168:\penalty0 156--173, 2012.

\bibitem[Alon et~al.(1997)Alon, Yuster, and Zwick]{alon_etal97}
N.~Alon, R.~Yuster, and U.~Zwick.
\newblock Finding and counting given length cycles.
\newblock \emph{Algorithmica}, 17:\penalty0 209--223, 1997.

\bibitem[Alon(2007)]{alon07}
U.~Alon.
\newblock Network motifs: {T}heory and experimental approaches.
\newblock \emph{Nature Reviews Genetics}, 8:\penalty0 450--461, 2007.

\bibitem[Angel et~al.(2017)Angel, {van der Hofstad}, and
  Holmgren]{angel_etal17}
O.~Angel, R.~{van der Hofstad}, and C.~Holmgren.
\newblock Limit laws for self-loops and multiple edges in the configuration
  model.
\newblock Technical Report arXiv:1603.07172v2, 2017.

\bibitem[{{\'A}ngeles Serrano} et~al.(2008){{\'A}ngeles Serrano}, Krioukov, and
  Bogu{\~n}{\'a}]{angelesserrano_etal08}
M.~{{\'A}ngeles Serrano}, D.~Krioukov, and M.~Bogu{\~n}{\'a}.
\newblock Self-similarity of complex networks and hidden metric spaces.
\newblock \emph{Physical Review Letters}, 100:\penalty0 078701, 2008.

\bibitem[Barab{\'a}si(2016)]{barabasi16}
{A.-L}. Barab{\'a}si.
\newblock \emph{Network Science}.
\newblock Cambridge University Press, 2016.

\bibitem[Barrat et~al.(2004)Barrat, Barth{\'e}lemy, Pastor-Satorras, and
  Vespignani]{barrat_etal04}
A.~Barrat, M.~Barth{\'e}lemy, R.~Pastor-Satorras, and A.~Vespignani.
\newblock The architecture of complex weighted networks.
\newblock \emph{PNAS}, 101:\penalty0 3747--3752, 2004.

\bibitem[Bj{\"o}rklund et~al.(2014)Bj{\"o}rklund, Pagh, {Vassilevska Williams},
  and Zwick]{bjorklund_etal14}
A.~Bj{\"o}rklund, R.~Pagh, V.~{Vassilevska Williams}, and U.~Zwick.
\newblock Listing triangles.
\newblock In \emph{ICALP 14 (International Colloquium on Automata, Languages,
  and Programming)}, 2014.

\bibitem[Bloch et~al.(2016)Bloch, Jackson, and Tebaldi]{bloch_etal16}
F.~Bloch, M.O. Jackson, and P.~Tebaldi.
\newblock Centrality measures in networks.
\newblock Technical Report arXiv:1608.05845v1, 2016.

\bibitem[Bounova and {de Weck}(2012)]{bounova_weck12}
G.~Bounova and O.~{de Weck}.
\newblock Overview of metrics and their correlation patterns for
  multiple-metric topology analysis on heterogeneous graph ensembles.
\newblock \emph{Physical Review E}, 85:\penalty0 016117, 2012.

\bibitem[Bounova(2009)]{bounova09}
G.A. Bounova.
\newblock \emph{Topological evolution of networks: {C}ase studies in the {US}
  airlines and language {W}ikipedias}.
\newblock PhD thesis, Massachusetts Institute of Technology, 2009.

\bibitem[Bron and Kerbosch(1973)]{bron_kerbosch73}
C.~Bron and J.~Kerbosch.
\newblock Algorithm 457: {F}inding all cliques of an undirected graph.
\newblock \emph{Communications of the ACM}, 16:\penalty0 575--577, 1973.

\bibitem[Chen et~al.(2013)Chen, Qu, Cao, Zhou, Li, Liang, Li, He, Feng, Jia,
  and He]{chen_etal13}
L.~Chen, X.~Qu, M.~Cao, Y.~Zhou, W.~Li, B.~Liang, W.~Li, W.~He, C.~Feng,
  X.~Jia, and Y.~He.
\newblock Identification of breast cancer patients based on human signaling
  network motifs.
\newblock \emph{Scientific Reports}, 3:\penalty0 3368, 2013.

\bibitem[Chu and Cheng(2012)]{chu_cheng12}
S.~Chu and J.~Cheng.
\newblock Triangle listing in massive networks.
\newblock \emph{ACM Transactions on Knowledge Discovery from Data}, 6:\penalty0
  17:1--17:32, 2012.

\bibitem[Ciliberto and Tamer(2009)]{ciliberto_tamer09}
F.~Ciliberto and E.~Tamer.
\newblock Market structure and multiple equilibria in airline markets.
\newblock \emph{Econometrica}, 77:\penalty0 1791--1828, 2009.

\bibitem[Clauset et~al.(2009)Clauset, Shalizi, and Newman]{clauset_etal09}
A.~Clauset, C.R. Shalizi, and M.E.J. Newman.
\newblock Power-law distributions in empirical data.
\newblock \emph{SIAM Review}, 51:\penalty0 661--703, 2009.

\bibitem[Cormen et~al.(2009)Cormen, Leiserson, Rivest, and
  Stein]{cormen_etal09}
T.H. Cormen, C.E. Leiserson, R.L. Rivest, and C.~Stein.
\newblock \emph{Introduction to Algorithms}.
\newblock MIT Press, 3rd edition, 2009.

\bibitem[Dai et~al.(2014)Dai, Liu, and Serfes]{dai_etal14}
M.~Dai, Q.~Liu, and K.~Serfes.
\newblock Is the effect of competition on price dispersion non-monotonic?
  {E}vidence from the {U.S.} airline industry.
\newblock \emph{Review of Economics and Statistics}, 96:\penalty0 161--170,
  2014.

\bibitem[Diestel(2017)]{diestel17}
R.~Diestel.
\newblock \emph{Graph Theory}.
\newblock Springer, 5th edition, 2017.

\bibitem[Dobrin et~al.(2004)Dobrin, Beg, Barab{\'a}si, and
  Oltvai]{dobrin_etal04}
R.~Dobrin, Q.K. Beg, {A.-L}. Barab{\'a}si, and Z.N. Oltvai.
\newblock Aggregation of topological motifs in the \emph{Escherichia coli}
  transcriptional regulatory network.
\newblock \emph{BMC Bioinformatics}, 5:\penalty0 10, 2004.

\bibitem[Dossin and Lawford(2017)]{dossin_lawford17}
A.~Dossin and S.~Lawford.
\newblock Weighted network centrality measures: with application to {U.S.}
  domestic airlines.
\newblock DEVI / ENAC unpublished report, 2017.

\bibitem[Durlauf(2005)]{durlauf05}
S.N. Durlauf.
\newblock Complexity and empirical economics.
\newblock \emph{Economic Journal}, 115:\penalty0 F225--F243, 2005.

\bibitem[Eglese(1990)]{eglese90}
R.W. Eglese.
\newblock Simulated annealing: {A} tool for operational research.
\newblock \emph{European Journal of Operational Research}, 46:\penalty0
  271--281, 1990.

\bibitem[Estrada and
  {Rodr{\'{i}}guez-Vel{\'{a}}zquez}(2005{\natexlab{a}})]{estrada_rodriguez-velazquez05}
E.~Estrada and J.A. {Rodr{\'{i}}guez-Vel{\'{a}}zquez}.
\newblock Subgraph centrality in complex networks.
\newblock \emph{Physical Review E}, 71:\penalty0 056103, 2005{\natexlab{a}}.

\bibitem[Estrada and
  {Rodr{\'{i}}guez-Vel{\'{a}}zquez}(2005{\natexlab{b}})]{estrada_rodriguez-velazquez05b}
E.~Estrada and J.A. {Rodr{\'{i}}guez-Vel{\'{a}}zquez}.
\newblock Spectral measures of bipartivity in complex networks.
\newblock \emph{Physical Review E}, 72:\penalty0 046105, 2005{\natexlab{b}}.

\bibitem[Facebook(2017)]{facebook17}
Facebook.
\newblock {Press Release: Facebook Reports Second Quarter 2017 Results}.
\newblock \ \url{http://investor.fb.com/investor-news/default.aspx}\, July 26
  2017.
\newblock \
  \url{http://s21.q4cdn.com/399680738/files/doc_news/2017/FB-Q2'17-Earnings-Release.pdf}\
  (Retrieved on September 6, 2017).

\bibitem[Fisher and Ryan(1992)]{fisher_ryan92}
D.C. Fisher and J.~Ryan.
\newblock Bounds on the number of complete subgraphs.
\newblock \emph{Discrete Mathematics}, 103:\penalty0 313--320, 1992.

\bibitem[Gabaix(2009)]{gabaix09}
X.~Gabaix.
\newblock Power laws in economics and finance.
\newblock \emph{Annual Review of Economics}, 1:\penalty0 255--293, 2009.

\bibitem[Gabaix(2016)]{gabaix16}
X.~Gabaix.
\newblock Power laws in economics: {A}n introduction.
\newblock \emph{Journal of Economic Perspectives}, 30:\penalty0 185--206, 2016.

\bibitem[Gabaix et~al.(2003)Gabaix, Gopikrishnan, Plerou, and
  Stanley]{gabaix_etal03}
X.~Gabaix, P.~Gopikrishnan, V.~Plerou, and H.E. Stanley.
\newblock A theory of power-law distributions in financial market fluctuations.
\newblock \emph{Nature}, 423:\penalty0 267--270, 2003.

\bibitem[Goolsbee and Syverson(2008)]{goolsbee_syverson08}
A.~Goolsbee and C.~Syverson.
\newblock How do incumbents respond to the threat of entry? {E}vidence from the
  major airlines.
\newblock \emph{Quarterly Journal of Economics}, 123:\penalty0 1611--1633,
  2008.

\bibitem[Harary and Schwenk(1979)]{harary_schwenk79}
F.~Harary and A.J. Schwenk.
\newblock The spectral approach to determining the number of walks in a graph.
\newblock \emph{Pacific Journal of Mathematics}, 80:\penalty0 443--449, 1979.

\bibitem[Hayot and Jayaprakash(2005)]{hayot_jayaprakash05}
F.~Hayot and C.~Jayaprakash.
\newblock A feedforward loop motif in transcriptional regulation: {I}nduction
  and repression.
\newblock \emph{Journal of Theoretical Biology}, 234:\penalty0 133--143, 2005.

\bibitem[Ingram et~al.(2006)Ingram, Stumpf, and Stark]{ingram_etal06}
P.J. Ingram, M.P.H. Stumpf, and J.~Stark.
\newblock Network motifs: {S}tructure does not determine function.
\newblock \emph{BMC Genomics}, 7:\penalty0 108, 2006.

\bibitem[Isihara et~al.(2005)Isihara, Fujimoto, and Shibata]{isihara_etal05}
S.~Isihara, K.~Fujimoto, and T.~Shibata.
\newblock Cross talking of network motifs in gene regulation that generates
  temporal pulses and spatial stripes.
\newblock \emph{Genes to Cells}, 10:\penalty0 1025--1038, 2005.

\bibitem[Itzhack et~al.(2007)Itzhack, Mogilevski, and Louzoun]{itzhack_etal07}
R.~Itzhack, Y.~Mogilevski, and Y.~Louzoun.
\newblock An optimal algorithm for counting network motifs.
\newblock \emph{Physica A}, 381:\penalty0 482--490, 2007.

\bibitem[Itzkovitz and Alon(2005)]{itzkovitz_alon05}
S.~Itzkovitz and U.~Alon.
\newblock Subgraphs and network motifs in geometric networks.
\newblock \emph{Physical Review E}, 71:\penalty0 026117, 2005.

\bibitem[Itzkovitz et~al.(2003)Itzkovitz, Milo, Kashtan, Ziv, and
  Alon]{itzkovitz_etal03}
S.~Itzkovitz, R.~Milo, N.~Kashtan, G.~Ziv, and U.~Alon.
\newblock Subgraphs in random networks.
\newblock \emph{Physical Review E}, 68:\penalty0 026127, 2003.

\bibitem[Itzkovitz et~al.(2005)Itzkovitz, Levitt, Kashtan, Milo, Itzkovitz, and
  Alon]{itzkovitz_etal05}
S.~Itzkovitz, R.~Levitt, N.~Kashtan, R.~Milo, M.~Itzkovitz, and U.~Alon.
\newblock Coarse-graining and self-dissimilarity of complex networks.
\newblock \emph{Physical Review E}, 71:\penalty0 016127, 2005.

\bibitem[Jackson(2008)]{jackson08}
M.O. Jackson.
\newblock \emph{Social and Economic Networks}.
\newblock Princeton University Press, 2008.

\bibitem[Jackson(2011)]{jackson11}
M.O. Jackson.
\newblock An overview of social networks and economic applications.
\newblock In J.~Benhabib, A.~Bisin, and M.O. Jackson, editors, \emph{Handbook
  of Social Economics}. North Holland, 2011.

\bibitem[Jungnickel(2008)]{jungnickel08}
D.~Jungnickel.
\newblock \emph{Graphs, Networks and Algorithms}.
\newblock Springer, 3rd edition, 2008.

\bibitem[Kashtan et~al.(2004{\natexlab{a}})Kashtan, Itzkovitz, Milo, and
  Alon]{kashtan_etal04}
N.~Kashtan, S.~Itzkovitz, R.~Milo, and U.~Alon.
\newblock Topological generalizations of network motifs.
\newblock \emph{Physical Review E}, 70:\penalty0 031909, 2004{\natexlab{a}}.

\bibitem[Kashtan et~al.(2004{\natexlab{b}})Kashtan, Itzkovitz, Milo, and
  Alon]{kashtan_etal04b}
N.~Kashtan, S.~Itzkovitz, R.~Milo, and U.~Alon.
\newblock Efficient sampling algorithm for estimating subgraph concentrations
  and detecting network motifs.
\newblock \emph{Bioinformatics}, 20:\penalty0 1746--1758, 2004{\natexlab{b}}.

\bibitem[Khakabimamaghani et~al.(2013)Khakabimamaghani, Sharafuddin, Dichter,
  Koch, and {Masoudi-Nejad}]{khakabimamaghani_etal13}
S.~Khakabimamaghani, I.~Sharafuddin, N.~Dichter, I.~Koch, and
  A.~{Masoudi-Nejad}.
\newblock {QuateXelero}: {A}n accelerated exact network motif detection
  algorithm.
\newblock \emph{PloS ONE}, 8:\penalty0 e68073, 2013.

\bibitem[Lawford(2020)]{lawford20}
S.~Lawford.
\newblock Counting five-node subgraphs.
\newblock Technical Report arXiv:2009.11318, 2020.

\bibitem[Lawford and Mehmeti(2020)]{lawford_mehmeti20}
S.~Lawford and Y.~Mehmeti.
\newblock Cliques and a new measure of clustering: with application to {U.S.}
  domestic airlines.
\newblock \emph{Physica A}, 560:\penalty0 125158, 2020.

\bibitem[Leskovec and Krevl(2020)]{snapnets}
J.~Leskovec and A.~Krevl.
\newblock {SNAP Datasets}: {Stanford Large Network Dataset Collection}.
\newblock \ \url{http://snap.stanford.edu/data}\, 2020.

\bibitem[Lin and Ban(2014)]{lin_ban14}
J.~Lin and Y.~Ban.
\newblock The evolving network structure of {US} airline system during
  1990-2010.
\newblock \emph{Physica A}, 410:\penalty0 302--312, 2014.

\bibitem[Lordan et~al.(2014)Lordan, Sallan, and Simo]{lordan_etal14}
O.~Lordan, J.M. Sallan, and P.~Simo.
\newblock Study of the topology and robustness of airline route networks from
  the complex network approach: a survey and research agenda.
\newblock \emph{Journal of Transport Geography}, 37:\penalty0 112--120, 2014.

\bibitem[Mangan and Alon(2003)]{mangan_alon03}
S.~Mangan and U.~Alon.
\newblock Structure and function of the feed-forward loop network motif.
\newblock \emph{PNAS}, 100:\penalty0 11980--11985, 2003.

\bibitem[Milo et~al.(2002)Milo, {Shen-Orr}, Itzkovitz, Kashtan, Chklovskii, and
  Alon]{milo_etal02}
R.~Milo, S.~{Shen-Orr}, S.~Itzkovitz, N.~Kashtan, D.~Chklovskii, and U.~Alon.
\newblock Network motifs: {S}imple building blocks of complex networks.
\newblock \emph{Science}, 298:\penalty0 824--827, 2002.

\bibitem[Mrvar and Batagelj(2016)]{mrvar_batagelj16}
A.~Mrvar and V.~Batagelj.
\newblock Analysis and visualization of large networks with program package
  {P}ajek.
\newblock \emph{Complex Adaptive Systems Modeling}, 4:\penalty0 6, 2016.

\bibitem[Newman(2003)]{newman03}
M.E.J. Newman.
\newblock The structure and function of complex networks.
\newblock \emph{SIAM Review}, 45:\penalty0 167--256, 2003.

\bibitem[Newman(2009)]{newman09}
M.E.J. Newman.
\newblock Random graphs with clustering.
\newblock \emph{Physical Review Letters}, 103:\penalty0 058701, 2009.

\bibitem[{Ordu{\~n}a-Malea} et~al.(2015){Ordu{\~n}a-Malea}, Ayll{\'o}n,
  {Mart{\'i}n-Mart{\'i}n}, and {L{\'o}pez-C{\'o}zar}]{orduna-malea_etal15}
E.~{Ordu{\~n}a-Malea}, J.M. Ayll{\'o}n, A.~{Mart{\'i}n-Mart{\'i}n}, and E.D.
  {L{\'o}pez-C{\'o}zar}.
\newblock Methods for estimating the size of {G}oogle {S}cholar.
\newblock Technical Report arXiv:1506.03009v1, 2015.

\bibitem[Pinar et~al.(2017)Pinar, Seshadhri, and Vishal]{pinar_etal17}
A.~Pinar, C.~Seshadhri, and V.~Vishal.
\newblock {ESCAPE}: {E}fficiently counting all 5-vertex subgraphs.
\newblock In \emph{Proceedings of the International World Wide Web Conference
  (IW3C2)}, pages 1431--1440, 2017.

\bibitem[Prill et~al.(2005)Prill, Iglesias, and Levchenko]{prill_etal05}
R.J. Prill, P.A. Iglesias, and A.~Levchenko.
\newblock Dynamic properties of network motifs contribute to biological network
  organization.
\newblock \emph{PLoS Biology}, 3:\penalty0 1881--1892, 2005.

\bibitem[Roucolle et~al.(2020)Roucolle, Seregina, and Urdanoz]{roucolle_etal20}
C.~Roucolle, T.~Seregina, and M.~Urdanoz.
\newblock Measuring the development of airline networks: {C}omprehensive
  indicators.
\newblock \emph{Transportation Research Part A}, 133:\penalty0 303--324, 2020.

\bibitem[Ruci{\'{n}}ski(1988)]{rucinski88}
A.~Ruci{\'{n}}ski.
\newblock When are small subgraphs of a random graph normally distributed?
\newblock \emph{Probability Theory and Related Fields}, 78:\penalty0 1--10,
  1988.

\bibitem[Schlauch and Zweig(2015)]{schlauch_zweig15}
W.E. Schlauch and K.A. Zweig.
\newblock Influence of the null-model on motif detection.
\newblock In \emph{ASONAM 15 (Advances in Social Networks Analysis and
  Mining)}, 2015.

\bibitem[Schreiber and Schw{\"o}bbermeyer(2005)]{schreiber_schwobbermeyer05}
F.~Schreiber and H.~Schw{\"o}bbermeyer.
\newblock Frequency concepts and pattern detection for the analysis of motifs
  in networks.
\newblock In \emph{Transactions on Computational Systems Biology III}, pages
  89--104, 2005.

\bibitem[{Shen-Orr} et~al.(2002){Shen-Orr}, Milo, Mangan, and
  Alon]{shen-orr_etal02}
S.S. {Shen-Orr}, R.~Milo, S.~Mangan, and U.~Alon.
\newblock Network motifs in the transcriptional regulation network of
  \emph{Escherichia coli}.
\newblock \emph{Nature Genetics}, 31:\penalty0 64--68, 2002.

\bibitem[Song et~al.(2005)Song, Havlin, and Makse]{song_etal05}
C.~Song, S.~Havlin, and H.A. Makse.
\newblock Self-similarity of complex networks.
\newblock \emph{Nature}, 433:\penalty0 392--395, 2005.

\bibitem[Sporns and K{\"o}tter(2004)]{sporns_kotter04}
O.~Sporns and R.~K{\"o}tter.
\newblock Motifs in brain networks.
\newblock \emph{PLoS Biology}, 2:\penalty0 1910--1918, 2004.

\bibitem[Tran et~al.(2014)Tran, Mohan, Xu, and Huang]{tran_etal14}
N.T.L. Tran, S.~Mohan, Z.~Xu, and {C.-H.} Huang.
\newblock Current innovations and future challenges of network motif detection.
\newblock \emph{Briefings in Bioinformatics}, 16:\penalty0 497--525, 2014.

\bibitem[Twitter(2017)]{twitter17}
Twitter.
\newblock {Press Release: Twitter Announces Second Quarter 2017 Results}.
\newblock \ \url{http://investor.twitterinc.com/results.cfm}\, July 27 2017.
\newblock \
  \url{http://files.shareholder.com/downloads/AMDA-2F526X/5135079665x0x951003/11DEB964-E7A5-43F8-96E2-D074A947255B/TWTR_Q2_17_Earnings_Press_Release.pdf}\
  (Retrieved on September 6, 2017).

\bibitem[Vassilevska(2009)]{vassilevska09}
V.~Vassilevska.
\newblock Efficient algorithms for clique problems.
\newblock \emph{Information Processing Letters}, 109:\penalty0 254--257, 2009.

\bibitem[{Vassilevska Williams}(2014)]{vassilevskawilliams14}
V.~{Vassilevska Williams}.
\newblock Multiplying matrices in ${O}(n^{2.373})$ time.
\newblock Mimeo, Available at:
  \url{http://people.csail.mit.edu/virgi/matrixmult-f.pdf}\, 2014.

\bibitem[{Vassilevska Williams} et~al.(2015){Vassilevska Williams}, Wang,
  Williams, and Yu]{vassilevskawilliams_etal15}
V.~{Vassilevska Williams}, J.R. Wang, R.~Williams, and H.~Yu.
\newblock Finding four-node subgraphs in triangle time.
\newblock In \emph{SODA 15 (ACM-SIAM Symposium on Discrete Algorithms)}, 2015.

\bibitem[Verma et~al.(2014)Verma, Ara{\'u}jo, and Herrmann]{verma_etal14}
T.~Verma, N.A.M. Ara{\'u}jo, and H.J. Herrmann.
\newblock Revealing the structure of the world airline network.
\newblock \emph{Scientific Reports}, 4:\penalty0 5638, 2014.

\bibitem[Verma et~al.(2016)Verma, Russmann, Ara{\'u}jo, Nagler, and
  Herrmann]{verma_etal16}
T.~Verma, F.~Russmann, N.A.M. Ara{\'u}jo, J.~Nagler, and H.J. Herrmann.
\newblock Emergence of core-peripheries in networks.
\newblock \emph{Nature Communications}, 7:\penalty0 10441, 2016.

\bibitem[Watts(1999)]{watts99}
D.J. Watts.
\newblock Networks, dynamics, and the small-world phenomenon.
\newblock \emph{American Journal of Sociology}, 105:\penalty0 493--527, 1999.

\bibitem[Watts and Strogatz(1998)]{watts_strogatz98}
D.J. Watts and S.H. Strogatz.
\newblock Collective dynamics of `small-world' networks.
\newblock \emph{Nature}, 393:\penalty0 440--442, 1998.

\bibitem[Wong et~al.(2011)Wong, Baur, Quader, and Huang]{wong_etal11}
E.~Wong, B.~Baur, S.~Quader, and {C.-H.} Huang.
\newblock Biological network motif detection: {P}rinciples and practice.
\newblock \emph{Briefings in Bioinformatics}, 13:\penalty0 202--215, 2011.

\bibitem[Wu et~al.(2013)Wu, Qian, Zhang, Yang, Liu, Dong, Zhang, Zhu, and
  Feng]{wu_etal13}
S.F. Wu, W.Y. Qian, J.W. Zhang, Y.B. Yang, Y.~Liu, Y.~Dong, Z.B. Zhang, Y.P.
  Zhu, and Y.J. Feng.
\newblock Network motifs in the transcriptional regulation network of cervical
  carcinoma cells respond to {EGF}.
\newblock \emph{Archives of Gynecology and Obstetrics}, 287:\penalty0 771--777,
  2013.

\bibitem[Wuchty and Stadler(2003)]{wuchty_stadler03}
S.~Wuchty and P.F. Stadler.
\newblock Centers of complex networks.
\newblock \emph{Journal of Theoretical Biology}, 223:\penalty0 45--53, 2003.

\bibitem[Wuellner et~al.(2010)Wuellner, Roy, and {D'Souza}]{wuellner_etal10}
D.R. Wuellner, S.~Roy, and R.M. {D'Souza}.
\newblock Resilience and rewiring of the passenger airline networks in the
  {United States}.
\newblock \emph{Physical Review E}, 82:\penalty0 056101, 2010.

\bibitem[{Yeger-Lotem} et~al.(2004){Yeger-Lotem}, Sattath, Kashtan, Itzkovitz,
  Milo, Pinter, and Alon]{yeger-lotem_etal04}
E.~{Yeger-Lotem}, S.~Sattath, N.~Kashtan, S.~Itzkovitz, R.~Milo, R.Y. Pinter,
  and U.~Alon.
\newblock Network motifs in integrated cellular networks of
  transcription--regulation and protein--protein interaction.
\newblock \emph{PNAS}, 101:\penalty0 5934--5939, 2004.

\end{thebibliography}

\end{document}